\documentclass[cmp]{svjour}%
\usepackage{amssymb}
\usepackage{amsfonts}
\usepackage{amsmath}%
\usepackage{graphicx}
\usepackage{hyperref}
\usepackage{enumerate}

\begin{document}

\journalname{Commun. Math. Phys.}
\title{Acausal quantum theory for non-Archimedean scalar fields}
\author{M. L. Mendoza-Mart\'{\i}nez\inst{1}
\and J. A. Vallejo\inst{2}
\and W. A. Z\'{u}\~{n}iga-Galindo\inst{1}
\thanks{\emph{emails}: \texttt{mmendoza@math.cinvestav.mx}, \texttt{wazuniga@math.cinvestav.edu.mx},
\texttt{jvallejo@ fc.uaslp.mx}}}
\institute{Centro de Investigaci\'{o}n y de Estudios Avanzados del Instituto
Polit\'{e}cnico Nacional, Departamento de Matem\'{a}ticas, Unidad
Quer\'{e}taro, Libramiento Norponiente \#2000, Fracc. Real de Juriquilla.
Santiago de Quer\'{e}taro, Qro. 76230, M\'{e}xico.
\and
Facultad de Ciencias, Universidad Aut\'onoma de San Luis Potos\'i, Av. Salvador
Nava s/n, San Luis Potos\'{\i}, SLP 78290, M\'{e}xico.
}

\maketitle

\begin{abstract}

We construct a family of quantum scalar fields over a $p-$adic
spacetime which satisfy $p-$adic analogues of the G\aa rding--Wightman
axioms. Most of the axioms can be formulated the same way in both, the
Archimedean and non-Archimedean frameworks; however, the axioms depending on
the ordering of the background field must be reformulated, reflecting the 
acausality of $p-$adic spacetime. The $p-$adic scalar
fields satisfy certain $p-$adic Klein-Gordon pseudo-differential equations. The
second quantization of the solutions of these Klein-Gordon equations
corresponds exactly to the scalar fields introduced here.

\end{abstract}

\tableofcontents

\section{Introduction}

Ever since the advent of Quantum Mechanics, the question of its compatibility with Special Relativity 
was raised. The occurrence of non-locality in the quantum world and its implications regarding the
relativistic causal structure was the central theme in the well-known works by Einstein, Podolsky and
Rosen, and Bell. These issues are still debated today, but there is a increasing amount of research
pointing towards the fact that quantum mechanics is incompatible at a fundamental level not only
with the causal structure furnished by Special Relativity (through light cones), but with any other
possible \emph{causal} ordering\footnote{Notice that we emphasize the causal character. There are other
possible orderings (chronological, horismos) that will be not considered here, although they are related,
see \cite{KP67}.}. In \cite{Cav18}, it is concluded that the description of non-localities
requires fine-tuning of the system's parameters, thus violating a basic principle of any causal model.
In \cite{Ore12}, quantum correlations
incompatible with a definite causal order are constructed (although they prove that a causal order
emerges in the classical limit), and the experimental existence of these correlations is reported in
\cite{Rub17}. See also \cite{Rin16} for the incompatibility of
Quantum Mechanics with some \emph{non-local} causal models. Applications of the absence of a predefined
causal structure to quantum computations are given in \cite{Chi13}.

Motivated by these considerations, one could wonder whether it is possible to construct a quantum
field theory (QFT) on a spacetime devoid of any \emph{a priori} causal structure. The notions of spacelike
and timelike intervals which, from an operational point of view, characterize the causal structure, are 
intimately tied to the existence of a total order on the field number $\mathbb{R}$ compatible with the 
algebraic field operations, so a possibility is to start from a non-ordered number field. Leaving aside 
the case of finite fields, the most obvious choice is to consider the non-Archimedean
field of $p-$adic numbers
$\mathbb{Q}_p$. The corresponding spacetime would be $\mathbb{Q}^4_p$. In this way,  ($p-$adic) time no 
longer acts
as an ordering parameter. While this is completely consistent with the requirement of covariance, it
raises some questions about its meaning in Quantum Mechanics; for some theoretical points of view about
the possibility of quantum processes without a time parameter see \cite{Woo84,Rov90}.

The spacetime $\mathbb{Q}^4_p$ is acausal in the broad sense of lacking a causal structure, but also in the particular, technical, sense that for any pair of points on it, there exists no causal curve connecting them (which, in particular, also implies that it is achronal).
The question of the intrinsic (a)causality of
spacetime has been studied sometime ago \cite{Poi92}, and is a topic of obligated discussion
when dealing with the possibility of `travels in time' \cite{Luk03,Tip17}. Acausal (portions of)
spacetimes appears often in relation with wormholes in General Relativity \cite{Mor88}. 
There have been problems
in constructing the $S$ matrix for interacting massive scalar fields in this setting 
\cite{Fri92}, 
but it should be stressed that these are due to the
interaction along closed timelike curves, which do not exist at all in the framework of a globally
acausal spacetime such as the one presented here, where the very notion of `timelike' does not make
sense. 

A problem present in any acausal theory is the characterization of microcausality or local commutativity, 
that is, the vanishing of the commutator of field operator-valued distributions when the test functions
have support in spacelike separated regions. It is not clear \emph{a priori} that a theory without a
causal structure will allow for vanishing commutators even restricting the domain of the involved
operators, but we will show 
below that a similar property holds when the test functions are supported in the $p-$adic unit ball.
Thus, there is no room for phenomena arising in the non-Archimedean case, such as
the connection of spacelike regions by large timelike loops.
It is also reasonable to expect that the consideration of $p-$adics numbers could also cure the divergences in $1-$loop effective Lagrangians that appear in the real Euclidean case 
\cite{Cas97}, although no attempt is made here to pursue this direction of research. 

Another, different, kind of motivation for studying quantum field theory in the $p-$adic setting
comes from the conjecture of Vladimirov and Volovich stating that spacetime
has a non-Archimedean nature at the Planck scale, \cite{Vol}, see also
\cite{Var1}. The existence of the Planck scale implies that below it the very
notion of measurement as well as the idea of `infinitesimal length' become
meaningless, and this fact translates into the mathematical statement that the
Archimedean axiom is no longer valid. Before Volovich, some authors explored
the possibility of constructing theories of the spacetime using background
fields different from \ $\mathbb{R}$ and $\mathbb{C}$; for instance, in
\cite{EU66} Everett and Ulam study the Lorentz group over $\mathbb{Q}_{p}$ in
the hope that `spaces of this sort might be useful in some future models of
nuclear or subnuclear theories', see also \cite{Var1}, \cite[Chapter 6]{Var2}
and references therein. Volovich's conjecture propelled a wide variety of
investigations in cosmology, quantum mechanics, string theory, QTF, etc., and
the influence of this conjecture is still relevant nowadays, see e.g.
\cite{Abd et al}, \cite{B-F}-\cite{DD97}, \cite{D-K-K-V}, \cite{Dra01},
\cite{Gubser et al.}-\cite{Harlow et al}, \cite{Koch-Sait}-\cite{Mis2},
\cite{Var et al}-\cite{Vol}, \cite{Zuniga-FAA-2017}, \cite{Zuniga-LNM-2016}.
In a completely different framework, that of the physics of complex systems, the paradigm
asserting that the space of states of several complex systems has an
ultrametric structure has also originated a large amount of research, see
\cite{RTV86}, \cite{KKZuniga}\ and references therein. These two ideas are
the main motivations driving the development of $p-$adic mathematical physics.
In particular, during the last thirty years $p-$adic QFT has been studied
intensively, a topic whose importance has been highlighted by Varadarajan 
in \cite{Var2}.

In this article we present a second-quantization, based on Segal's
formalism, for $p-$adic free scalar fields whose evolution is described by
a certain class of Klein-Gordon type pseudo-differential operators. 
In order to guarantee that the resulting theory has some physical content, we show
that the corresponding quantum non-Archimedean scalar fields satisfy
$p-$adic versions of G\aa rding--Wightman's axioms. Most of them can be
formulated in a way valid in both the Archimedean and non-Archimedean cases,
but some of them must be appropriately re-formulated in the $p-$adic
setting by introducing new mathematical ideas and re-interpreting some
classical constructions that are not directly available in the $p-$adic
context. For instance, the absence of an ordering in the background number
field implies some profound modifications in the usual interpretation of
notions such as the timelike or spacelike character of $p-$adic spacetime
events, and the introduction of new mathematical objects such as the
$p-$adic restricted Lorentz group, that we will discuss below. As another
example, our $p-$adic spectral condition does not provide a definition of
energy and momentum operators, because this would require a theory of
semigroups, with $p-$adic time, for operators acting on complex-valued
functions, and such a theory does not exist at the moment. However, the outcomes
of our analysis are consistent with the requirement that the
mathematical description of physical reality must not depend on the background
number field, see \cite{Volovich1}. This property is due to the particular nature
of the Klein-Gordon field, notice that the same is not true for the Schr\"odinger
equation, as the number $i$ does not have an analog in an arbitrary field.

Thus, the main conclusion is that there seems to be
no obstruction to the existence of a mathematically rigorous quantum field
theory (QFT) for free fields in the $p-$adic framework, based on an acausal
spacetime. It must be remarked
that we deal with free fields, omitting interactions. The reason for this is
that, due to Haag's theorem, interactions require a more technical treatment,
but having a consistent theory for the free case is the first step towards a
complete $p-$adic QFT.

We have remarked some features derived from the fact that the spacetime is
$p-$adic. Let us now make some comment about those originated in the
configuration space of the fields. A key fact is that we work with
complex-valued fields. This allow us to use the tools from classical
functional analysis, in particular Segal quantization. On the other hand, it
is also possible to work with $p-$adic valued fields. In this setting,
Khrennikov developed a theory of Gaussian integration of
non-Archimedean-valued functions on infinite-dimensional non-Archimedean
spaces and a calculus of pseudo-differential operators which is suitable for
the second-quantization representation in non-Archimedean quantum field
theory, see \cite{Kh1}-\cite{Khrennikov1} and references therein.
Mathematically speaking, this is a completely different setting from ours: for
instance, $p-$adic Hilbert spaces are radically different to their complex counterparts.

The construction of a quantum field theory over a $p-$adic spacetime raises
the question about the physical meaning of the prime $p$. Once a choice for
$p$ is made, we can construct $\mathbb{Q}_{p}^{4}$ (endowed with the maximum
norm) and then give it a geometric structure through a quadratic form
$\mathfrak{q}$. The geometry of the resulting spacetime, the quadratic space
$(\mathbb{Q}_{p}^{4},\mathfrak{q})$, depends crucially on both, $p$ and
$\mathfrak{q}$. We choose the simplest case in which the quadratic form is the
unique elliptic form of dimension four and a prime number $p\equiv
1\,\mathrm{\operatorname{mod}}\,4$. The first choice is motivated by the need
for ellipticity when doing the explicit computation of the fundamental
solutions (and the corresponding propagators) of the Klein-Gordon equation.
Notice that the naive choice $\mathfrak{q}(k)=k_{0}^{2}-(k_{1}^{2}+k_{2}%
^{2}+k_{3}^{2})$ is excluded because it is not elliptic. It is possible to
develop a theory based on this form, but at the cost of facing greater
technical difficulties. However, as we will see, our choice for $\mathfrak{q}$
retains all the essential features of a relativistic theory, so it is
justifiable from a physical point of view. Regarding the choice of $p$, the
quantum fields introduced here will strongly depend on the geometry of the
hypersurface $V=\left\{  k\in\mathbb{Q}_{p}^{4};\mathfrak{q}(k)=1\right\}  $,
and if we pick $p\equiv1\,\mathrm{\operatorname{mod}}\,4$, then we can
guarantee that $\sqrt{\omega(\mathbf{k})}\neq0$ for any $\mathbf{k}\in U_{\mathfrak{q}}$, where
$U_{\mathfrak{q}}\subset\mathbb{Q}_{p}^{3}$ is a certain open and compact subset (depending on
$\mathfrak{q}$) that will be defined later on. Notice that, due to these
choices, we are actually defining a family of quantizations, a fact that
could be viewed as an advantage over the rigidity of the classical case.

Thus, given a prime number $p\equiv1$ $\operatorname{mod}4$ and a $p-$adic
elliptic quadratic form $\mathfrak{q}$ of dimension $4$, we will denote by
$\boldsymbol{O}(\mathfrak{q})$ the orthogonal group of $\mathfrak{q}$. As
stated, the $p-$adic Minkowski spacetime is, by definition, the quadratic
space $(\mathbb{Q}_{p}^{4},\mathfrak{q})$, so the Lorentz group of spacetime
is $\boldsymbol{O}(\mathfrak{q})$. In this article, `time' is a $p-$adic
variable, so the notions of past and future are not clearly defined. However,
the $p-$adic implicit function theorem allows us to determine $k_{0} $, from
$\mathfrak{q}\left(  k_{0},\mathbf{k}\right)  =1$, as $k_{0}=\pm\sqrt{\omega\left(
\mathbf{k}\right)  }$, where $\sqrt{\omega\left(  \mathbf{k}\right)  }$ is a $p-$adic analytic
function \ defined in $U_{\mathfrak{q}}$, and in this way we can define the
mass shells:
\[
V^{\pm}=\left\{  \left(  k_{0},\mathbf{k}\right)  \in\mathbb{Q}_{p}\times\mathbb{Q}%
_{p}^{3};k_{0}=\pm\sqrt{\omega\left(  \mathbf{k}\right)  }\mbox{ , }\mathbf{k}\in U_{\mathfrak{q}%
}\right\}  \,.
\]

In the $p-$adic setting the usual geometric notion of cone does not make
sense, because it depends on the fact that the real numbers form an ordered
field. Therefore, the notion of closed forward light cone is replace by the
notion of `closed forward semigroup', which is the topological closure of the
additive semigroup generated by $V^{+}$. This notion allow us to construct a
spectral measure attached to a strongly continuous unitary representation \ of
the $p-$adic Poincar\'{e} group as in the classical case, see Theorem
\ref{Theorem2}.

We will denote by $\mathcal{F}$ the Fourier transform operator associated to
the quadra\-tic form $\mathfrak{q}$. The $p-$adic Klein-Gordon operator
attached to $\mathfrak{q}$ with unit mass is defined as%
\[
\square_{\mathfrak{q},\alpha}\varphi=\mathcal{F}^{-1}\left(  \left\vert
\mathfrak{q}-1\right\vert _{p}^{\alpha}\mathcal{F}\varphi\right)  \,
\]
where $\varphi$ is a test function and $\alpha$ is a fixed positive number.

In conventional QFT there have been some studies devoted to the optimal choice of
the space of test functions. In \cite{Jaffe}, Jaffe discussed this topic (see
also \cite{Strocchi} and \cite{Lopu}); his conclusion was that, rather than an
optimal choice, there exists a set of conditions that must be satisfied by the
candidate space, and any class of test functions with these properties should
be considered as valid. The main condition is that the space of test functions
must be a nuclear countable Hilbert one. In this article, we use the following
Gel'fand triple: $\mathcal{H}_{\infty}\left(  \mathbb{K}\right)  \subset
L_{\mathbb{K}}^{2}\subset\mathcal{H}_{\infty}^{\ast}\left(  \mathbb{K}\right)
$, where $\mathbb{K}=\mathbb{R}$, $\mathbb{C}$. This triple was introduced in
\cite{Zuniga-FAA-2017}. The space $\mathcal{H}_{\infty}\left(  \mathbb{K}%
\right)  $ is a nuclear countable Hilbert space, which is invariant under the
action of a large class of pseudo-differential operators. This space can be considered the
`true' non-Archimedean analogue of the classical Schwartz space, as we will repeatedly
justify in what follows. In fact, our results
could be summarized by saying that the G\aa rding--Wightman axioms make sense in the
$p-$adic context if we replace the Schwartz space of the classical framework
by $\mathcal{H}_{\infty}\left(  \mathbb{C}\right)  $.

The $p-$adic Klein-Gordon equation
\begin{equation}
\square_{\mathfrak{q},\alpha}u\left( t,\mathbf{x}\right)  =0\label{eq_klein_gordon}%
\end{equation}
admits solutions of plane wave type, more precisely, the functions
\[
\exp2\pi i\left\{  tE^{\pm}-sx_{1}l_{1}-px_{2}l_{2}+spx_{3}l_{3}\right\}_{p},
\]
where $\left\{  \cdot\right\}  _{p}$ denotes the $p-$adic fractional part,
$\mathbf{l}=\left(  l_{1},l_{2},l_{3}\right)
\in\mathbb{Q}_{p}^{3}$ is a fixed vector, and $E^{\pm}=\pm\sqrt{\omega\left(
\mathbf{l}\right)  }$ (here $\sqrt{\omega\left(  \mathbf{k}\right)  }$ is
the $p-$adic dispersion) are weak solutions \ of (\ref{eq_klein_gordon}), see
Theorem \ref{Theorem3}. The general solution of (\ref{eq_klein_gordon}), up to
multiplication by a non-zero complex constant, is
\begin{equation}\label{eq_sol_klein_gordon}
 \int\limits_{U_{\mathfrak{q}}}\left( \chi_{p}\left( -\sqrt{\omega (\mathbf{k})}t
 +\mathbf{k}\cdot \mathbf{x}\right) 
 a\left(\mathbf{k}\right)  +\chi_{p}\left(  \sqrt
{\omega (\mathbf{k})}t-\mathbf{k}\cdot \mathbf{x}\right)  a^{\dagger} (-\mathbf{k})
\right) \frac{d^{3}\mathbf{k}}{\left\vert \sqrt{\omega\left( \mathbf{k}\right)  }\right\vert
_{p}}\,,
\end{equation}
where $\chi_{p}\left(  \cdot\right)  =\exp\left(  2\pi i\left\{
\cdot\right\}  _{p}\right)  $ is the standard additive character of
$\mathbb{Q}_{p}$, $U_{\mathfrak{q}}\subset\mathbb{Q}_{p}^{3}$ is an open and
compact subset, $\mathbf{k}\cdot \mathbf{x}$ denotes a suitable bilinear form, and $a\left(
\mathbf{k}\right)  $, $a^{\dagger}\left(  -\mathbf{k}\right)  $ are test functions, see Theorem
\ref{Theorem3}. The solutions (\ref{eq_sol_klein_gordon}) can be quantized
using the techniques described below, and the corresponding Klein-Gordon
fields satisfy the corresponding Wightman axioms, see Theorem \ref{Theorem2}.

The $p-$adic Klein-Gordon equations in the form used in this article were
introduced by the third author, see \cite[Chapter 6]{Zuniga-LNM-2016}\ and
references therein, where also the problem of the second quantization of their
solutions was posed \cite[Chapter 7]{Zuniga-LNM-2016}. The resulting field theory
has a strong number-theoretic flavor. For instance, the
calculation of the Green functions is related to the meromorphic continuation
of Igusa's local zeta functions, see Theorem \ref{Theorem1} and the references
\cite{Igusa}, \cite[Chapter 10]{KKZuniga}, \cite[Chapter 5]{Zuniga-LNM-2016}.

Finally, let us remark that there are a lot of open questions related to
$p-$adic quantum fields and their underlying mathematical techniques that
remain to be studied within the present framework. Among them, probably the
most important one is the reconstruction theorem, which depends on an
appropriate definition of Wightman distributions, and, of course, the
inclusion of non-trivial interactions, that will be discussed elsewhere.
The corresponding theory for non-elliptic quadratic forms $\mathfrak{q}$, though
much more difficult, is also of interest.

\section{\label{Section_1}Preliminaries}

\label{se2} Along this article $p$ will denote a prime number different from
$2$. Due to physical considerations we will formulate all our results in
dimension $4$, however, many of our results are still valid in arbitrary dimension.

\subsection{The field of $p-$adic numbers}

In this section we summarize the essential aspects and basic results on
$p-$adic analysis that we will use through the article. For a detailed
exposition of $p-$adic analysis the reader may consult \cite{A-K-S,Taibleson,V-V-Z}.

The field of $p-$adic numbers $\mathbb{Q}_{p}$ is defined as the completion of
the field of rational numbers $\mathbb{Q}$ with respect to the $p-$adic norm
$|\cdot|_{p}$, which in turn is defined as
\[
|x|_{p}=%
\begin{cases}
0 & \text{if }x=0\\
p^{-\gamma} & \text{if }x=p^{\gamma}\dfrac{a}{b}\,,
\end{cases}
\]
where $a$ and $b$ are integers coprime with $p$. The integer $\gamma:=ord(x)
$, with $ord(0):=+\infty$, is called the \emph{$p-$adic order} of
$x$. Any $p-$adic number $x\neq0$ has a unique expansion of the form
\begin{equation}
x=p^{ord(x)}\sum_{j=0}^{\infty}x_{j}p^{j}\,,\label{eq_expansion}%
\end{equation}
where $x_{j}\in\{0,\dots,p-1\}$ and $x_{0}\neq0$. Any non-zero $p-$adic number
$x$\ can be written uniquely as $x=p^{ord(x)}ac\left(  x\right)  $, with
$\left\vert ac\left(  x\right)  \right\vert _{p}=1$, $ac\left(  x\right)  $ is
called \textit{the angular component} of $x$.

By using expansion (\ref{eq_expansion}), we define \textit{the fractional part
of }$x\in\mathbb{Q}_{p}$, denoted $\{x\}_{p}$, as the rational number
\[
\{x\}_{p}=%
\begin{cases}
0 & \text{if }x=0\text{ or }ord(x)\geq0\\
p^{ord(x)}\sum_{j=0}^{-ord(x)-1}x_{j}p^{j} & \text{if }ord(x)<0\,.
\end{cases}
\]
As a topological space $\mathbb{Q}_{p}$\ is homeomorphic to a Cantor-like
subset of the real line, see e.g. \cite{A-K-S,V-V-Z}. The balls and
spheres are compact subsets.

We extend the $p-$adic norm to $\mathbb{Q}_{p}^{4}$ by taking%
\[
||x||_{p}:=\max_{0\leq i\leq3}|x_{i}|_{p},\qquad\text{for }x=(x_{0}%
,x_{1},x_{2},x_{3})\in\mathbb{Q}_{p}^{4}.
\]
We define $ord(x)=\min_{0\leq i\leq3}\{ord(x_{i})\}$, then $||x||_{p}%
=p^{-ord(x)}$. The metric space $\left(  \mathbb{Q}_{p}^{4},||\cdot
||_{p}\right)  $ is a complete ultrametric space. Thus $(\mathbb{Q}_{p}%
^{4},\Vert\cdot\Vert_{p})$ is a locally compact topological space.

For $l\in\mathbb{Z}$, denote by $B_{l}^{4}(a)=\{x\in\mathbb{Q}_{p}%
^{4}:||x-a||_{p}\leq p^{l}\}$ \textit{the ball of radius }$p^{l}$ \textit{with
center at} $a=(a_{0},a_{1},a_{2},a_{3})\in\mathbb{Q}_{p}^{4}$, and take
$B_{l}^{4}(0):=B_{l}^{4}$. Note that $B_{l}^{4}(a)=B_{l}(a_{0})\times
\cdots\times B_{l}(a_{3})$, where $B_{l}(a_{i}):=\{x\in\mathbb{Q}_{p}%
:|x-a_{i}|_{p}\leq p^{l}\}$ is the one-dimensional ball of radius $p^{l}$ with
center at $a_{i}\in\mathbb{Q}_{p}$. The ball $B_{0}^{4}$ equals the product of
four copies of $B_{0}:=\mathbb{Z}_{p}$, \textit{the ring of }$p-$\textit{adic
integers}. For $l\in\mathbb{Z}$, denote by $S_{l}^{4}(a)=\{x\in\mathbb{Q}%
_{p}^{4}:||x-a||_{p}=p^{l}\}$ \textit{the sphere of radius }$p^{l}$
\textit{with center at} $a\in\mathbb{Q}_{p}^{4}$, and take $S_{l}%
^{4}(0):=S_{l}^{4}$.

\begin{remark}
The natural map $\mathbb{Z}_{p}\rightarrow\mathbb{Z}_{p}/p\mathbb{Z}_{p}%
\simeq\mathbb{F}_{p}$, where $\mathbb{F}_{p}$ is the finite field with $p $
elements, is called the reduction modulo $p$, denoted as $\overline{\cdot}$.
We will identify $\mathbb{F}_{p}=\left\{  \overline{0},\overline{1}%
,\ldots,\overline{p-1}\right\}  $, where the addition and multiplication are
defined modulo $p$. We will distinguish between $\left\{  0,1,\ldots
,p-1\right\}  \subset\mathbb{Z}_{p}$ and $\mathbb{F}_{p}$. Later on, we will
also use the symbol `$\overline{\cdot}$' to mean conjugation of complex
numbers, but it will clear from the context which case it is being used.
\end{remark}

\begin{note} Let us collect here some conventions. 
\begin{enumerate}[(i)]
\item We denote by $\Omega(\left\Vert x\right\Vert _{p})$ the characteristic
function of $B_{0}^{4}$. For more general sets, say Borel sets, we use
${\LARGE 1}_{A}\left(  x\right)  $ to denote the characteristic function of
$A$.

\item From now on, we denote by $d^{4}x$ the Haar measure of the locally
compact group $\left(  \mathbb{Q}_{p}^{4},+\right)  $ normalized so that the
volume of $\ \mathbb{Z}_{p}^{4}$ is one.

\item We will use \ the notation $x=\left(  x_{0},x_{1},x_{2},x_{3}\right)
=\left(  x_{0},\boldsymbol{x}\right)  \in\mathbb{Q}_{p}\times\mathbb{Q}%
_{p}^{3}$ from now up to Section \ref{Section 5.5}.
\end{enumerate}
\end{note}

\subsection{Some function spaces}

\subsubsection{The Bruhat-Schwartz space}

We take $\mathbb{K}$ to mean $\mathbb{R}$ or $\mathbb{C}$. A $\mathbb{K}%
$-valued function $\varphi$ defined on $\mathbb{Q}_{p}^{4}$ is \textit{called
locally constant,} if for any $x\in\mathbb{Q}_{p}^{4}$ there exists an integer
$l(x)\in\mathbb{Z}$ such that%
\begin{equation}
\varphi(x+x^{\prime})=\varphi(x)\text{ for }x^{\prime}\in B_{l(x)}%
^{4}.\label{local_constancy}%
\end{equation}
A function $\varphi:\mathbb{Q}_{p}^{4}\rightarrow\mathbb{K}$ is called a
\textit{Bruhat-Schwartz function (or a test function),} if it is locally
constant with compact support. The $\mathbb{K}$-vector space of
Bruhat-Schwartz functions is denoted by $\mathcal{D}_{\mathbb{K}}%
(\mathbb{Q}_{p}^{4}):=\mathcal{D}_{\mathbb{K}}$. Let $\mathcal{D}_{\mathbb{K}%
}^{^{\prime}}(\mathbb{Q}_{p}^{4}):=\mathcal{D}_{\mathbb{K}}^{^{\prime}}$
denote the space of all continuous functionals (distributions) on $\mathcal{D}%
_{\mathbb{K}}$. The space $\mathcal{D}_{\mathbb{K}}^{^{\prime}}$ coincides
with the algebraic dual of $\mathcal{D}_{\mathbb{K}}$, i.e. any linear
functional on $\mathcal{D}_{\mathbb{K}}$ is continuous. For an in-depth
discussion the reader may consult \cite{A-K-S}, \cite{Taibleson}, \cite{V-V-Z}.

\begin{remark}
Most of the time we will work in dimension four, with spaces like
$\mathcal{D}_{\mathbb{K}}(\mathbb{Q}_{p}^{4})$ and $\mathcal{D}_{\mathbb{K}%
}^{^{\prime}}(\mathbb{Q}_{p}^{4})$, in these cases \ we will use the
abbreviated notation $\mathcal{D}_{\mathbb{K}}$, $\mathcal{D}_{\mathbb{K}%
}^{^{\prime}}$. In a few occasions we will work in dimensions different from
$4$, then we will use the notation $\mathcal{D}_{\mathbb{K}%
}(\mathbb{Q}_{p}^{n})$, $\mathcal{D}_{\mathbb{K}}^{\prime}(\mathbb{Q}_{p}%
^{n})$. A similar rule will be used for other function spaces.
\end{remark}

\subsubsection{The spaces $L^{r}$}

Given $r\in\left[  1,+\infty\right)  $, we denote by $L_{\mathbb{K}}%
^{r}\left(  \mathbb{Q}_{p}^{4},d^{4}x\right)  \allowbreak:=L_{\mathbb{K}}^{r}
$, the $\mathbb{K}$-vector space of all the $\mathbb{K}$-valued functions $g$
satisfying $\int_{\mathbb{Q}_{p}^{4}}\left\vert g\left(  x\right)  \right\vert
^{r}d^{4}x\allowbreak<\infty$.

\subsection{\label{Sect_Fourier_Trans}Fourier transform}

Set $\chi_{p}(y)=\exp(2\pi i\{y\}_{p})$ for $y\in\mathbb{Q}_{p}$. The map
$\chi_{p}(\cdot)$ is an additive character on $\mathbb{Q}_{p}$, i.e. a
continuous map from $\mathbb{Q}_{p}$ into the unit circle satisfying $\chi
_{p}(y_{0}+y_{1})=\chi_{p}(y_{0})\chi_{p}(y_{1})$, $y_{0},y_{1}\in
\mathbb{Q}_{p}$.

We set
\[
\mathfrak{B}\left(  x,y\right)  =x_{0}y_{0}-sx_{1}y_{1}-px_{2}y_{2}%
+spx_{3}y_{3},
\]
where $s\in\mathbb{Z}$ is a quadratic non-residue module $p$, i.e. the
congruence $x^{2}\equiv s$ $\operatorname{mod}$ $p$ does not have solution.
Then \ $\mathfrak{B}\left(  x,y\right)  $ is a symmetric non-degenerate
$\mathbb{Q}_{p}-$bilinear form on $\mathbb{Q}_{p}^{4}\times\mathbb{Q}_{p}^{4}%
$, and
\[
\mathfrak{q}(x):=\mathfrak{B}\left(  x,x\right)  =x_{0}^{2}-sx_{1}^{2}%
-px_{2}^{2}+spx_{3}^{2}\text{, }x\in\mathbb{Q}_{p}^{4}%
\]
is a \textit{non-degenerate quadratic form} on $\mathbb{Q}_{p}^{4}$.\ In
addition, $\mathfrak{q}(x)$ is the unique (up to linear equivalence)
\textit{elliptic quadratic form} in dimension four, here elliptic means that
$\mathfrak{q}(x)=0\Leftrightarrow x=0$ (notice that this is not equivalent
to the non-degeneracy of $\mathfrak{B}$, as the equation $\mathfrak{q}(x)=0$
could have its own solutions, not coming from vectors orthogonal to all the
vectors in $\mathbb{Q}_{p}^{4}$).

We identify the $\mathbb{Q}_{p}-$vector space $\mathbb{Q}_{p}^{4}$ with its
algebraic dual $\left(  \mathbb{Q}_{p}^{4}\right)  ^{\ast}$ by means of
$\mathfrak{B}\left(  \cdot,\cdot\right)  $. We now identify the dual group
(i.e. the Pontryagin dual) of $\left(  \mathbb{Q}_{p}^{4},+\right)  $ with
$\left(  \mathbb{Q}_{p}^{4}\right)  ^{\ast}$ by taking $x^{\ast}\left(
x\right)  =\chi_{p}\left(  \mathfrak{B}\left(  x,x^{\ast}\right)  \right)  $.
The Fourier transform \ is defined by%
\[
(\mathcal{F}g)(k)=%
{\displaystyle\int\limits_{\mathbb{Q}_{p}^{4}}}
g\left(  x\right)  \chi_{p}\left(  \mathfrak{B}\left(  x,k\right)  \right)
d\mu\left(  x\right)  \text{,}\quad\text{for }g\in L^{1}_\mathbb{C},
\]
where $d\mu\left(  x\right)  $ is a Haar measure on $\mathbb{Q}_{p}^{4}$. Let
$\mathcal{L}\left(  \mathbb{Q}_{p}^{4}\right)  $ be the space of
complex-valued continuous functions $g$ in $L_{\mathbb{C}}^{1}$ whose Fourier
transform $\mathcal{F}g$ is integrable. The measure $d\mu\left(  x\right)  $
can be normalized uniquely in such manner that
\[
(\mathcal{F}(\mathcal{F}g))(x)=g(-x)\text{ for every }g\text{ belonging to
}\mathcal{L}\left(  \mathbb{Q}_{p}^{4}\right)  .
\]
We say that $d\mu\left(  x\right)  $ is \textit{\ a self-dual measure
relative\ to} $\chi_{p}\left(  \mathfrak{B}\left(  \cdot,\cdot\right)
\right)  $. Notice that $d\mu\left(  x\right)  =C(\mathfrak{q})d^{4}x$ where
$C(\mathfrak{q})$ is a positive constant and $d^{4}x$ is the normalized\ Haar
measure on $\mathbb{Q}_{p}^{4}$. For further details about the material
presented in this section the reader may consult \cite{We1}.

We will also use the notation $\mathcal{F}_{x\rightarrow\xi}g$ and
$\widehat{g}$\ for the Fourier transform of $g$. The Fourier transform
$\mathcal{F}\left[  T\right]  $ of a distribution $T\in\mathcal{D}%
_{\mathbb{C}}^{^{\prime}}$ is defined by%
\[
\left(  \mathcal{F}\left[  T\right]  ,\varphi\right)  =\left(
T,\mathcal{F\varphi}\right)  \text{ for all }\varphi\in\mathcal{D}%
_{\mathbb{C}}\text{.}%
\]
The Fourier transform $T\rightarrow\mathcal{F}\left[  T\right]  $ is a linear
isomorphism from $\mathcal{D}_{\mathbb{C}}^{^{\prime}}$\ onto itself.
Furthermore, $T\left(  \xi\right)  =\mathcal{F}\left[  \mathcal{F}\left[
T\right]  \left(  -\xi\right)  \right]  $.

\begin{note}
Along this article we will use the notation $\mathfrak{q}(x)=x_{0}%
^{2}-\mathfrak{q}_{0}(\boldsymbol{x})$, where $\mathfrak{q}_{0}(\boldsymbol{x}%
)=sx_{1}^{2}+px_{2}^{2}-spx_{3}^{2}$ is an elliptic quadratic form. The
bilinear form corresponding to $\mathfrak{q}_{0}$ will be denoted $\mathfrak{B}%
_{0}(\cdot,\cdot)$. Then $\mathfrak{B}\left(  x,y\right)  =x_{0}%
y_{0}-\mathfrak{B}_{0}(\boldsymbol{x},\boldsymbol{y})$.
\end{note}

\subsection{The $p-$adic Minkowski space}

Take $\mathfrak{q}(x)$ as before, and define%
\[
G=\left[
\begin{array}
[c]{cccc}%
1 & 0 & 0 & 0\\
0 & -s & 0 & 0\\
0 & 0 & -p & 0\\
0 & 0 & 0 & sp
\end{array}
\right]  .
\]
Then $\mathfrak{q}(x)=x^\top Gx$, where $\top$ denotes the transpose of a matrix,
and $x$ is identified with the column vector $\left[  x_{0},x_{1},x_{2}%
,x_{3}\right]  ^\top$. The orthogonal group of $\mathfrak{q}$ is defined as
\begin{align}
\boldsymbol{O}(\mathfrak{q}) &  =\{\Lambda\in GL_{4}(\mathbb{Q}_{p}%
);\mathfrak{B}\left(  \Lambda x,\Lambda y\right)  =\mathfrak{B}\left(
x,y\right)  \}\nonumber\\
&  =\{\Lambda\in GL_{4}(\mathbb{Q}_{p});\Lambda^\top G\Lambda=G\}.\nonumber
\end{align}
Notice that any $\Lambda\in\boldsymbol{O}(\mathfrak{q})$ satisfies
$\det\Lambda=\pm1$. We call the quadratic space $(\mathbb{Q}_{p}%
^{4},\mathfrak{q})$ \textit{the }$p$\textit{-adic Minkowski space}, and we
define \textit{the }$p$\textit{-adic Lorentz group} to be $\boldsymbol{O}%
(\mathfrak{q})$. Later on, we will introduce the $p-$adic restricted Lorentz
group and the $p-$adic restricted \ Poincar\'{e} group.

\begin{remark}
Special relativity in the $p-$adic framework was discussed in
\cite{EU66}, however, our definitions of Lorentz group and `light cones' are
completely different to the ones used in this article. In \cite{Var et
al}-\cite{Var et al 2}, the authors investigated the representations of the
$p-$adic Poincar\'{e} group, our notion of Lorentz group agrees with the one
used in these works.
\end{remark}

\subsection{The Dirac distribution supported on a hypersurface}

Take $\mathfrak{f}\in\mathbb{Q}_{p}\left[  x_{0},x_{1},x_{2},x_{3}\right]  $
to be a non-constant polynomial. The hypersurface attached to $\mathfrak{f}$
is the set%
\[
H:=H(\mathfrak{f})=\left\{  x\in\mathbb{Q}_{p}^{4};\mathfrak{f}(x)=0\right\}  .
\]
We say that $H$ is a \textit{non-singular hypersurface,} if
\begin{equation}
\nabla\mathfrak{f}(x)\neq0\text{ for any }x\in H.\label{smooth_condition}%
\end{equation}
By using the $p-$adic implicit function theorem, see e.g. \cite{Igusa},
\cite{Serre}, one shows, like in the case $\mathbb{R}^{4}$, that $H$ is a
$p-$adic manifold embedded in $\mathbb{Q}_{p}^{4}$. More exactly, $H$ is a
closed submanifold of $\mathbb{Q}_{p}^{4}$ (which is a $p-$adic manifold of
dimension $4$) of codimension $1$. For further details about $p-$adic
manifolds the reader may consult \cite{Igusa}, \cite{Serre}.

The condition (\ref{smooth_condition}) implies the existence of a $3$-form
$\lambda$ (whose restriction to $H$ is unique) satisfying%
\begin{equation}
dx_{0}\wedge dx_{1}\wedge dx_{2}\wedge dx_{3}=d\mathfrak{f}\wedge
\lambda.\label{3_form_labmda}%
\end{equation}
Usually $\lambda$ is called a \textit{Gel'fand-Leray form} for $H$. We
denote by $d\lambda$ the measure induced by $\lambda$ on $H$. For the details
about the construction of $d\lambda$, the reader may consult \cite[Chapter
7]{Igusa}. This construction is similar to one done in the real case,
\cite[Chapter III]{G-S}.

The linear functional%
\[%
\begin{array}
[c]{ccc}%
\mathcal{D}_{\mathbb{K}} & \rightarrow & \mathbb{K}\\
&  & \\
\varphi & \rightarrow & \left(  \delta_{H},\varphi\right)  =\int
\limits_{H}\varphi\left(  x\right)  d\lambda
\end{array}
\]
gives rise to a distribution $\mathcal{D}_{\mathbb{K}}^{\prime}$, which is
called \textit{the Dirac distribution} $\delta_{H}$\ supported on $H$.

Denote $\mathbb{Q}_{p}^{\times}=\mathbb{Q}_{p}-\{0\}$. For $t\in\mathbb{Q}_{p}^{\times}$, we set
\[
V_{t}:=V_{t}(\mathfrak{q})=\{x\in\mathbb{Q}_{p}^{4};\mathfrak{q}(x)=t\}.
\]
Then $V_{t}$ is a non-singular hypersurface in $\mathbb{Q}_{p}^{4}$. The
orthogonal group $\boldsymbol{O}(\mathfrak{q})$ acts transitively on $V_{t}$.
On each non-empty orbit $V_{t}$ there is a non-zero, po\-sitive measure which
is invariant under $\boldsymbol{O}(\mathfrak{q})$ and unique up to
multiplication by a positive constant, see \cite[Proposition 2-2]%
{Rallis-Schiffman}.

For each $t\in\mathbb{Q}_{p}^{\times}$, let $d\mu_{t}$ be a measure on $V_{t}
$ invariant under $\boldsymbol{O}(\mathfrak{q})$. Since $V_{t}$ is closed in
$\mathbb{Q}_{p}^{4}$, it is possible to consider $d\mu_{t}$ as a measure on
$\mathbb{Q}_{p}^{4}$ supported on $V_{t}$, and by the using the Caratheodory
theorem, we can identify $d\mu_{t}$ with a positive distribution, i.e. if
$\phi$ is a non-negative function, then $\left(  d\mu_{t},\phi\right)  \geq0$.
The Rallis-Schiffman result above mentioned can be reformulated as follows: on
each non-empty orbit $V_{t}$ there is a non-zero, positive distribution which
is invariant under $\boldsymbol{O}(\mathfrak{q})$ and unique up to
multiplication by a positive constant.

Now, since $\delta_{V_{t}}$ is invariant under $\boldsymbol{O}(\mathfrak{q})
$, see \cite[Lemma 156]{Zuniga-LNM-2016} for a similar calculation, we
conclude that $d\mu_{t}$ agrees (up to a positive constant) with
$\delta_{V_{t}}$. From now on we identify $\delta_{V_{t}}$ with $d\mu_{t}$.

\begin{note}
From now on, we will use $\delta\left(  \mathfrak{f}\right)  $ to
denote the Dirac distribution supported on the non-singular hypersurface
attached to the polynomial $\mathfrak{f}$.
\end{note}

\subsection{The spaces $\mathcal{H}_{\infty}$}

The Bruhat-Schwartz space $\mathcal{D}_{\mathbb{K}}$ is not invariant under
the action of pseudodifferential operators. In \cite{Zuniga-FAA-2017}, see
also \cite[Chapter 10]{KKZuniga}, the third author introduced a class of
nuclear countably Hilbert spaces which are invariant under the action of a
large class of pseudo-differential operators. In this section, we review some
basic results about these spaces that we will use in the remaining sections.

\begin{note}
We set $\mathbb{R}_{+}:=\{x\in\mathbb{R}:x\geq0\}$, $[\xi]_{p}:=\max
(1,\Vert\xi\Vert_{p})$ and consider $\mathbb{N}$ to be the set of non-negative
integers.
\end{note}

We define for $f,g\in\mathcal{D}_{\mathbb{K}}$, with $\mathbb{K}%
=\mathbb{R},\mathbb{C}$, the following scalar product:%
\[
\langle f,g\rangle_{l}:=\int_{\mathbb{Q}_{p}^{4}}[\xi]_{p}^{l}\widehat{f}%
(\xi)\overline{\widehat{g}(\xi)}d^{4}\xi,
\]
for $l\in\mathbb{N}$, where the bar denotes the complex conjugate. We also set
$\Vert f\Vert_{l}^{2}=\langle f,f\rangle_{l}$. Notice that $\Vert\cdot
\Vert_{l}\leq\Vert\cdot\Vert_{m}$ for $l\leq m$. Let denote by $\mathcal{H}%
_{l}(\mathbb{Q}_{p}^{4},\mathbb{K})=:\mathcal{H}_{l}(\mathbb{K})$ the
completion of $\mathcal{D}_{\mathbb{K}}$ with respect to $\langle\cdot
,\cdot\rangle_{l}$. Then $\mathcal{H}_{m}(\mathbb{K})\hookrightarrow
\mathcal{H}_{l}(\mathbb{K})$ is a continuous embedding for $l\leq m$. We set
\[
\mathcal{H}_{\infty}(\mathbb{Q}_{p}^{4},\mathbb{K}):=\mathcal{H}_{\infty
}(\mathbb{K})=\bigcap_{l\in\mathbb{N}}\mathcal{H}_{l}(\mathbb{K}).
\]
Notice that $\mathcal{H}_{0}(\mathbb{K})=L_{\mathbb{K}}^{2}$ and that
$\mathcal{H}_{\infty}(\mathbb{K})\subset L_{\mathbb{K}}^{2}$. With the
topology induced by the family of seminorms $\Vert\cdot\Vert_{l}$,
$\mathcal{H}_{\infty}(\mathbb{K})$ becomes a locally convex space, which is
metrizable. Indeed,%
\[
d(f,g):=\underset{l\in\mathbb{N}}{max}\left\{  2^{-l}\dfrac{\Vert f-g\Vert
_{l}}{1+\Vert f-g\Vert_{l}}\right\}  \text{, for }f\text{, }g\in
\mathcal{H}_{\infty}(\mathbb{K})\text{,}%
\]
is a metric for the topology of the convex topological space $\mathcal{H}%
_{\infty}(\mathbb{K})$. A sequence $\{f_{l}\}_{l\in\mathbb{N}}\in
(\mathcal{H}_{\infty}(\mathbb{K}),d)$ converges to $f\in\mathcal{H}_{\infty
}(\mathbb{K})$, if and only if, $\{f_{l}\}_{l\in\mathbb{N}}$ converges to $f $
in the norm $\Vert\cdot\Vert_{l}$ for all $l\in\mathbb{N}$. From this
observation it follows that the topology of $\mathcal{H}_{\infty}(\mathbb{K})$
coincides with the projective limit topology $\tau_{P}$. An open neighborhood
base at zero of $\tau_{P}$ is given by the choice of $\epsilon>0$ and
$l\in\mathbb{N} $, and the sets
\[
U_{\epsilon,l}:=\{f\in\mathcal{H}_{\infty}(\mathbb{K}):\Vert f\Vert
_{l}<\epsilon\}.
\]
The space $\mathcal{H}_{\infty}(\mathbb{K})$ endowed with the topology
$\tau_{P}$ is a countably Hilbert space in the sense of Gel'fand and Vilenkin,
see e.g. \cite[ Chapter I, Section 3.1]{Gel-Vil} or \cite[Section 1.2]{Obata}.
Furthermore $(\mathcal{H}_{\infty}(\mathbb{K}),\tau_{P})$ is metrizable and
complete and hence a Fr\'{e}chet space, cf. \cite[Lemma 3.3]{Zuniga-FAA-2017}.
In addition, the completion of the metric space $(\mathcal{D}_{\mathbb{K}%
}(\mathbb{Q}_{p}^{4}),d)$ is $(\mathcal{H}_{\infty}(\mathbb{K}),d)$, and this
space is a nuclear countably Hilbert space, see \cite[Lemma 3.4, Theorem
3.6]{Zuniga-FAA-2017} or \cite[Chapter 10]{KKZuniga}.

For $m\in\mathbb{N}$ and $T\in\mathcal{D}_{\mathbb{K}}^{\prime}$, we set
\[
\Vert T\Vert_{-m}^{2}:=\int_{\mathbb{Q}_{p}^{4}}[\xi]_{l}^{-m}|\widehat{T}%
(\xi)|^{2}d^{4}\xi.
\]
Then $\mathcal{H}_{-m}(\mathbb{K}):=\mathcal{H}_{-m}(\mathbb{Q}_{p}%
^{4},\mathbb{K})=\{T\in\mathcal{D}_{\mathbb{K}}^{\prime};\Vert T\Vert_{-m}%
^{2}<\infty\}$ is a Hilbert space over $\mathbb{K}$. Denote by $\mathcal{H}%
_{m}^{\ast}\left(  \mathbb{K}\right)  $ the strong dual space of
$\mathcal{H}_{m}\left(  \mathbb{K}\right)  $. It is useful to suppress the
correspondence between $\mathcal{H}_{m}^{\ast}\left(  \mathbb{K}\right)  $ and
$\mathcal{H}_{m}\left(  \mathbb{K}\right)  $ given by the Riesz theorem.
Instead we identify $\mathcal{H}_{m}^{\ast}\left(  \mathbb{K}\right)  $ and
$\mathcal{H}_{-m}\left(  \mathbb{K}\right)  $ by associating $T\in
\mathcal{H}_{-m}\left(  \mathbb{K}\right)  $ with the functional on
$\mathcal{H}_{m}\left(  \mathbb{K}\right)  $ given by
\begin{equation}
\lbrack T,g]:=\int_{\mathbb{Q}_{p}^{4}}\overline{\widehat{T}(\xi)}\widehat
{g}(\xi)d^{4}\xi.\label{ec[T,g]}%
\end{equation}
Notice that $|[T,g]|\leq\Vert T\Vert_{-m}\Vert g\Vert_{m}$. Now by a
well-known result in the theory of countable Hilbert spaces, see
\cite{Gel-Vil}, $\mathcal{H}_{0}^{\ast}\left(  \mathbb{K}\right)
\subset\mathcal{H}_{1}^{\ast}\left(  \mathbb{K}\right)  \subset\ldots
\subset\mathcal{H}_{m}^{\ast}\left(  \mathbb{K}\right)  \subset\ldots$ and
\begin{equation}
\mathcal{H}_{\infty}^{\ast}\left(  \mathbb{K}\right)  =\bigcup_{m\in
\mathbb{N}}\mathcal{H}_{-m}\left(  \mathbb{K}\right)  =\{T\in\mathcal{D}%
_{\mathbb{K}}^{\prime};\Vert T\Vert_{-l}<\infty,\text{ for some }%
l\in\mathbb{N}\}\label{dual_space}%
\end{equation}
as vector spaces. Since $\mathcal{H}_{\infty}\left(  \mathbb{K}\right)  $ is a
nuclear space, the weak and strong convergence are equivalent in
$\mathcal{H}_{\infty}^{\ast}\left(  \mathbb{K}\right)  $, see e.g.
\cite{Gel-Vil}. We consider $\mathcal{H}_{\infty}^{\ast}\left(  \mathbb{K}%
\right)  $ endowed with the strong topology. On the other hand, let
$B:\mathcal{H}_{\infty}^{\ast}\left(  \mathbb{K}\right)  \times\mathcal{H}%
_{\infty}\left(  \mathbb{K}\right)  \rightarrow\mathbb{K}$ be a bilinear
functional. Then $B$ is continuous in each of its arguments if and only if
there exist norms $\Vert\cdot\Vert_{m}^{(a)}$ in $\mathcal{H}_{m}^{\ast
}\left(  \mathbb{K}\right)  $ and $\Vert\cdot\Vert_{l}^{(b)}$ in
$\mathcal{H}_{l}\left(  \mathbb{K}\right)  $ such that $|B(T,g)|\leq
M\Vert T\Vert_{m}^{(a)}\Vert g\Vert_{l}^{(b)}$ with $M$ a positive
constant independent of $T$ and $g$, see e.g. \cite{Gel-Vil}. This implies
that (\ref{ec[T,g]}) is a continuous bilinear form on $\mathcal{H}_{\infty
}^{\ast}\left(  \mathbb{K}\right)  \times\mathcal{H}_{\infty}\left(
\mathbb{K}\right)  $, which we will use as a paring between $\mathcal{H}%
_{\infty}^{\ast}\left(  \mathbb{K}\right)  $ and $\mathcal{H}_{\infty}\left(
\mathbb{K}\right)  $.

\begin{remark}
The spaces $\mathcal{H}_{\infty}\left(  \mathbb{K}\right)  \subset
L_{\mathbb{K}}^{2}\subset\mathcal{H}_{\infty}^{\ast}\left(  \mathbb{K}\right)
$ form a Gel'fand triple (also called a rigged Hilbert space), i.e.
$\mathcal{H}_{\infty}\left(  \mathbb{K}\right)  $ is a nuclear space which is
densely and continuously embedded in $L^{2}_\mathbb{K}$ and $\Vert g\Vert_{L_{\mathbb{K}%
}^{2}}^{2}=[g,g]$. This Gel'fald triple was introduced in
\cite{Zuniga-FAA-2017}.
\end{remark}

The following result will be used later on:

\begin{lemma}
\label{lemma4}With the above notation, the following assertions hold:
\begin{enumerate}[(i)]
\item $\mathcal{H}_{l}(\mathbb{K})=\{f\in L_{\mathbb{K}}^{2};\Vert f\Vert
_{l}<\infty\}=\{T\in\mathcal{D}_{\mathbb{K}}^{\prime};\Vert T\Vert_{l}%
<\infty\};$

\item $\mathcal{H}_{\infty}(\mathbb{K})=\{f\in L_{\mathbb{K}}^{2};\Vert
f\Vert_{l}<\infty$, for any $l\in\mathbb{N}\}$;

\item $\mathcal{H}_{\infty}(\mathbb{K})=\{T\in\mathcal{D}_{\mathbb{K}}%
^{\prime};\Vert T\Vert_{l}<\infty$, for any $l\in\mathbb{N}\}$.
\end{enumerate}
\end{lemma}

For the proof the reader may consult (\cite[Lemma 3.2]{Zuniga-PNUAA-2017}) or
\cite[Lemma 10.8]{KKZuniga}.

\section{Fundamental Solutions for Pseudo-differential Operators of Klein-Gordon Type}

\subsection{Some preliminary results}

For $\alpha>0$, $m\in\mathbb{Q}_{p}^{\times}$, and $\mathfrak{q}$ as before,
we define the following pseudo-differential operator:
\begin{equation}
\square_{\mathfrak{q},\alpha,m}=\mathcal{F}^{-1}\circ|\mathfrak{q}-m^{2}%
|_{p}^{\alpha}\circ\mathcal{F},\label{ec Klein-Gordon}%
\end{equation}
where $|\mathfrak{q}-m^{2}|_{p}^{\alpha}$ denotes the multiplication operator
by the function $|\mathfrak{q}-m^{2}|_{p}^{\alpha}$. We call operators of type
(\ref{ec Klein-Gordon}), $p$\textit{-adic Klein-Gordon pseudo-differential
operators}. These operators were introduced by Z\'{u}\~{n}iga-Galindo, see
\cite[Chapter 6]{Zuniga-LNM-2016} and the references therein.

In this section, we consider operators $\square_{q,\alpha,m}$ with domain
\[
Dom(\square_{\mathfrak{q},\alpha,m})=\{T\in\mathcal{D}_{\mathbb{C}}^{\prime
}:|\mathfrak{q}-m^{2}|_{p}^{\alpha}\mathcal{F}T\in\mathcal{D}_{\mathbb{C}%
}^{\prime}\}.
\]

\begin{remark}
Notice that
\[
\square_{\mathfrak{q},\alpha,m}\left( T\left(  mx\right)\right)  =|m|_{p}^{2\alpha}\left(
\square_{\mathfrak{q},\alpha,1}T\right)  \left(  mx\right)  \text{ for any
}T\in Dom(\square_{\mathfrak{q},\alpha,m}).
\]
Consequently, we \ may normalize the mass $m$ to one. From now on we assume
that $m=1$, and we use the notation $\square_{\mathfrak{q},\alpha}$ instead of
$\square_{\mathfrak{q},\alpha,1}$.
\end{remark}

\begin{definition}
We say that $E_{\mathfrak{q},\alpha}\in\mathcal{D}_{\mathbb{C}}^{\prime}$ is a
fundamental solution for
\begin{equation}
\square_{\mathfrak{q},\alpha}u=\varphi,\label{defsolfundamental}%
\end{equation}
if $u=E_{\mathfrak{q},\alpha}\ast\varphi$ is a solution of
(\ref{defsolfundamental}) in $\mathcal{D}_{\mathbb{C}}^{\prime}$, for any
$\varphi\in\mathcal{D}_{\mathbb{C}}$.
\end{definition}

From now on, by an abuse of language, we will say that $E_{\mathfrak{q},\alpha}$
is a fundamental solution of $\square_{\mathfrak{q},\alpha}$.

\begin{lemma}
\label{LemaFS} $E_{\mathfrak{q},\alpha}$ is a fundamental solution of
$\ \square_{\mathfrak{q},\alpha}$ if and only if
\begin{equation}
|\mathfrak{q}-1|_{p}^{\alpha}\mathcal{F}(E_{\mathfrak{q},\alpha}%
)=1\label{ec11FS}%
\end{equation}
in $\mathcal{D}_{\mathbb{C}}^{\prime}$.
\end{lemma}

\begin{proof}
If $E_{\mathfrak{q},\alpha}$ is a fundamental solution of $\ \square
_{\mathfrak{q},\alpha}$, then
\[
\left(  |\mathfrak{q}-1|_{p}^{\alpha}\mathcal{F}(E_{\mathfrak{q},\alpha
})-1\right) \cdot \mathcal{F}\varphi=0,
\]
for any test function in $\mathcal{D}_{\mathbb{C}}$, which implies
(\ref{ec11FS}). Now, if (\ref{ec11FS}) holds, by using the fact that the
product of two distributions, if it exists, is commutative and associative
(see e.g. \cite[p. 127. Theorem 3.19]{Taibleson}), we get that
\[
\left(  |\mathfrak{q}-1|_{p}^{\alpha}\mathcal{F}\varphi\right)\cdot 
\mathcal{F}(E_{\mathfrak{q},\alpha})=\mathcal{F}\varphi
\]
for any test function $\varphi$.
\end{proof}

\subsection{The $p-$adic submanifold $V$}

Since $\mathfrak{q}(k)=k_{0}^{2}-sk_{1}^{2}-pk_{2}^{2}+spk_{3}^{2}$, where
$s\in\mathbb{Z}_{p}^{\times}=\mathbb{Z}_{p}-\{0\}$ a quadratic non-residue $\operatorname{mod}$
$p$, is an elliptic quadratic form (i.e. $\mathfrak{q}(k)=0\Leftrightarrow
k=0$), we have
\begin{equation}
|\mathfrak{q}(k)|_{p}\geq\left( \underset{x\in S_{0}^{4}}{\inf}|\mathfrak{q}%
(x)|_{p}\right) \Vert k\Vert_{p}^{2},\label{desigualdaLNM}%
\end{equation}
see e.g. \cite[Lemma 25]{Zuniga-LNM-2016}. Set
\[
V:=\{k=(k_{0},\boldsymbol{k})\in\mathbb{Q}_{p}\times\mathbb{Q}_{p}%
^{3};\mathfrak{q}(k)=1\}.
\]
By using (\ref{desigualdaLNM}), and the fact that $\underset{x\in S_{0}^{4}%
}{\inf}|\mathfrak{q}(x)|_p=p^{-1}$, we get that $V\subseteq\mathbb{Z}_{p}^{4}$,
which implies that $V$ is a compact submanifold of $\mathbb{Z}_{p}^{4}$\ of
codimension $1$. Let us emphasize that $V$
is bounded (in contrast to the classical case). 
Given $(\widetilde{k}_{0},\widetilde{\boldsymbol{k}})\in V$ with $\widetilde{k}_{0}\neq 0$,
by applying the $p-$adic implicit function theorem, see e.g. \cite{Igusa},
there exist open \ and compact subsets $U_{j}^{0}\subset\mathbb{Z}_{p}$,
$U_{j}^{1}\subset\mathbb{Z}_{p}^{3}$ such that $(\widetilde{k}_{0}%
,\widetilde{\boldsymbol{k}})\in U_{j}=U_{j}^{0}\times U_{j}^{1}$, and a
$p-$adic analytic function $h_{j}(\boldsymbol{x}):U_{j}^{1}\rightarrow
U_{j}^{0} $ such that%
\[
V\cap U_{j}=\left\{  (k_{0},\boldsymbol{k})\in U_{j};k_{0}=h_{j}%
(\boldsymbol{k})\right\}  .
\]
Notice that $k_{0}=-h_{j}(\boldsymbol{k})$ is also a `local parametrization'
of $V$. By using the compactness of $V$, there exists a finite number of
analytic functions $\pm h_{j}(\boldsymbol{k}):U_{j}^{1}\rightarrow\pm
U_{j}^{0}$, $j=1,\ldots,N$ such that%
\begin{align*}
V &  =\bigsqcup\limits_{j=1}^{N}\left\{  (k_{0},\boldsymbol{k})\in U_{j}%
^{0}\times U_{j}^{1};k_{0}=h_{j}(\boldsymbol{k})\right\}  \bigsqcup\\
&  \bigsqcup\limits_{j=1}^{N}\left\{  (k_{0},\boldsymbol{k})\in-U_{j}%
^{0}\times U_{j}^{1};k_{0}=-h_{j}(\boldsymbol{k})\right\}\bigsqcup W\,,
\end{align*}
where $W=\{(0,\boldsymbol{k}):\mathfrak{q}_0(\boldsymbol{k})=1\}$.
We set $U_{\mathfrak{q}}:=\bigsqcup\limits_{j=1}^{N}U_{j}^{1}\subset
\mathbb{Z}_{p}^{3}$. We now define in $U_{\mathfrak{q}}$, two analytic
functions as follows:%
\[%
\begin{array}
[c]{cccc}%
U_{\mathfrak{q}} & \rightarrow & \mathbb{Q}_{p} & \\
&  &  & \\
\boldsymbol{k} & \rightarrow & \pm\sqrt{1+sk_{1}^{2}+pk_{2}^{2}-spk_{3}^{2}%
}=: & \pm\sqrt{\omega\left(  \boldsymbol{k}\right)  },
\end{array}
\]
where $\pm\sqrt{\omega\left(  \boldsymbol{k}\right)  }\mid_{U_{j}^{1}}=\pm
h_{j}(\boldsymbol{k})$.

\subsubsection{\label{Sect_positivity}A notion of positivity}

We set $\mathbb{F}_{p}^{\times}=\left[  \mathbb{F}_{p}^{\times}\right]
_{+}\bigsqcup\left[  \mathbb{F}_{p}^{\times}\right]  _{-}$, where $\left[
\mathbb{F}_{p}^{\times}\right]  _{+}:=\left\{  \overline{1},\ldots,\overline{\frac{p-1}{2}}\right\}
$ and 
$\left[\mathbb{F}_{p}^{\times}\right]_{-}=\left\{\overline{\frac{p+1}{2}},\ldots,\overline{p-1}\right\}$. We define the elements of $\left[  \mathbb{F}%
_{p}^{\times}\right]  _{+}$ as \textit{positive} and the elements of $\left[
\mathbb{F}_{p}^{\times}\right]  _{-}$\ as \textit{negative}. Notice that since
$p\neq2$,%
\[%
\begin{array}
[c]{cccc}%
\left[  \mathbb{F}_{p}^{\times}\right]  _{+} & \rightarrow & \left[
\mathbb{F}_{p}^{\times}\right]  _{-} & \\
&  &  & \\
\overline{y} & \rightarrow & -\overline{y}\ \mathrm{mod}\ p &
\end{array}
\]
is a bijection. Now, we say that \ a non-zero $p-$adic number
\[
a=p^{-L}\left(  a_{0}+a_{1}p+\ldots\right)  \text{, with }L\in\mathbb{Z}\text{
and }a_{0}\neq0,
\]
is \textit{positive (denoted as }$a>0$\textit{)} if $\overline{a}_{0}%
\in\left[  \mathbb{F}_{p}^{\times}\right]  _{+}$, otherwise we say that $a$ is
\textit{negative (denoted as }$a<0$\textit{)}. This is a well-defined and
useful notion of `positivity' in $\mathbb{Q}_{p}^{\times}$, however, this
notion of positivity is not compatible with the field operations,
consequently, this notion does not give rise to an order in $\mathbb{Q}%
_{p}^{\times}$. We also recall that in the case $p\neq2$, the equation $x^{2}=a$
has two solutions in $\mathbb{Q}_{p}$ if an only $L$ is even and the
congruence $z^{2}\equiv \overline{a_0}$ $\operatorname{mod}$ $p$ has two solutions, one
\ in $\left[  \mathbb{F}_{p}^{\times}\right]  _{+}$ and the other in $\left[
\mathbb{F}_{p}^{\times}\right]  _{-}$, we denote them as $\pm\sqrt{a_{0}}%
\in\mathbb{F}_{p}^{\times}$. Then
\[
x=p^{-\frac{L}{2}}\left(  \sqrt{a_{0}}+b_{1}p+b_{2}p^{2}+\ldots\right)  ,
\]
where the $b$'s are recursively determined by $\sqrt{a_{0}}$, i.e. $b_{1}%
=f_{1}\left(  \sqrt{a_{0}}\right)  $, $b_{2}=\allowbreak f_{2}\left(
\sqrt{a_{0}},b_{1}\right)  ,\ldots$, and
\begin{align*}
-x &  =-p^{-\frac{L}{2}}\left(  \sqrt{a_{0}}+b_{1}p+b_{2}p^{2}+\ldots\right)
\\
&  =p^{-\frac{L}{2}}\left(  p-\sqrt{a_{0}}+\left(  p-1-b_{1}\right)  p+\left(
p-1-b_{2}\right)  p^{2}+\ldots\right)  .
\end{align*}
We now define%
\begin{align*}
V^{+} &  =\left\{  (k_{0},\boldsymbol{k})\in V;k_{0}>0\text{ and }k_{0}%
=\sqrt{\omega\left(  \boldsymbol{k}\right)  }\right\}\,,\\
V^{-} &  =\left\{  (k_{0},\boldsymbol{k})\in V;k_{0}<0\text{ and }k_{0}%
=-\sqrt{\omega\left(  \boldsymbol{k}\right)  }\right\}  .
\end{align*}
We call $V^{+}$ \textit{the positive mass shell} and $V^{-}$ the
\textit{negative mass shell}. Therefore%
\[
V=V^{+}\bigsqcup V^{-}\bigsqcup W\,.
\]
Consequently, $W$ has $d\lambda$-measure zero, so $\int_{W}\varphi d\lambda\equiv0$
for any $\varphi\in\mathcal{D}_{\mathbb{C}}$.

\subsection{The distributions $\delta_{V^{\pm}}$}

\begin{remark}
\label{note_q}Set $\mathfrak{q}(k_{0},\boldsymbol{k}):=k_{0}^{2}%
-\mathfrak{q}_{0}(\boldsymbol{k})$, then
\[
W=\left\{  (k_{0},\boldsymbol{k})\in\mathbb{Z}_{p}^{4};\mathfrak{q}%
(0,\boldsymbol{k})=1\right\}  =\left\{  \boldsymbol{k}\in\mathbb{Z}_{p}%
^{3};-\mathfrak{q}_{0}(\boldsymbol{k})=1\right\}  .
\]
A necessary and sufficient condition to have \ $W\neq\emptyset$ is that%
\begin{equation}
-\mathfrak{q}_{0}(\boldsymbol{k})\equiv1\text{ }\operatorname{mod}\text{
}p\text{ \ i.e. }-sk_{1}^{2}\equiv1\operatorname{mod}\text{ }p\text{.}%
\label{condition_W}%
\end{equation}
The sufficiency of condition (\ref{condition_W}) follows from the
Hensel-Newton lemma, see e.g. \cite[Lemma 1]{Greeberg}. The existence of
solutions for congruence (\ref{condition_W}) requires the computation of the
following Legendre symbol:%
\[
\left(  \frac{-s^{-1}}{p}\right)  =
\begin{cases}
1 \mbox{ if congruence } \eqref{condition_W} \mbox{ has a solution,}\\[6pt]
-1 \mbox{ if congruence } \eqref{condition_W} \mbox{ has no solution.}%
\end{cases}
\]
By using the fact that the Legendre symbol is a multiplicative function and
that $\left(  \frac{s}{p}\right)  =-1$, we get that%
\[
\left(  \frac{-s^{-1}}{p}\right)  =
\begin{cases}
-1 \mbox{ if }p\equiv 1\ \mathrm{mod}\ 4 \Leftrightarrow W=\emptyset\\[6pt]

1 \mbox{ if }p\equiv 3\ \mathrm{mod}\ 4 \Leftrightarrow W\neq\emptyset\,.
\end{cases}
\]
Taking these results into account, we will set $p\equiv 1\mathrm{mod}\, 4$ from now on,
so $W=\emptyset$.
\end{remark}

We set $\delta_{V}=\delta\left(  \mathfrak{q}-1\right)  $\ as before. The
characteristic functions $1_{V^{\pm}}$ are locally constant functions, so
the product distributions $1_{V^{\pm}}\delta\left(  \mathfrak{q}-1\right)  $
are well-defined. We set $\delta_{V^{\pm}}:=1_{V^{\pm}}\delta\left(
\mathfrak{q}-1\right)  $. Then
\[
\delta_{V}=\delta_{V^{+}}+\delta_{V^{-}}\text{ \ in }\mathcal{D}_{\mathbb{C}%
}^{\prime}\text{.}%
\]
In the open subset of $\mathbb{Q}_{p}^{4}$ defined by $k_{0}\neq0$, the
$3$-form $\lambda$ satisfying (\ref{3_form_labmda}) (with $\mathfrak{f=q}$)
\ is given by%
\[
\lambda=\frac{dk_{1}\wedge dk_{2}\wedge dk_{3}}{2k_{0}},
\]
therefore the corresponding measure is%
\[
d\lambda=\frac{dk_{1}dk_{2}dk_{3}}{\left\vert k_{0}\right\vert _{p}}%
=\frac{d^{3}\boldsymbol{k}}{\left\vert \sqrt{\omega\left(  \boldsymbol{k}%
\right)  }\right\vert _{p}}=\frac{d^{3}\boldsymbol{k}}{\sqrt{\left\vert
1+\mathfrak{q}_{0}\left(  \boldsymbol{k}\right)  \right\vert _{p}}}\text{ for
}\boldsymbol{k}\in U_{\mathfrak{q}}\text{.}%
\]
If $p\equiv1$ $\operatorname{mod}$ $4$, then $\sqrt{\omega\left(
\boldsymbol{k}\right)  }\neq0$ or any $\boldsymbol{k}\in U_{\mathfrak{q}}$,
and
\[
\left(  \delta_{V^{\pm}},\varphi\right)  =\int_{U_{\mathfrak{q}}}%
\varphi\left(  \pm\sqrt{\omega\left(  \boldsymbol{k}\right)  },\boldsymbol{k}%
\right)  \frac{d^{3}\boldsymbol{k}}{\left\vert \sqrt{\omega\left(
\boldsymbol{k}\right)  }\right\vert _{p}}\text{ for any }\varphi\in
\mathcal{D}_{\mathbb{C}}\text{.}%
\]

\begin{remark}
Take $\overline{a}\in\mathbb{F}_{p}^{4}$ satisfying $\mathfrak{q}(\overline
{a})\equiv1$ $\operatorname{mod}$ $p$. Since $\nabla\mathfrak{q}(\overline
{a})\not \equiv 0$ $\operatorname{mod}$ $p$, by the Hensel-Newton lemma, see
e.g. \cite[Lemma 1]{Greeberg}, there exists $b\in\mathbb{Z}_{p}^{4}$ such that
$\mathfrak{q}(b)=1$ and $b\equiv\overline{a}$ $\operatorname{mod}$ $p$. This
$b$ is not unique. We now define the following tubular neighborhood of $V$:%
\[
E_{V}=\underset{\underset{\mathfrak{q}(\overline{a})\equiv1\text{ }mod\text{
}p}{\overline{a}\in\mathbb{F}_{p}^{4}}}{\bigsqcup}b+p\mathbb{Z}_{p}^{4},
\]
where implicitly we are choosing for each $\overline{a}\in\mathbb{F}_{p}^{4} $
a point $b$\ in $V$. Notice that $E_{V}\neq\emptyset$. Indeed, the solution
set of the equation $k_{0}^{2}-sk_{1}^{2}\equiv1$ $\operatorname{mod}$ $p$
contains the set $A:=\left\{  \left(  1,0,\overline{u},\overline{v}\right)
;\overline{u},\overline{v}\in\mathbb{F}_{p}\right\}  $, and the gradient
satisfies \ the condition $\nabla\mathfrak{q}(\overline{y})\not \equiv 0$
$\operatorname{mod}$ $p$, for any $\overline{y}$\ in $A$.
\end{remark}

\begin{lemma}
\label{Lemma5}Let $b=(b_{0},b_{1},b_{2},b_{3})\in V$, with $b_{0}\in
\mathbb{Z}_{p}^{\times}$.Then
\begin{gather*}
(\delta(\mathfrak{q}(k)-1),\phi(k)\Omega(p\Vert k-b\Vert_{p}))=\\
p^{-3}\int_{\mathbb{Z}_{p}^{3}}\phi(b_{0}+pf(0,u_{1},u_{2},u_{3}),b_{1}%
+pu_{1},b_{2}+pu_{2},b_{3}+pu_{3})du_{1}du_{2}du_{3},
\end{gather*}
where $f(0,u_{1},u_{2},u_{3})$ is a $p-$adic analytic function on the ball
$\mathbb{Z}_{p}^{3}$.
\end{lemma}

\begin{proof}
Recall that
\[
(\delta(\mathfrak{q}(k)-1),\phi(k)\Omega(p\Vert k-b\Vert_{p}))=\int
_{V\cap\left(  b+p\mathbb{Z}_{p}^{4}\right)  }\phi(k)\dfrac{dk_{1}dk_{2}%
dk_{3}}{|k_{0}|_{p}}.
\]
Now, by changing variables as $k=b+pz$,
\begin{equation}
(\delta(\mathfrak{q}(k)-1),\phi(k)\Omega(p\Vert k-b\Vert_{p}))=p^{-3}%
\int_{\{\mathfrak{q}(b+pz)=1\}\cap\mathbb{Z}_{p}^{4}}\phi(b+pz)dz_{1}%
dz_{2}dz_{3},\label{ec19_A}%
\end{equation}
where we are assuming that $z_{0}$ is an analytic function of the variables
$z_{1},z_{2},z_{3}$. We set
\begin{equation}
u=F(z)\text{, with }u_{0}=\dfrac{\mathfrak{q}(b+pz)-1}{p},u_{i}=z_{i}\text{
for }i=1,2,3.\label{change_of_var}%
\end{equation}
Then $Jac_{F}(z)\equiv2b_{0}+2pz_{0}\equiv\overline{b}_{0}\not \equiv 0$
$\operatorname{mod}$ $p$, by \cite[Lemma 7.4.3]{Igusa}, $F$ gives rise to an
analytic isomorphism from $\mathbb{Z}_{p}^{4}$ into itself which preserves the
Haar measure, in this coordinate system $\{\mathfrak{q}(b+pz)=1\}\cap
\mathbb{Z}_{p}^{4}$ becomes $\{u_{0}=0\}\times\mathbb{Z}_{p}^{3}$, and
(\ref{ec19_A}) takes the form
\begin{gather}
(\delta(\mathfrak{q}(k)-1),\phi(k)\Omega(p\Vert k-b\Vert_{p}))=\label{ec19}\\
p^{-3}\int_{\mathbb{Z}_{p}^{3}}\phi(b_{0}+pf(0,u_{1},u_{2},u_{3}),b_{1}%
+pu_{1},b_{2}+pu_{2},b_{3}+pu_{3})du_{1}du_{2}du_{3},\nonumber
\end{gather}
where $f(0,u_{1},u_{2},u_{3}):\mathbb{Z}_{p}^{3}\rightarrow\mathbb{Z}_{p}$ is
a $p-$adic analytic function.
\end{proof}

\begin{remark}\label{note6}
Let us comment about some related results.
\begin{enumerate}[(i)]
\item In the case $b_{0}\in p\mathbb{Z}_{p}$, $b_{1}\in
\mathbb{Z}_{p}^{\times}$, a calculation similar to the one done in the
proof of Lemma \ref{Lemma5} shows that%
\begin{gather*}
(\delta(\mathfrak{q}(k)-1),\phi(k)\Omega(p\Vert k-b\Vert_{p}))=\\
p^{-3}\int_{\mathbb{Z}_{p}^{3}}\phi(b_{0}+pu_{0},b_{1}+pg(u_{0},0,u_{2}%
,u_{3}),b_{2}+pu_{2},b_{3}+pu_{3})du_{0}du_{2}du_{3},
\end{gather*}
where $g(u_{0},0,u_{2},u_{3}):\mathbb{Z}_{p}^{3}\rightarrow\mathbb{Z}_{p}$ is
a $p-$adic analytic function.

\item In the case $b_{0}\in p\mathbb{Z}_{p}$, $b_{1}\in p\mathbb{Z}_{p}$,
$b_{2}\in\mathbb{Z}_{p}^{\times}$, we have%
\begin{gather*}
\left\{  \mathfrak{q}(k)=1\right\}  \cap\left[  p\mathbb{Z}_{p}\times
p\mathbb{Z}_{p}\times\left[  b_{2}+p\mathbb{Z}_{p}\right]  \times\left[
b_{3}+p\mathbb{Z}_{p}\right]  \right]  =\\
\left\{  p\left(  pk_{0}^{2}-spk_{1}^{2}-k_{2}^{2}+sk_{3}^{2}\right)
=1\right\}  \cap\left[  \mathbb{Z}_{p}\times\mathbb{Z}_{p}\times\left[
b_{2}+p\mathbb{Z}_{p}\right]  \times\left[  b_{3}+p\mathbb{Z}_{p}\right]
\right]  =\emptyset.
\end{gather*}
A similar result is valid in the cases where $b_{0}\in p\mathbb{Z}_{p}$,
$b_{1}\in p\mathbb{Z}_{p}$,$b_{2}\in p\mathbb{Z}_{p}$, $b_{3}\in\mathbb{Z}%
_{p}^{\times}$, and where $b_{0}\in p\mathbb{Z}_{p}$, $b_{1}\in p\mathbb{Z}%
_{p}$,$b_{2}\in p\mathbb{Z}_{p}$, $b_{3}\in p\mathbb{Z}_{p}$.
\end{enumerate}
\end{remark}

\subsection{Fundamental solutions}

The existence of fundamental solutions for operators $\square_{\mathfrak{q}%
,\alpha}$ is closely related to the meromorphic continuation of the Igusa
local zeta function attached to the polynomial $\mathfrak{q}-1$, which is the
distribution defined as%
\begin{equation}
\left(  |\mathfrak{q}-1|_{p}^{s},\theta\right)  =\int_{\mathbb{Q}_{p}%
^{4}\setminus V}|\mathfrak{q}(x)-1|_{p}^{s}\theta(x)d^{4}x\text{ for
}\operatorname{Re}(s)>0\text{, and }\theta\in\mathcal{D}_{\mathbb{C}%
}.\label{zeta_function}%
\end{equation}
Here we use \ that for $a>0$ and $s\in\mathbb{C}$, $a^{s}=e^{s\ln a}$.
Integrals of type (\ref{zeta_function}) admit meromorphic continuations to the
whole complex plane as rational functions of $p^{-s}$, see \cite[Theorem
8.2.1]{Igusa}.

For further calculations, we rewrite (\ref{zeta_function}) as
\begin{align*}
(|\mathfrak{q}(x)-1|_{p}^{s},\theta(x))  &  =\int_{\mathbb{Q}_{p}^{4}\setminus
E_{V}}|\mathfrak{q}(x)-1|_{p}^{s}\theta(x)d^{4}x+\int_{E_{V}\smallsetminus
V}|\mathfrak{q}(x)-1|_{p}^{s}\theta(x)d^{4}x\\
&  =:(I_{0}\left(  s\right)  ,\theta)+(I_{1}\left(  s\right)  ,\theta).
\end{align*}
A fundamental solution $E_{\mathfrak{q},\alpha}$ for operator $\square
_{\mathfrak{q},\alpha}$ is obtained by computing the Laurent expansion of the
local zeta function $|\mathfrak{q}-1|_{p}^{s}$ at $s=-\alpha$, see
\cite[Theorem 134]{Zuniga-LNM-2016}. Indeed, if
\begin{equation}
|\mathfrak{q}-1|_{p}^{s}=\sum_{j=-j_{0}}^{\infty}c_{j}(s+\alpha)^{j}\text{,
where }c_{j}\in\mathcal{D}_{\mathbb{C}}^{\prime},\text{ with }-j_{0}%
\in\mathbb{Z}\text{,}\label{exapnsinon_local_zeta_function}%
\end{equation}
then $\widehat{E}_{\mathfrak{q},\alpha}=c_{0}$.

\begin{note}
Given two subsets $A$, $B$\ in $\mathbb{Q}_{p}^{4}$, we denote the distance
between them as
\[
dist(A,B):=\inf_{x\in A\text{, }y\in B}\left\Vert x-y\right\Vert _{p}\text{.}%
\]

\end{note}

\begin{lemma}
\label{lemma7}For any $\theta\in\mathcal{D}_{\mathbb{C}}$, the function
$(I_{0}\left(  s\right)  ,\theta)$ is holomorphic in the whole complex plane.
\end{lemma}

\begin{proof}
The result follows, by using a well-known result about the analyticity of
integrals depending on a complex parameter, see \cite[Lemma 5.3.1]{Igusa},
from the fact that there exists a positive constant $\varepsilon
=\varepsilon\left(  \mathfrak{q}\right)  $, such that
\begin{equation}
|\mathfrak{q}(x)-1|_{p}\geq\varepsilon\text{ for any }x\in\mathbb{Q}_{p}%
^{4}\setminus E_{V}\text{.}\label{Ec-Def}%
\end{equation}
If (\ref{Ec-Def}) is false, there exists a sequence $\left\{  y_{n}\right\}
_{n\in\mathbb{N}}$ in $\mathbb{Q}_{p}^{4}\setminus E_{V}$ such that
$|\mathfrak{q}(y_{n})-1|_{p}\rightarrow0$ as $n\rightarrow\infty$, which means
that
\begin{equation}
dist\left(  V,\mathbb{Q}_{p}^{4}\setminus E_{V}\right)
=0,\label{assetion distance}%
\end{equation}
because, since $V$ is compact, there exists $x_{0}\in V$ such that
\[
dist\left(  V,\mathbb{Q}_{p}^{4}\setminus E_{V}\right)  =\inf_{y\in
\mathbb{Q}_{p}^{4}\setminus E_{V}}\left\Vert x_{0}-y\right\Vert _{p}%
=\inf_{y\in\mathbb{Q}_{p}^{4}\setminus E_{V}}dist\left(  V,y\right)  \text{.}%
\]

The assertion (\ref{assetion distance}) is not true. Indeed, since $V$ is
compact and $\mathbb{Q}_{p}^{4}\setminus E_{V}$ is closed (because $E_{V}$ is
open and closed), we have $dist\left(  V,\mathbb{Q}_{p}^{4}\setminus
E_{V}\right)  >0$.
\end{proof}

\begin{remark}
\label{note8}%
Notice the following computation:
\begin{align}
\left(  I_{1}(s),\theta\right)   &  =\int_{E_{V}\backslash V}|\mathfrak{q}(x)-1|_{p}%
^{s}\theta(x)d^{4}x=\underset{\underset{\mathfrak{q}(\overline{b}%
)\equiv1\text{ }\operatorname{mod}\text{ }p}{\overline{b}\in\mathbb{F}_{p}%
^{4}}}{\sum}\int_{b+p\mathbb{Z}_{p}^{4}}|\mathfrak{q}(x)-1|_{p}^{s}%
\theta(x)d^{4}x\nonumber\\
&  =p^{-4}\underset{\underset{\mathfrak{q}(\overline{b})\equiv1\text{
}\operatorname{mod}\text{ }p}{\overline{b}\in\mathbb{F}_{p}^{4}}}{\sum}%
\int_{\mathbb{Z}_{p}^{4}}|\mathfrak{q}(b+pz)-1|_{p}^{s}\theta(b+pz)d^{4}%
z\nonumber\\
&  =:p^{-4}\underset{\underset{\mathfrak{q}(\overline{b})\equiv1\text{
}\operatorname{mod}\text{ }p}{\overline{b}\in\mathbb{F}_{p}^{4}}}{\sum}\left(
I_{b}(s),\theta\right)  .\label{Eq10}%
\end{align}

\end{remark}

\begin{lemma}
\label{lemma9}With the above notations and setting%
\[
I_{b}(s)=\sum_{j=0}^{\infty}c_{j}\left(  I_{b},\alpha\right)
(s+\alpha)^{j}\text{, where }c_{j}\left(  I_{b},\alpha\right)  \in
\mathcal{D}_{\mathbb{C}}^{\prime},
\]
for $\overline{b}\in\mathbb{F}_{p}^{4}$, $\mathfrak{q}(\overline{b})\equiv1 $
$\operatorname{mod}$ $p$, the coefficient $c_{0}\in\mathcal{D}_{\mathbb{C}%
}^{\prime}$ in expansion (\ref{exapnsinon_local_zeta_function}) is given by%
\[
\left(  c_{0},\theta\right)  =\int_{\mathbb{Q}_{p}^{4}\setminus E_{V}%
}|\mathfrak{q}(x)-1|_{p}^{-\alpha}\theta(x)d^{4}x+p^{-4}\underset
{\underset{\mathfrak{q}(\overline{b})\equiv1\text{ }\operatorname{mod}\text{
}p}{\overline{b}\in\mathbb{F}_{p}^{4}}}{\sum}\left(  c_{0}\left(  I_{b}%
,\alpha\right)  ,\theta\right)  .
\]

\end{lemma}

\begin{proof}
The formula follows from Lemma \ref{lemma7} and Remark \ref{note8}.
\end{proof}

We now compute the coefficients $c_{0}\left(  I_{b},\alpha\right)  $ for some
$b$s, the calculation of the missing cases is similar to the one presented here.

\begin{lemma}
\label{lemma10}Assume that $\overline{b}_{0}\not \equiv 0$ $\operatorname{mod}%
$ $p$. If $\alpha\neq1$, then
\[
\left(  c_{0}\left(  I_{b},\alpha\right)  ,\theta\right)  =p^{\alpha}%
\int_{\mathbb{Z}_{p}}|u_{0}|_{p}^{-\alpha}(\Theta_{b}(u_{0})-\Theta
_{b}(0))du_{0}+\frac{p^{\alpha}(1-p^{-1})}{1-p^{-1+\alpha}}\Theta_{b}(0),
\]
where $\Theta_{b}=T_{I_{b},\alpha}\left(  \theta\right)  \in\mathcal{D}%
_{\mathbb{C}}\left(  \mathbb{Q}_{p}\right)  $, and $T_{I_{b},\alpha}$ is a
linear operator from $\mathcal{D}_{\mathbb{C}}\left(  \mathbb{Q}_{p}%
^{4}\right)  $ into $\mathcal{D}_{\mathbb{C}}\left(  \mathbb{Q}_{p}\right)  $,
and
\[
\Theta_{b}\left(  0\right)  =p^{3}(\delta(\mathfrak{q}(k)-1),\theta
(k)\Omega(p\Vert k-b\Vert_{p})).
\]
In addition,
\begin{equation}
\left(  {\LARGE 1}_{V}c_{0}\left(  I_{b}, \alpha\right)  ,\theta\right)
=\frac{p^{\alpha}(1-p^{-1})}{1-p^{-1+\alpha}}\Theta_{b}(0).\label{for_rest_1}%
\end{equation}
If $\alpha=1$, then
\[
\left(  c_{0}\left(  I_{b},1\right)  ,\theta\right)  =p\int_{\mathbb{Z}%
_{p}^{4}}|u_{0}|_{p}^{-1}(\Theta_{b}(u_{0})-\Theta_{b}(0))du_{0}-\dfrac
{p-1}{2}\Theta_{b}(0).
\]
Moreover,
\begin{equation}
\left(  {\LARGE 1}_{V}c_{0}\left(  I_{b},1\right)  ,\theta\right)
=-\dfrac{p-1}{2}\Theta_{b}(0).\label{for_rest_2}%
\end{equation}

\end{lemma}

\begin{proof}
By changing variables as $u=F(z)$, see (\ref{change_of_var}), we get
\begin{multline*}
\left(  I_{b}(s),\theta\right)  =\int_{\mathbb{Z}_{p}^{4}}\left\vert
\mathfrak{q}\left(  b+pz\right)  -1\right\vert _{p}^{s}\theta\left(
b+pz\right)  d^{4}z\\
=p^{-s}\int_{\mathbb{Z}_{p}^{4}}|u_{0}|_{p}^{s}\theta(b_{0}+pf(u_{0}%
,\ldots,u_{3}),b_{1}+pu_{1},b_{2}+pu_{2},b_{3}+pu_{3})du_{0}du_{1}du_{2}du_{3}%
\end{multline*}
where $f(u_{0},\ldots,u_{3})$ is a $p-$adic analytic function on
$\mathbb{Z}_{p}^{4}$. Set
\begin{multline*}
\Theta_{b}(u_{0}):=\\
\int\limits_{\mathbb{Z}_{p}^{3}\setminus D}\theta(b_{0}+pf(u_{0},\ldots
,u_{3}),b_{1}+pu_{1},b_{2}+pu_{2},b_{3}+pu_{3})du_{1}du_{2}du_{3},
\end{multline*}
where $D=\{b_{0}+pf(u_{0},\ldots,u_{3})=0\}$. Then $\Theta_{b}(u_{0}%
)\in\mathcal{D}_{\mathbb{C}}\left(  \mathbb{Q}_{p}\right)  $ and $\Theta
_{b}(0)=p^{3}(\Omega(p\Vert k-b\Vert_{p})\delta(\mathfrak{q}(k)-1),\theta
(k))$, see (\ref{ec19}). Notice that for $u_{0}$ given, the set
$\{b_{0}+pf(u_{0},\ldots,u_{3})=0\}$ has measure zero, and that $b_{0}%
+pf(u_{0},\ldots,u_{3})$ is locally constant in $u_{0}$ on $\mathbb{Z}_{p}%
^{3}\setminus D$, this last fact is verified by using the $p-$adic Taylor
expansion, see e.g. \cite{Serre}. Therefore
\begin{equation}
\left(  I_{b}(s),\theta\right)  =p^{-s}\int_{\mathbb{Z}_{p}^{4}}|u_{0}%
|_{p}^{s}(\Theta_{b}(u_{0})-\Theta_{b}(0))du_{0}+\frac{p^{-s}(1-p^{-1}%
)}{1-p^{-1-s}}\Theta_{b}(0).\label{formula_int}%
\end{equation}
If $\alpha\neq1$, then $\left(  c_{0}\left(  I_{b},\alpha\right)
,\theta\right)  $ is obtained by replacing $s=-\alpha$ in (\ref{formula_int}).
In the case $\alpha=1$, the computation of $\left(  c_{0}\left(
I_{b},1\right)  ,\theta\right)  $ is achieved by computing the Laurent
expansion of $\left(  I_{b}(s),\theta\right)  $ around $(s+1)$, which follows
from the formula:
\begin{equation}
\dfrac{p^{-s}(1-p^{-1})}{1-p^{-1-s}}=\left(  \dfrac{p-1}{\ln p}\right)
\dfrac{1}{s+1}-\dfrac{p-1}{2}+O(s+1),\nonumber
\end{equation}
where $O(s+1)$ denotes a holomorphic function. Finally formulae
(\ref{for_rest_1})-(\ref{for_rest_2}) follow from the fact that in the
coordinate system $(u_{0},\ldots,u_{3})$, $u_{0}=0$ is a local equation of $V$.
\end{proof}

\begin{remark}
\label{note11}Lemma \ref{lemma10} is valid for general $b$, but there are
small variations in the formulae for the $c_{0}\left(  I_{b},\alpha\right)
$s. In the case $\overline{b}_{0}\equiv0$ $\operatorname{mod}$ $p$,
$\overline{b}_{1}\not \equiv 0$ $\operatorname{mod}$ $p$, the statement of
Lemma \ref{lemma10}\ and the corresponding proof are similar to ones presented
here, see Remark \ref{note6}. We outline the calculations for the case
$\overline{b}_{0}\equiv0$ $\operatorname{mod}$ $p$, $\overline{b}_{1}\equiv0$
$\operatorname{mod}$ $p$, $\overline{b}_{2}\not \equiv 0$ $\operatorname{mod}$
$p$. In this case, we use the following change of variables:
\[
u=G(z)\text{ with }%
\begin{array}
[c]{cccc}%
u_{0}=z_{0}, & u_{1}=z_{1}, & u_{2}=\dfrac{\mathfrak{q}(b+pz)-1}{p^{2}}, &
u_{3}=z_{3}.
\end{array}
\]
Then
\[
Jac_{G}(z)=\det\left[
\begin{array}
[c]{cccc}%
1 & 0 & 0 & 0\\
0 & 1 & 0 & 0\\
\frac{1}{p^{2}}\frac{\partial u_{0}}{\partial z_{0}} & \frac{1}{p^{2}}%
\frac{\partial u_{1}}{\partial z_{1}} & \frac{1}{p^{2}}\frac{\partial u_{2}%
}{\partial z_{2}} & \frac{1}{p^{2}}\frac{\partial u_{3}}{\partial z_{3}}\\
0 & 0 & 0 & 1
\end{array}
\right]  =\frac{1}{p^{2}}\frac{\partial u_{2}}{\partial z_{2}}=-2\left(
b_{2}+pz_{2}\right)  ,
\]
and thus $Jac_{G}(z)\equiv\overline{-2b_{2}}\equiv\overline{b_{2}}%
\not \equiv 0$ $\operatorname{mod}$ $p$, and by Lemma 7.4.3 \ in \cite{Igusa},
$G$ gives rise to an analytic isomorphism from $\mathbb{Z}_{p}^{4}$ to itself
which preserves the Haar measure. By changing variables in integral $\left(
I_{b}(s),\theta\right)  $, we get that%
\begin{multline*}
\left(  I_{b}(s),\theta\right)  =\\
p^{-2s}\int_{\mathbb{Z}_{p}^{4}}|u_{2}|_{p}^{s}\theta(b_{0}+pz_{0}%
,b_{1}+pu_{1},b_{2}+ph(u_{0},\ldots,u_{3}),b_{3}+pu_{3})du_{0}du_{1}%
du_{2}du_{3}.
\end{multline*}
Now the calculations proceed as in the proof of Lemma \ref{lemma10} .
\end{remark}

\begin{remark}
\label{note12} Set $\delta_{k}(x):=p^{4k}\Omega(p^{k}\Vert x\Vert_{p})$. We
recall the definition of the product of two distributions: given
$F,G\in\mathcal{D}_{\mathbb{C}}^{\prime}$, their product is defined as
$(F\cdot G,\varphi)=\lim_{k\rightarrow\infty}(G,(F\ast\delta_{k})\varphi)$, if
the limit exist for all $\varphi\in\mathcal{D}_{\mathbb{C}}$. If the product
$F\cdot G$ exists then the product $G\cdot F$ exists and they are equal.
\end{remark}

\begin{lemma}
\label{lemma12} $(|\mathfrak{q}-1|_{p}^{\alpha}\delta(\mathfrak{q}-1),\psi)=0
$ for any $\psi\in\mathcal{D}_{\mathbb{C}}$ and for any $\alpha>0 $.
\end{lemma}

\begin{proof}
By Remark \ref{note12}, $(|\mathfrak{q}-1|_{p}^{\alpha}\delta(\mathfrak{q}%
-1),\psi)=\lim_{k\rightarrow\infty}(\delta(\mathfrak{q}-1),(|\mathfrak{q}%
-1|_{p}^{\alpha}\ast\delta_{k})\psi)$. Now
\[
(|\mathfrak{q}-1|_{p}^{\alpha}\ast\delta_{k})(x)=p^{4k}\int_{x+p^{k}%
\mathbb{Z}^{4}_p}|\mathfrak{q}(y)-1|_{p}^{\alpha}d^{4}y.
\]
Since $V\subseteq\mathbb{Z}_{p}^{4}$ has measure zero, we may assume without
loss of generality that $x\notin V$. Now, if $z\in\mathbb{Z}_{p}^{4} $ then
$\mathfrak{q}(x+p^{k}z)-1=\mathfrak{q}(x)-1+p^{k}A$, with $A\in\mathbb{Z}_{p}$
and $\mathfrak{q}(x)-1\neq0$, then by taking $k$ sufficiently large, we have
$|\mathfrak{q}(x+p^{k}z)-1|_{p}^{\alpha}=|\mathfrak{q}(x)-1|_{p}^{\alpha}$,
consequently $(|\mathfrak{q}-1|_{p}^{\alpha}\ast\delta_{k})(x)=|q(x)-1|_{p}%
^{\alpha}$ for $k$ sufficiently large. Finally, $(|\mathfrak{q}-1|_{p}%
^{\alpha}\delta(\mathfrak{q}-1),\psi)=(\delta(\mathfrak{q}-1),|\mathfrak{q}%
-1|_{p}^{\alpha}\psi)=0$ because supp $\delta(\mathfrak{q}-1)=V$.
\end{proof}

\begin{remark}
\label{note13}For any locally constant function $h$, it holds that
$h|\mathfrak{q}-1|_{p}^{\alpha}\delta(\mathfrak{q}-1)\in\mathcal{D}%
_{\mathbb{C}}^{\prime}$, see e.g. \cite[p. 126, Proposition 3.16]{Taibleson}.
Then $(h|\mathfrak{q}-1|_{p}^{\alpha}\delta(\mathfrak{q}-1),\psi
)=(|\mathfrak{q}-1|_{p}^{\alpha}\delta(\mathfrak{q}-1),h\psi)=(\delta
(\mathfrak{q}-1),|\mathfrak{q}-1|_{p}^{\alpha}h\psi)=0$ for any $\psi
\in\mathcal{D}_{\mathbb{C}}$.
\end{remark}

\begin{remark}\label{note14}
Let us make some comments about orthogonal invariance in this setting.
\begin{enumerate}[(i)]
\item Let $\varphi\in\mathcal{D}_{\mathbb{C}}$ and let
$T\in\mathcal{D}_{\mathbb{C}}^{\prime}$. We define the action of $\Lambda
\in\boldsymbol{O}\left(  \mathfrak{q}\right)  $, by putting%
\[
\left(  \Lambda\varphi\right)  \left(  x\right)  =\varphi\left(  \Lambda
^{-1}x\right)  ,
\]
and the action of $\Lambda$\ on $T$, by putting%
\[
\left(  \Lambda T,\varphi\right)  =\left(  T,\Lambda^{-1}\varphi\right)  .
\]
We say that $T$ is invariant under $\boldsymbol{O}\left(  \mathfrak{q}\right)
$, if $\Lambda T=T$ for any $\Lambda\in\boldsymbol{O}\left(  \mathfrak{q}%
\right)  $.

\item $T$ is invariant under $\boldsymbol{O}\left(  \mathfrak{q}\right)
\Leftrightarrow\widehat{T}$ is invariant \ under $\boldsymbol{O}\left(
\mathfrak{q}\right)  $. We first notice that by using $\mathcal{B}\left(
\Lambda^{-1}y,\Lambda^{-1}k\right)  =\mathcal{B}\left(  y,k\right)  $ for any
$\Lambda\in\boldsymbol{O}\left(  \mathfrak{q}\right)  $, we have
\begin{gather*}
\left(  \widehat{\Lambda^{-1}\varphi}\right)  \left(  k\right)  =\int
_{\mathbb{Q}_{p}^{4}}\chi_{p}\left(  \mathcal{B}\left(  x,k\right)  \right)
\left(  \Lambda^{-1}\varphi\right)  \left(  x\right)  d\mu\left(  x\right) \\
=\int_{\mathbb{Q}_{p}^{4}}\chi_{p}\left(  \mathcal{B}\left(  x,k\right)
\right)  \varphi\left(  \Lambda x\right)  d\mu\left(  x\right)  =\int
_{\mathbb{Q}_{p}^{4}}\chi_{p}\left(  \mathcal{B}\left(  \Lambda^{-1}%
y,\Lambda^{-1}\left(  \Lambda k\right)  \right)  \right)  \varphi\left(
y\right)  d\mu\left(  y\right) \\
=\int_{\mathbb{Q}_{p}^{4}}\chi_{p}\left(  \mathcal{B}\left(  y,  \Lambda
k  \right)  \right)  \varphi\left(  y\right)  d\mu\left(  y\right)
=\widehat{\varphi}\left(  \Lambda k\right)  ,
\end{gather*}
i.e. $\left(  \widehat{\Lambda^{-1}\varphi}\right)  =\Lambda^{-1}%
\widehat{\varphi}$. Now, assuming that $\Lambda T=T$ for any $\Lambda
\in\boldsymbol{O}\left(  \mathfrak{q}\right)  $, we have%
\begin{align*}
\left(  \Lambda\widehat{T},\varphi\right)   & =\left(  \widehat{T}%
,\Lambda^{-1}\varphi\right)  =\left(  T,\widehat{\Lambda^{-1}\varphi}\right)
=\left(  T,\Lambda^{-1}\widehat{\varphi}\right)  =\left(  \Lambda
T,\widehat{\varphi}\right) \\
& =\left(  T,\widehat{\varphi}\right)  =(\widehat{T},\varphi).
\end{align*}
Here, it is worth to mention that our definition of Fourier transform using
the bilinear form $\mathcal{B}$ plays a crucial role.

\item By a result of Rallis-Schiffman, the distribution $\delta(\mathfrak{q}%
-1)$ is the unique (up to multiplication by complex constants) distribution
supported on V invariant under $\boldsymbol{O}\left(  \mathfrak{q}\right)  $,
\cite{Rallis-Schiffman}.
\end{enumerate}
\end{remark}

\begin{theorem}
\label{Theorem1}There exist fundamental solutions $E_{\mathfrak{q},\alpha}$
for operators $\square_{\mathfrak{q},\alpha}$ which are invariant under the
action of $\boldsymbol{O}\left(  \mathfrak{q}\right)  $. Furthermore, the
distributions $E_{\mathfrak{q},\alpha}$ satisfy the following:
\begin{enumerate}[(i)]
\item\mbox{}
\begin{equation}
\mathcal{F}(E_{\mathfrak{q},\alpha})=\mathcal{F}
(E_{\mathfrak{q},\alpha}^{0})+C\delta(\mathfrak{q}-1),\label{eq_fun_sol_1}%
\end{equation}
where $C$ is a non-zero complex constant and $\mathcal{F}(E_{\mathfrak{q}%
,\alpha}^{0})$, $\delta(\mathfrak{q}-1)$ are distributions invariant under
$\boldsymbol{O}(\mathfrak{q})$.
\item\mbox{}
\begin{equation}
1_{V}\mathcal{F}(E_{\mathfrak{q},\alpha})=C\delta(\mathfrak{q}%
-1).\label{eq_fun_sol_2}%
\end{equation}
In particular, the restriction of $\mathcal{F}(E_{\mathfrak{q},\alpha})$ to
$V$ is unique up to multiplication for a non-zero complex constant.
\end{enumerate}
\end{theorem}

\begin{proof}
The existence of fundamental solutions for operators $\square_{\mathfrak{q}%
,\alpha}$ is guaranteed by Theorem 134 in \cite{Zuniga-LNM-2016}. If
$E_{\mathfrak{q},\alpha}^{0}$ is a fundamental solution for $\square
_{\mathfrak{q},\alpha}$, then, by Lemmas \ref{LemaFS}, \ref{lemma12},
$E_{\mathfrak{q},\alpha}^{0}+C\mathcal{F}^{-1}\left[  \delta(\mathfrak{q}%
-1)\right]  $ is also a fundamental solution for any non-zero complex constant
$C $. Therefore, the Fourier transform of any fundamental solution may be
written as%
\begin{equation}\label{eqaux1}
\mathcal{F}\left[  E_{\mathfrak{q},\alpha}\right]  =\mathcal{F}\left[
E_{\mathfrak{q},\alpha}^{0}\right]  +C\delta(\mathfrak{q}-1),
\end{equation}
for some fundamental solution $E_{\mathfrak{q},\alpha}^{0}$ and some non-zero
complex constant $C$.

\begin{remark}
In fact, if there is another fundamental solution $E'_{\mathfrak{q},\alpha}$ of
$\square_{\mathfrak{q},\alpha}$, invariant under $\boldsymbol{O}\left(\mathfrak{q}\right)$,
satisfying
\begin{equation}\label{eqaux2}
\mathcal{F}\left[ E_{\mathfrak{q},\alpha}\right]=
\mathcal{F}\left[ E'_{\mathfrak{q},\alpha}\right] +C\delta (\mathfrak{q}-1)\,,
\end{equation}
then from \eqref{eqaux1} and \eqref{eqaux2} we get that
$\mathcal{F}\left[ E'_{\mathfrak{q},\alpha}-E_{\mathfrak{q},\alpha}^{0}\right]$
is a distribution supported on $V$ and invariant under $\boldsymbol{O}\left(\mathfrak{q}\right)$,
and consequently 
$\mathcal{F}\left[ E'_{\mathfrak{q},\alpha}-E_{\mathfrak{q},\alpha}^{0}\right]=C_0\delta (\mathfrak{q}-1)\,,$
for some constant $C_0$.
\end{remark} 

By Lemmas \ref{lemma9}, \ref{lemma10}\ and Remark
\ref{note11}, there exists a fundamental solution $E_{\mathfrak{q},\alpha}%
^{0}$, such that $\mathcal{F}\left[  E_{\mathfrak{q},\alpha}^{0}\right]  $ is
a linear combination of distributions of any of the types
\[
\int_{\mathbb{Q}^4_p\backslash E_V} |\mathfrak{q}(x)-1|^{-\alpha}_p\theta (x)\, d^4x\ \mbox{ or }\
p^{\alpha}\int_{\mathbb{Z}_{p}}|u_{0}|_{p}^{-\alpha}(\Theta_{b}%
(u_{0})-\Theta_{b}(0))\,du_{0},
\]
with $\Theta_{b}(u_{0})$ defined as in Lemma \ref{lemma10}. In addition, we
have%
\[
{\LARGE 1}_{V}\mathcal{F}\left[  E_{\mathfrak{q},\alpha}^{0}\right]  =0\text{
in }\mathcal{D}_{\mathbb{C}}^{\prime}(\mathbb{Q}_{p}^{4})\text{.}%
\]
The rest of assertions announced follows from Remark \ref{note14} by the
following assertion:

\textbf{Claim.} The distribution $E_{\mathfrak{q},\alpha}^{0}$ is invariant
under $\boldsymbol{O}\left(  \mathfrak{q}\right)  $.

We first note that
\begin{equation}
\Lambda|\mathfrak{q}-1|_{p}^{s}=|\mathfrak{q}-1|_{p}^{s}\text{ for any
}\Lambda\in\boldsymbol{O}\left(  \mathfrak{q}\right)  \text{, and
\ }\operatorname{Re}(s)>0\text{,}\label{formula_1}%
\end{equation}
because $\mathfrak{q}\left(  \Lambda^{-1}y\right)  =\mathfrak{q}\left(
y\right)  $ for any $\Lambda\in\boldsymbol{O}\left(  \mathfrak{q}\right)  $,
and any $y\in\mathbb{Q}_{p}^{4}$. Now, we rewrite (\ref{formula_1}) as
\[
\left(  |\mathfrak{q}-1|_{p}^{s},\Lambda^{-1}\varphi\right)  =\left(
|\mathfrak{q}-1|_{p}^{s},\varphi\right)  \text{ for }\Lambda\in\boldsymbol{O}%
\left(  \mathfrak{q}\right)  \text{, }\varphi\in\mathcal{D}_{\mathbb{C}%
}\text{, and \ }\operatorname{Re}(s)>0\text{,}%
\]
and use that $\Lambda^{-1}\varphi\in\mathcal{D}_{\mathbb{C}}$ for $\varphi
\in\mathcal{D}_{\mathbb{C}}$, and that the distribution $|\mathfrak{q}-1|_{p}^{s}
$ admits a meromorphic continuation to the whole complex plane to conclude
that (\ref{formula_1}) is valid for any $s$. We now recall that $\mathcal{F}%
\left[  E^0_{q,\alpha}\right]  =c_{0}\in\mathcal{D}'_\mathbb{C}$, where%
\begin{align*}
\left(  |\mathfrak{q}-1|_{p}^{s},\varphi\right)   &  =\sum_{j=-j_{0}}^{\infty
}\left(  c_{j},\varphi\right)  (s+\alpha)^{j}=\left(  \Lambda|\mathfrak{q}%
-1|_{p}^{s},\varphi\right)  =\left(  |\mathfrak{q}-1|_{p}^{s},\Lambda
^{-1}\varphi\right) \\
&  =\sum_{j=-j_{0}}^{\infty}\left(  c_{j},\Lambda^{-1}\varphi\right)
(s+\alpha)^{j},
\end{align*}
then $\left(  c_{0},\varphi\right)  =\left(  c_{0},\Lambda^{-1}\varphi\right)
$, which implies that $c_{0}$\ is invariant under $\boldsymbol{O}\left(
\mathfrak{q}\right)  $, and consequently, $E^0_{q,\alpha}$ is invariant under
$\boldsymbol{O}\left(  \mathfrak{q}\right)  $.
\end{proof}

\section{Klein-Gordon type operators acting on $\mathcal{H}_\infty$}

\begin{lemma}
\label{lemma14} Let $\mathfrak{f}({k})\in\mathbb{Q}_{p}[k_{0},k_{1}%
,k_{2},k_{3}]$ be a non-constant homogeneous polynomial of degree $e$ and
$\alpha>0$. Then there exists a positive constant $A=A(\mathfrak{f},\alpha)$
such that
\[
|\mathfrak{f}({k})-1|_{p}^{\alpha}\leq A[k]_{p}^{e\alpha}\text{ \ for }%
k\in\mathbb{Q}_{p}^{4}.
\]

\end{lemma}

\begin{proof}
We first note that $|\mathfrak{f}({k})-1|_{p}^{\alpha}\leq\left[
\max\{\left\vert \mathfrak{f}({k})\right\vert _{p},1\}\right]  ^{\alpha}$. We
now use that $|\mathfrak{f}({k})|_{p}\leq C\left(  \mathfrak{f}\right)
[k]_{p}^{e}$ for $k\in\mathbb{Q}_{p}^{4}$, to obtain%
\begin{align*}
|\mathfrak{f}({k})-1|_{p}^{\alpha}  &  \leq\left[  \max\{C\left(
\mathfrak{f}\right)  [k]_{p}^{e},1\}\right]  ^{\alpha}\leq\left[
\max\{C\left(  \mathfrak{f}\right)  ,1\}\right]  ^{\alpha}\left[  \max\left[
[k]_{p}^{e},1\right]  \right]  ^{\alpha}\\
&  =A[k]_{p}^{e\alpha}.
\end{align*}

\end{proof}

\begin{remark}
For $\alpha\in\mathbb{R}$, we set $\lceil\alpha\rceil:=\min\{\gamma
\in\mathbb{Z};\gamma\geq\alpha\}$, the ceiling function.
\end{remark}

\begin{lemma}
\label{lemma15}The mapping%
\[%
\begin{array}
[c]{cccc}%
\square_{\mathfrak{q},\alpha}: & \mathcal{H}_{\infty}(\mathbb{K}) &
\rightarrow & \mathcal{H}_{\infty}(\mathbb{K)}\\
&  &  & \\
& h & \rightarrow & \square_{\mathfrak{q},\alpha}h
\end{array}
\]
is a well-defined continuous linear operator between locally convex spaces.
\end{lemma}

\begin{proof}
Take $\mathbb{K}=\mathbb{C}$. Let us first prove that $\square_{\mathfrak{q}%
,\alpha}$ is a well-defined linear operator. Let $h\in\mathcal{H}%
_{l+\lceil4\alpha\rceil}(\mathbb{C})$, then by the Lemma \ref{lemma14}, with
$e=2$, we have
\begin{align}
\Vert\square_{\mathfrak{q},\alpha}h\Vert_{l}^{2}  &  =\int_{\mathbb{Q}_{p}%
^{4}}[\xi]_{p}^{l}|\widehat{(\square_{\mathfrak{q},\alpha}h)}(k)|^{2}%
d^{4}k=\int_{\mathbb{Q}_{p}^{4}}[\xi]_{p}^{l}|\mathfrak{q}(k)-1|_{p}^{2\alpha
}|\widehat{h}(k)|^{2}d^{4}k\nonumber\\
&  \leq C\int_{\mathbb{Q}_{p}^{4}}[\xi]_{p}^{l+4\alpha}|\widehat{h}%
(k)|^{2}d^{4}k\leq C\int_{\mathbb{Q}_{p}^{4}}[\xi]_{p}^{l+\lceil4\alpha\rceil
}|\widehat{h}(k)|^{2}d^{4}k=C\Vert h\Vert_{l+\lceil4\alpha\rceil}%
^{2}.\nonumber
\end{align}
\ By Lemma \ref{lemma4}-(i), $\square_{\mathfrak{q},\alpha}h\in\mathcal{H}%
_{l}(\mathbb{C})$, i.e. $\square_{\mathfrak{q},\alpha}$ is a well-defined,
linear, and continuous operator from $\mathcal{H}_{l+\lceil4\alpha\rceil
}(\mathbb{C})$ into $\mathcal{H}_{l}(\mathbb{C})$ for any $l\in\mathbb{N}$. In
turn, this implies that $\square_{\mathfrak{q},\alpha}$ is a well-defined
linear operator from $\mathcal{H}_{\infty}(\mathbb{C})$ into $\mathcal{H}%
_{\infty}(\mathbb{C})$. To establish the continuity, we use the fact that
$(\mathcal{H}_{\infty}(\mathbb{C}),d)$\ is a metric space. Take a sequence
$\{\varphi_{n}\}_{n\in\mathbb{N}}\subset\mathcal{H}_{\infty}(\mathbb{C})$ such
that $\varphi_{n}\overset{d}{\rightarrow}\varphi$, with $\varphi\in
\mathcal{H}_{\infty}(\mathbb{C})$, which is equivalent to say that
$\varphi_{n}\overset{\Vert\cdot\Vert_{r}}{\rightarrow}\varphi$, for all
$r\in\mathbb{N}$. Take $l\in\mathbb{N}$ and $\varphi$, $\varphi_{n}%
\in\mathcal{H}_{l+\lceil4\alpha\rceil}(\mathbb{C})$, then by the continuity of
$\square_{\mathfrak{q},\alpha}:\mathcal{H}_{l+\lceil4\alpha\rceil}%
(\mathbb{C})\rightarrow\mathcal{H}_{l}(\mathbb{C})$, we have $\square
_{\mathfrak{q},\alpha}\varphi_{n}\overset{\Vert\cdot\Vert_{l}}{\rightarrow
}\square_{\mathfrak{q},\alpha}\varphi$, and since $l$ is arbitrary in
$\mathbb{N}$, we conclude that $\square_{\mathfrak{q},\alpha}\varphi
_{n}\overset{d}{\rightarrow}\square_{\mathfrak{q},\alpha}\varphi$.

We know turn to the case $\mathbb{K}=\mathbb{R}$. Since $(\square
_{\mathfrak{q},\alpha}\varphi)(x)=\overline{(\square_{\mathfrak{q},\alpha
}\varphi)(x)}$ for $\varphi\in\mathcal{H}_{\infty}(\mathbb{R})$, the statement
is also valid in $\mathcal{H}_{\infty}(\mathbb{R})$.
\end{proof}

\begin{remark}
The preceding lemma remains valid if we replace $|\mathfrak{q}(k)-1|^\alpha_p$ by
$g\left( [k]_p\right) |\mathfrak{q}(k)-1|^\alpha_p$, where $g:\mathbb{R}_+\to\mathbb{C}$
is any continuous function.
\end{remark}

\begin{remark}
We recall that $V$ is a $p-$adic compact submanifold of $\mathbb{Z}_{p}^{4}$
of codimension one. We denote by $d\lambda$ the measure corresponding to the
distribution $\delta\left(  \mathfrak{q}-1\right)  $ as before. Then $\left(
V,\mathcal{B}(V),d\lambda\right)  $ is a measure space, where $\mathcal{B}%
(V)$\ is the Borel $\sigma$-algebra generated by the open compact subsets of
$V$, and thus the space $L_{\mathbb{K}}^{2}\left(  V,d\lambda\right)  $ is well-defined.
\end{remark}

\begin{proposition}
\label{Prop-1}The
mapping%
\[%
\begin{array}
[c]{cccc}%
R: & \mathcal{H}_{l}(\mathbb{C}) & \rightarrow & L_{\mathbb{C}}^{2}\left(
V^{+},d\lambda\right) \\
&  &  & \\
& f & \rightarrow & \widehat{f}\mid_{_{V^{+}}}%
\end{array}
\]
determines a well-defined operator satisfying%
\begin{equation}
\left\Vert R(f)\right\Vert _{L_{\mathbb{C}}^{2}\left(  V^{+},d\lambda\right)
}\leq C\left\Vert f\right\Vert _{l}\label{Eq21}%
\end{equation}
for any $l\in\mathbb{N}$. Consequently, $R$ induces a continuous
operator from $\mathcal{H}_{\infty}(\mathbb{C})$ into $L_{\mathbb{C}}%
^{2}\left(  V^{+},d\lambda\right)  $.
\end{proposition}

\begin{proof}
Since $\mathcal{D}_{\mathbb{C}}$ is dense in $\mathcal{H}_{l}(\mathbb{C})$ for
any $l\in\mathbb{N}$, in order to prove (\ref{Eq21}) we may assume without
loss of generality that $f\in$ $\mathcal{D}_{\mathbb{C}}$ and that
$\widehat{f}\mid_{V^{+}}$ is not the constant function zero. Notice that%
\begin{equation}
\left\Vert R(f)\right\Vert _{L_{\mathbb{C}}^{2}\left(  V^{+},d\lambda\right)
}^{2}=\int_{U_{\mathfrak{q}}}\left\vert \widehat{f}\left(  \sqrt{\omega\left(
\boldsymbol{k}\right)  },\boldsymbol{k}\right)  \right\vert ^{2}\frac
{d^{3}\boldsymbol{k}}{\left\vert \sqrt{\omega\left(  \boldsymbol{k}\right)
}\right\vert _{p}},\label{Eq22}%
\end{equation}
where $\left\vert \sqrt{\omega\left(  \boldsymbol{k}\right)  }\right\vert
_{p}\neq0$, cf. Remark \ref{note_q}. For $m\in\mathbb{Q}_{p}^{\times}$, we
set
\[
V_{m}=\left\{  \left(  k_{0},\boldsymbol{k}\right)  \in\mathbb{Q}_{p}%
^{4};\mathfrak{q}\left(  k_{0},\boldsymbol{k}\right)  =m\right\}  .
\]
We recall that $\mathfrak{q}\left(  k_{0},\boldsymbol{k}\right)  =k_{0}%
^{2}-\mathfrak{q}_{0}\left(  \boldsymbol{k}\right)  $. Then $V_{m}$ is a
$p-$adic compact submanifold of $\mathbb{Q}_{p}^{4}$ of codimension one. In
the case in which $V_{m}\neq\emptyset$, we denote by $d\lambda\left(
m\right)  $ the measure on $V_{m}$ induced by the Gel'fand-Leray form on
$V_{m}$. Then $dk_{0}d^{3}\boldsymbol{k}=d\lambda\left(  m\right)  dm$, where
$dm$ is the normalized Haar measure of $\mathbb{Q}_{p}$.

\textbf{Claim C. }For $\widehat{f}\left(  k_{0},\boldsymbol{k}\right)
\in\mathcal{D}_{\mathbb{C}}$, the $\mathbb{R}$-valued function defined by%
\[
\int_{\mathbb{Q}_{p}^{\times}}\int_{V_{m}}\left\vert \widehat{f}\left(
k_{0},\boldsymbol{k}\right)  \right\vert ^{2}d\lambda\left(  m\right)  dm
\]
is in $\mathcal{D}_{\mathbb{R}}\left(  \mathbb{Q}_{p}\right)  $ and
\textbf{\ }
\begin{equation}
\left\Vert f\right\Vert _{0}^{2}=\int_{\mathbb{Q}_{p}^{4}}\left\vert
\widehat{f}\left(  k_{0},\boldsymbol{k}\right)  \right\vert ^{2}dk_{0}%
d^{3}\boldsymbol{k=}\int_{\mathbb{Q}_{p}^{\times}}\int_{V_{m}}\left\vert
\widehat{f}\left(  k_{0},\boldsymbol{k}\right)  \right\vert ^{2}%
d\lambda\left(  m\right)  dm.\label{Eq23A}%
\end{equation}
This claim is a very particular version of a general theorem on integration over
the fibers in the framework of $p-$adic manifolds, see \cite[Theorem
7.6.1]{Igusa}.

\textbf{Claim D. }There exists a positive constant $C_{0}$ such that%
\begin{equation}
\left\Vert f\right\Vert _{0}^{2}\geq C_{0}\int_{U_{\mathfrak{q}}}\left\vert
\widehat{f}\left(  \sqrt{\omega\left(  \boldsymbol{k}\right)  },\boldsymbol{k}%
\right)  \right\vert ^{2}\frac{d^{3}\boldsymbol{k}}{\left\vert \sqrt
{\omega\left(  \boldsymbol{k}\right)  }\right\vert _{p}}.\label{Eq23}%
\end{equation}
Estimation (\ref{Eq21}) follows from (\ref{Eq22})-(\ref{Eq23}). The fact that
operator $R$ extends \ to $\mathcal{H}_{\infty}(\mathbb{C})$\ follows from
(\ref{Eq21}), by using \ a classical argument based on convergence of
sequences due to the fact that the topology of $\mathcal{H}_{\infty
}(\mathbb{C})$ is metrizable.

\textbf{Proof of Claim D. }In order to prove the Claim we proceed as follows.
We set $G_{M}:=1+p^{M}\mathbb{Z}_{p}$, for $M\geq1$. Then $G_{M}$ is a
multiplicative subgroup of the group of squares of $\mathbb{Q}_{p}^{\times}$.
This is a compact subgroup so its Haar measure, denoted as $vol(G_{M})$,
is finite. Now, we notice that%
\begin{align}
\int_{\mathbb{Q}_{p}^{\times}}\int_{V_{m}}\left\vert \widehat{f}\left(
k_{0},\boldsymbol{k}\right)  \right\vert ^{2}d\lambda\left(  m\right)  dm &
\geq\int_{^{G_{M}}}\int_{V_{m}}\left\vert \widehat{f}\left(  k_{0}%
,\boldsymbol{k}\right)  \right\vert ^{2}d\lambda\left(  m\right)
dm\nonumber\\
&  =\int_{G_{M}}\int_{V_{m}}\left\vert \widehat{f}\left(  k_{0},\boldsymbol{k}%
\right)  \right\vert ^{2}\frac{d^{3}\boldsymbol{k}dm}{\left\vert
m+\mathfrak{q}_{0}\left(  \boldsymbol{k}\right)  \right\vert _{p}^{\frac{1}%
{2}}}.\label{Eq24}%
\end{align}
We now use the fact that

\textbf{Claim E. }The mapping
\[%
\begin{array}
[c]{cccc}%
\sqrt{\cdot}: & G_{M} & \rightarrow & G_{M}\\
&  &  & \\
& m & \rightarrow & \sqrt{m}%
\end{array}
\]
and its inverse are $p-$adic analytic functions, for $M$ sufficiently large.

We change variables in the last integral in (\ref{Eq24}) as $y_{0}=\frac
{k_{0}}{\sqrt{m}}$, $\boldsymbol{y}=\frac{\boldsymbol{k}}{\sqrt{m}}$, then
$dk_{0}d^{3}\boldsymbol{k=}dy_{0}d^{3}\boldsymbol{y}$ and
\begin{align*}
& \int_{G_{M}}\int_{V_{m}}\left\vert \widehat{f}\left(  k_{0},\boldsymbol{k}%
\right)  \right\vert ^{2}\frac{d^{3}\boldsymbol{k}dm}{\left\vert
m+\mathfrak{q}_{0}\left(  \boldsymbol{k}\right)  \right\vert _{p}^{\frac{1}%
{2}}}\\
& =\int_{G_{M}}\int_{V}\left\vert \widehat{f}\left(  \sqrt{m}y_{0},\sqrt
{m}\boldsymbol{y}\right)  \right\vert ^{2}\frac{d^{3}\boldsymbol{y}%
dm}{\left\vert 1+\mathfrak{q}_{0}\left(  \boldsymbol{y}\right)  \right\vert
_{p}^{\frac{1}{2}}}.
\end{align*}
Finally since $\widehat{f}$\ is locally constant and $\sqrt{m}$ is a unit for
every $m\in G_{M}$, we have for $M$ sufficiently large that%
\begin{gather*}
\int_{G_{M}}\int_{V}\left\vert \widehat{f}\left(  \sqrt{m}y_{0},\sqrt
{m}\boldsymbol{y}\right)  \right\vert ^{2}\frac{d^{3}\boldsymbol{y}%
dm}{\left\vert 1+\mathfrak{q}_{0}\left(  \boldsymbol{y}\right)  \right\vert
_{p}^{\frac{1}{2}}}=\int_{G_{M}}\int_{V}\left\vert \widehat{f}\left(
y_{0},\boldsymbol{y}\right)  \right\vert ^{2}\frac{d^{3}\boldsymbol{y}%
dm}{\left\vert 1+\mathfrak{q}_{0}\left(  \boldsymbol{y}\right)  \right\vert
_{p}^{\frac{1}{2}}}\\
\geq vol\left(  G_{M}\right)  \int_{V^{+}}\left\vert \widehat{f}\left(
y_{0},\boldsymbol{y}\right)  \right\vert ^{2}\frac{d^{3}\boldsymbol{y}}{\left\vert 1+\mathfrak{q}_{0}\left(  \boldsymbol{y}\right)  \right\vert
_{p}^{\frac{1}{2}}}.
\end{gather*}
\textbf{Proof of Claim E.}

We first notice that $\left(  1+p^{M}\mathbb{Z}_{p}\right)  ^{2}%
=1+2p^{M}\mathbb{Z}_{p}=1+p^{M}\mathbb{Z}_{p}$ for $M\geq2$, see Lemma 8.4.1
in \cite{Igusa}. This means \ that the mapping
\begin{equation}%
\begin{array}
[c]{ccc}%
G_{M} & \rightarrow & G_{M}\\
x & \rightarrow & x^{2}%
\end{array}
\label{mapping_1}%
\end{equation}
is well-defined and surjective. Then for any $m\in G_{M}$, the equation
$x^{2}=m$ has a solution $\sqrt{m}$ in $G_{M}$. Notice that there is another
solution $-\sqrt{m}=-1+\left(  \text{higher order terms}\right)  $ which does
not belong to $G_{M}$. Consequently the mapping
\begin{equation}%
\begin{array}
[c]{ccc}%
G_{M} & \rightarrow & G_{M}\\
m & \rightarrow & \sqrt{m}%
\end{array}
\label{mapping_2}%
\end{equation}
is well-defined. The fact that the mappings (\ref{mapping_1})-(\ref{mapping_2}%
) are $p-$adic analytic follows from the implicit function theorem.
\end{proof}

\begin{remark}\label{rem17}
The preceding Proposition remains valid if we replace $R(f)=\hat{f}|_{V^+}$
by $R(f)(k)=g\left( [k]_p\right) \hat{f}(k)|_{V^+}$, where $g$ is any continuous
function $g:\mathbb{R}_+\to\mathbb{C}$.
\end{remark}

\begin{lemma}
\label{lemma16}There exist a positive constant \ $C$ such that%
\[
\frac{1}{\left\vert 1+\mathfrak{q}_{0}(\boldsymbol{k})\right\vert _{p}}\leq
C\text{ for any }\boldsymbol{k}\in\mathbb{Q}_{p}^{3}\text{.}%
\]

\end{lemma}

\begin{proof}
The hypothesis $p\equiv1$ $\operatorname{mod}$ $4$ implies $W=\left\{
\boldsymbol{k}\in\mathbb{Z}_{p}^{3};1+\mathfrak{q}_{0}(\boldsymbol{k}%
)=0\right\}  =\emptyset$, see Remark \ref{note_q}.

\textbf{Claim A.} $\left\vert 1+\mathfrak{q}_{0}(\boldsymbol{k})\right\vert
_{p}>C_{1}$ for any $C_{1}\in\left(  0,p\right)  $ and for any $\left\Vert
\boldsymbol{k}\right\Vert _{p}\geq p$.

We recall that $\mathfrak{q}_{0}(\boldsymbol{k})$ and $\mathfrak{q}%
(k_{0},\boldsymbol{k})$ are elliptic quadratic forms and that%
\begin{equation}
\left\vert \mathfrak{q}_{0}(\boldsymbol{k})\right\vert _{p}=|\mathfrak{q}%
(0,\boldsymbol{k})|_{p}\geq\left( \inf_{\boldsymbol{x}\in S_{0}^{3}}|\mathfrak{q}%
(0,\boldsymbol{x})|_{p}\right)\Vert\boldsymbol{k}\Vert_{p}^{2}=p^{-1}\Vert
\boldsymbol{k}\Vert_{p}^{2}\text{ for any }\boldsymbol{k}\in\mathbb{Q}_{p}%
^{3}\text{,}\label{Eq20}%
\end{equation}
see (\ref{desigualdaLNM}). Now, $p^{-1}\Vert\boldsymbol{k}\Vert_{p}%
^{2}>1$ if and only if $\Vert\boldsymbol{k}\Vert_{p}\geq p$, and by applying the
ultrametric property of the norm $\Vert\boldsymbol{\cdot}\Vert_{p}$, we get
from (\ref{Eq20}), that for $\Vert\boldsymbol{k}\Vert_{p}\geq p$,
\[
\left\vert 1+\mathfrak{q}_{0}(\boldsymbol{k})\right\vert _{p}=\max\left\{
1,\mathfrak{q}_{0}(\boldsymbol{k})\right\}  \geq p^{-1}\Vert\boldsymbol{k}%
\Vert_{p}^{2}\geq p>C_{1}\text{ for any }C_{1}\in\left(  0,p\right)  .
\]

\textbf{Claim B. }There exist a constant $C_{0}$ such that 
$\inf_{\boldsymbol{k}\in\mathbb{Z}_{p}^{3}}\left\vert 1+\mathfrak{q}%
_{0}(\boldsymbol{k})\right\vert _{p}\geq C_{0}>0$.

This assertion follows from the fact that $\left\vert 1+\mathfrak{q}%
_{0}(\boldsymbol{k})\right\vert _{p}>0$ for any $\boldsymbol{k}\in
\mathbb{Z}_{p}^{3}$. The statement of the lemma is a consequence of Claims A and B.
\end{proof}

\begin{lemma}
\label{lemma16A}The mapping%
\[%
\begin{array}
[c]{cccc}%
R: & L_{\mathbb{C}}^{2}\left(  \mathbb{Q}_{p}^{3},d^{3}\boldsymbol{x}\right)
& \rightarrow & L_{\mathbb{C}}^{2}\left(  V^{+},d\lambda\right) \\
&  &  & \\
& g & \rightarrow & \widehat{g}\mid_{_{V^{+}}}%
\end{array}
\]
satisfies $\left\Vert R(g)\right\Vert _{L_{\mathbb{C}}^{2}\left(
V^{+},d\lambda\right)  }\leq C\left\Vert g\right\Vert _{L_{\mathbb{C}}%
^{2}\left(  \mathbb{Q}_{p}^{3},d^{3}\boldsymbol{x}\right)  }$. Here
$\widehat{g}\left(  \boldsymbol{k}\right)  $ denotes the $3$-dimensional
Fourier transform is defined\ with respect to the bilinear form $-\mathfrak{B}%
_{0}\left(  \mathbf{x},\mathbf{y}\right)=-sx_{1}y_{1}-px_{2}y_{2}+spx_{3}y_{3}$.
\end{lemma}

\begin{proof}
The results follows from Lemma \ref{lemma16}, by using that $\left\vert
k_{0}\right\vert _{p}=\left\vert \sqrt{\omega\left(  \boldsymbol{k}\right)
}\right\vert _{p}=\left\vert 1+\mathfrak{q}_{0}(\boldsymbol{k})\right\vert
_{p}^{\frac{1}{2}}$ for $\boldsymbol{k}\in U_{\mathfrak{q}}$.
\end{proof}

\begin{remark}\label{note_Prop_1} Some observations about the functional spaces involved here:
\begin{enumerate}[(i)]
\item Let $X$ be a locally compact totally disconnected
space. We denote by $\mathcal{D}_{\mathbb{C}}(X)$ the $\mathbb{C}$-vector
space \ of locally constant functions with compact support. We recall that
$V^{+}\subset\mathbb{Q}_{p}^{4}$ is an open and compact subset, then
$\mathbb{Q}_{p}^{4}\setminus V^{+}$ is open and closed subset, and \ thus
$V^{+}$ and $\mathbb{Q}_{p}^{4}\setminus V^{+}$ are locally compact totally
disconnected spaces. The following exact sequence holds:%
\begin{equation}
0\rightarrow\mathcal{D}_{\mathbb{C}}(V^{+})\rightarrow\mathcal{D}_{\mathbb{C}%
}(\mathbb{Q}_{p}^{4})\rightarrow\mathcal{D}_{\mathbb{C}}(\mathbb{Q}_{p}%
^{4}\setminus V^{+})\rightarrow0,\label{Eq_sequence}%
\end{equation}
see e.g. \cite[p. 99]{Igusa}.

\item It is well-known that the $\mathbb{C}$-space of finite-valued simple
functions is dense in $L_{\mathbb{C}}^{2}\left(  V^{+},d\lambda\right)  $. By
using the fact that $d\lambda=\frac{d^{3}\boldsymbol{k}}{\left\vert
\sqrt{\omega\left(  \boldsymbol{k}\right)  }\right\vert _{p}}$ is an inner
regular measure, one can show that any finite-valued simple function can be
approximated in the $L_{\mathbb{C}}^{2}\left(  V^{+},d\lambda\right)  $-
norm by an element of $\mathcal{D}_{\mathbb{C}}(V^{+})$. i.e. $\mathcal{D}%
_{\mathbb{C}}(V^{+})$ \ is dense in $L_{\mathbb{C}}^{2}\left(  V^{+}%
,d\lambda\right)  $.

\item The mapping%
\[%
\begin{array}
[c]{ccc}%
L_{\mathbb{C}}^{2}\left(  \mathbb{Q}_{p}^{4},d^{4}k\right)  & \underrightarrow
{R} & L_{\mathbb{C}}^{2}\left(  V^{+},d\lambda\right) \\
&  & \\
f & \rightarrow & \widehat{f}\mid_{V^{+}}%
\end{array}
\]
is a well-defined continuous mapping, more precisely,
\begin{equation}
\left\Vert \widehat{f}\mid_{V^{+}}\right\Vert _{L_{\mathbb{C}}^{2}\left(
V^{+},d\lambda\right)  }\leq C\left\Vert \widehat{f}\right\Vert
_{L_{\mathbb{C}}^{2}\left(  \mathbb{Q}_{p}^{4},d^{4}k\right)  }=C\left\Vert
f\right\Vert _{L_{\mathbb{C}}^{2}\left(  \mathbb{Q}_{p}^{4},d^{4}k\right)
}.\label{Eq_nota_Prop_1}%
\end{equation}
Indeed, (\ref{Eq_nota_Prop_1}) holds when $\widehat{f}\in\mathcal{D}%
_{\mathbb{C}}(\mathbb{Q}_{p}^{4})$ and $\widehat{f}\mid_{V^{+}}\in
\mathcal{D}_{\mathbb{C}}(V^{+})$, see Claim D, then (\ref{Eq_nota_Prop_1})
follows by the fact that $\mathcal{D}_{\mathbb{C}}(V^{+})$\ is dense in
$L_{\mathbb{C}}^{2}\left(  V^{+},d\lambda\right)  $ and \ that \ $\mathcal{D}%
_{\mathbb{C}}(\mathbb{Q}_{p}^{4})$\ is dense in $L_{\mathbb{C}}^{2}\left(
\mathbb{Q}_{p}^{4},d^{4}k\right)  $.
\end{enumerate}
\end{remark}

\begin{remark}
\label{Note_separable_space}
Regarding the spaces of integrable functions introduced in the preceding Remark, we note the following.
\begin{enumerate}[(i)]
\item We have the following sequence:%
\[
L_{\mathbb{C}}^{2}\left(  V^{+},d\lambda\right)  \text{ }\overset
{J}{\hookrightarrow}\text{ }L_{\mathbb{C}}^{2}\left(  U_{\mathfrak{q}}%
,d^{3}\boldsymbol{k}\right)  \text{ }\hookrightarrow\text{ }L_{\mathbb{C}}%
^{2}\left(  \mathbb{Q}_{p}^{3},d^{3}\boldsymbol{k}\right)  ,
\]
where `$\hookrightarrow$' denotes an isometry. The mapping $J$ is defined as%
\[
f\left(  \omega\left(  \boldsymbol{k}\right)  ,\boldsymbol{k}\right)  \text{
}\overset{J}{\rightarrow}\text{ }\frac{f\left(  \omega\left(  \boldsymbol{k}%
\right)  ,\boldsymbol{k}\right)  }{\left\vert \sqrt{\omega\left(
\boldsymbol{k}\right)  }\right\vert^{1/2}_{p}}\text{, \ }%
\boldsymbol{k}\in U_{\mathfrak{q}}\text{.}%
\]
Since $U_{\mathfrak{q}}\subset\mathbb{Q}_{p}^{3}$ is open and compact, any
function $f:U_{\mathfrak{q}}\rightarrow\mathbb{C}$ can be extended to
$\mathbb{Q}_{p}^{3}$ by putting $f\mid$ $_{\mathbb{Q}_{p}^{3}\smallsetminus
U_{\mathfrak{q}}}\equiv0$. It is known that $L_{\mathbb{C}}^{2}\left(
\mathbb{Q}_{p}^{3},d^{3}\boldsymbol{k}\right)  $ admits a countable wavelet
basis, see e.g. \cite[Theorem 8.12.1]{A-K-S}, consequently $L_{\mathbb{C}}%
^{2}\left(  V^{+},d\lambda\right)  $ \ is separable.

\item Since $L_{\mathbb{C}}^{2}\left(  \mathbb{Q}_{p}^{4},d^{4}\boldsymbol{k}%
\right)  $ and $L_{\mathbb{C}}^{2}\left(  V^{+},d\lambda\right)  $ are
separable spaces, we have%
\[
\bigotimes\limits_{j=1}^{n}L_{\mathbb{C}}^{2}\left(  \mathbb{Q}_{p}^{4}%
,d^{4}x_{j}\right)  =L_{\mathbb{C}}^{2}\left(  \mathbb{Q}_{p}^{4n}%
,\prod\limits_{j=1}^{n}d^{4}x_{j}\right)  ,
\]
and
\[
\bigotimes\limits_{j=1}^{n}L_{\mathbb{C}}^{2}\left(  V^{+},d\lambda
_{j}\right)  =L_{\mathbb{C}}^{2}\left(  \left(  V^{+}\right)  ^{n}%
,\prod\limits_{j=1}^{n}d\lambda_{j}\right)  ,
\]
where each $d^{4}x_{j}$ denotes a copy of normalized Haar measure of
$\mathbb{Q}_{p}^{4}$, and each $d\lambda_{j}$ denotes a copy of the measure
$d\lambda$.

\item Take $\theta^{\left(  n+1\right)  }\left(  y,x_{1},\ldots,x_{n}\right)
\in L_{\mathbb{C}}^{2}\left(  \mathbb{Q}_{p}^{4},d^{4}y\right)  \bigotimes
L_{\mathbb{C}}^{2}\left(  \mathbb{Q}_{p}^{4n},\prod\limits_{j=1}^{n}d^{4}%
x_{j}\right)  $, then%
\begin{align*}
\int\limits_{\mathbb{Q}_{p}^{4n}}\int\limits_{V}\left\vert \theta^{\left(
n+1\right)  }\left(  y,x_{1},\ldots,x_{n}\right)  \right\vert ^{2}%
d\lambda\left(  y\right)  \prod\limits_{j=1}^{n}d^{4}x_{j}  & \leq\\
C\int\limits_{\mathbb{Q}_{p}^{4n}}\int\limits_{\mathbb{Q}_{p}^{4}}\left\vert
\theta^{\left(  n+1\right)  }\left(  y,x_{1},\ldots,x_{n}\right)  \right\vert
^{2}d^{4}y\prod\limits_{j=1}^{n}d^{4}x_{j}  & =C\left\Vert \theta^{\left(
n+1\right)  }\right\Vert _{L_{\mathbb{C}}^{2}\left(  \mathbb{Q}_{p}^{4(n+1)}%
,\prod\limits_{j=1}^{n+1}d^{4}x_{j}\right)  }^{2}.
\end{align*}
This result follows from Claim D, by using Fubini's theorem.
\end{enumerate}
\end{remark}

\begin{lemma}
For $f\in L_{\mathbb{C}}^{2}\left(  V^{+},d\lambda\right)  $, we define
$T_{V^{+}}\left(  f\right)  \in\mathcal{D}_{\mathbb{C}}^{\prime}$ by%
\[
\left(  T_{V^{+}}\left(  f\right)  ,\varphi\right)  =\int_{V^{+}}%
f(x)\varphi\left(  x\right)  d\lambda\left(  x\right)  \text{ for }\varphi
\in\mathcal{D}_{\mathbb{C}}\text{.}%
\]
Then we have the following sequence of continuous mappings:%
\[%
\begin{array}
[c]{ccccc}%
\mathcal{H}_{\infty}\left(  \mathbb{C}\right)  & \overset{R}{\rightarrow} &
L_{\mathbb{C}}^{2}\left(  V^{+},d\lambda\right)  & \overset{T_{V^{+}}%
}{\rightarrow} & \mathcal{H}_{\infty}^{\ast}\left(  \mathbb{C}\right)  ,
\end{array}
\]
where the map $R$ is defined as in Proposition \ref{Prop-1}.
\end{lemma}

\begin{proof}
The support of $T_{V^{+}}\left(  f\right)  $ is compact since it is contained
in $V$, which is a compact subset of $\mathbb{Q}_{p}^{4}$. The Fourier
transform of $T_{V^{+}}\left(  f\right)  $ in $\mathcal{D}_{\mathbb{C}%
}^{\prime}$ is the locally constant function%
\[
\widehat{f}\left(  k\right)  =\int_{V^{+}}\chi_{p}\left(  \mathcal{B}\left(
x,k\right)  \right)  f(x)d\lambda\left(  x\right)  ,
\]
for a similar calculation the reader may see, for instance, \cite[Theorem
4.9.3]{A-K-S}. Now, identifying $f$ with the induced distribution
$T_{V^+}f$ on $V^+$, by using the definition of $\mathcal{H}_{\infty}^{\ast
}\left(  \mathbb{C}\right)  $ (see \eqref{dual_space}), the Cauchy-Schwartz
inequality, and the fact that $\int_{\mathbb{Q}_{p}^{4}}\left[  k\right]
^{-l}d^{4}k<\infty$ for $l\geq5$, we have
\begin{align*}
\left\Vert f\right\Vert _{-l}^{2} &  =\int_{\mathbb{Q}_{p}^{4}}\left[
k\right]  ^{-l}\left\vert \widehat{f}(k)\right\vert ^{2}d^{4}k=\int
_{\mathbb{Q}_{p}^{4}}\left[  k\right]  ^{-l}\left\vert \int_{V^{+}}\chi
_{p}\left(  \mathcal{B}\left(  x,k\right)  \right)  f(x)d\lambda\left(
x\right)  \right\vert ^{2}d^{4}k\\
&  \leq C(l)\int_{V^{+}}\left\vert f(x)\right\vert ^{2}d\lambda\left(
x\right)  =C(l)\left\Vert f\right\Vert _{L_{\mathbb{C}}^{2}\left(
V^{+},d\lambda\right)  }^{2},
\end{align*}
which implies that $T_{V^{+}}\left(  f\right)  \in\mathcal{H}_{\infty}^{\ast
}\left(  \mathbb{C}\right)  $.
\end{proof}

\section{Free non-Archimedean quantum fields}

\subsection{The Segal quantization}

\label{se5}

We start by \ reviewing some \ well-known fact about quantization. For an
in-depth discussion the reader may consult \cite{SW64,Reed-SimonII}, see also
\cite{Dimock,Folland,Lopu,Strocchi} for more physically-oriented approaches.
Our presentation follows closely the book of Reed and Simon \cite{Reed-SimonII}. 
In particular, our notation mimics the one used in that book. 
We set
$\mathcal{H}=L_{\mathbb{C}}^{2}\left(  V^{+},d\lambda\right)  $ and denote by
$\left\langle \cdot,\cdot\right\rangle $ the inner product of $\mathcal{H}$.
We assume that $\left\langle f,\alpha g\right\rangle =\alpha\left\langle
f,g\right\rangle $, for $\alpha\in\mathbb{C}$, and $f$, $g\in\mathcal{H}$. We
define \textit{the Fock space over} $\mathcal{H}$ as $\mathfrak{F}%
(\mathcal{H})=\oplus_{n=0}^{\infty}\mathcal{H}^{\left(  n\right)  }$, where
$\mathcal{H}^{\left(  n\right)  }=\otimes_{k=1}^{n}\mathcal{H}$, by definition
$\mathcal{H}^{\left(  0\right)  }=\mathbb{C}$. We denote by $S_{n}%
:\mathcal{H}^{\left(  n\right)  }\rightarrow S\mathcal{H}^{\left(  n\right)
}$, the symmetrization operator, and define $S=\oplus_{n=0}^{\infty}S_{n}$,
see \cite[Section II.4]{Reed-SimonI}. The symmetric Fock space over
$\mathcal{H}$ (also called \textit{the boson Fock space} over $\mathcal{H}$)
is defined as $\mathfrak{F}_{s}(\mathcal{H})=\oplus_{n=0}^{\infty}%
\mathcal{H}_{s}^{\left(  n\right)  }$, where $\mathcal{H}_{s}^{\left(
n\right)  }=S_{n}\mathcal{H}^{\left(  n\right)  }$. We call $\mathcal{H}%
_{s}^{\left(  n\right)  }$\ the $n$-\textit{particle subspace} of
$\mathfrak{F}_{s}(\mathcal{H})$. We use the same symbol 
$\left\langle \cdot,\cdot\right\rangle$ to denote the inner product of 
$\mathfrak{F}(\mathcal{H})$.

We now fix a vector $f$ in $\mathcal{H}$. For the vectors of the
form $\eta=\psi_{1}\otimes\ldots\otimes\psi_{n}$, we define a map
$b^{-}\left(  f\right)  :\mathcal{H}^{\left(  n\right)  }\rightarrow
\mathcal{H}^{\left(  n-1\right)  }$ by $b^{-}\left(  f\right)  \left(
\eta\right)  =\left\langle f,\psi_{1}\right\rangle \psi_{2}\otimes
\ldots\otimes\psi_{n}$. Then $b^{-}\left(  f\right)  $ extends to a bounded
map (of norm $\left\Vert f\right\Vert _{\mathcal{H}}$) of $\mathcal{H}%
^{\left(  n\right)  }$ in to $\mathcal{H}^{\left(  n-1\right)  }$. In the case
$n=0$, we define $b^{-}\left(  f\right)  :\mathcal{H}^{\left(  0\right)
}\rightarrow0$. The adjoint $b^{+}\left(  f\right)  :\mathcal{H}^{\left(
n\right)  }\rightarrow\mathcal{H}^{\left(  n+1\right)  }$ of $b^{-}\left(
f\right)  $ is defined as $b^{+}\left(  f\right)  \left(  \psi_{1}%
\otimes\ldots\otimes\psi_{n}\right)  =f\otimes\psi_{1}\otimes\ldots\otimes
\psi_{n}$. The map $f\rightarrow b^{+}\left(  f\right)  $ is linear, but
$f\rightarrow b^{-}\left(  f\right)  $ is anti-linear.

The boson Fock space is invariant under $b^{-}\left(  f\right)  $ but not
under $b^{+}\left(  f\right)  $. A vector $\psi=\left\{  \psi^{\left(
n\right)  }\right\}  _{n\in\mathbb{N}}\in\mathfrak{F}_{s}(\mathcal{H})$ is
called a \textit{finite particle vector} if $\psi_{n}=0$ for all but finitely
many $n$. \ The set of all finite vectors is denoted as $F_{0}$. We set the
vector $\Upsilon_{0}=\left(  1,0,0,\ldots\right)  $ to be the \textit{vacuum}.

Let $A$ be a self-adjoint operator on $\mathcal{H}$ with domain of essential
self-adjointness $D$. Let $D_{A}=\left\{  \psi\in F_{0};\psi^{\left(
n\right)  }\in\otimes_{k=1}^{n}D\text{ for each }n\right\}  $. We define the
operator $\Gamma\left(  A\right)  $ (\textit{the second quantization of }$A$)
on $D_{A}\cap\mathcal{H}_{s}^{\left(  n\right)  }$ as
\[
A\otimes I\otimes\cdots\otimes I+I\otimes A\otimes\cdots\otimes I+\cdots
+I\otimes I\otimes\cdots\otimes A,
\]
where $I$ is the identity operator. The operator $\Gamma\left(  A\right)  $ is
essentially self-adjoint on $D_{A}$. In the case $A=I$, the second
quantization $N=\Gamma\left(  A\right)  $ (\textit{the number operator}) is
essentially self-adjoint on $F_{0}$ and for $\phi\in\mathcal{H}_{s}^{\left(
n\right)  }$, $N\phi=n\phi$.

The\textit{\ annihilation operator} $a^{-}(f)$ on $\mathfrak{F}_{s}%
(\mathcal{H})$ with domain $F_{0}$ is given by%
\[
a^{-}(f)=\sqrt{N+1}\text{ }b^{-}(f).
\]
For $\psi$, $\eta$ in $F_{0}$,%
\[
\left\langle \sqrt{N+1}\text{ }b^{-}(f)\psi,\eta\right\rangle =\left\langle
\psi,Sb^{+}(f)\sqrt{N+1}\eta\right\rangle ,
\]
which implies that
\[
\left(  a^{-}(f)\right)  ^{\ast}\upharpoonright_{F_{0}}=Sb^{+}(f)\sqrt{N+1},
\]
where `$\ast$' denotes the adjoint operator. The operator $\left(
a^{-}(f)\right)  ^{\ast}$ is called the \textit{creation operator}. Both
$a^{-}(f)$ and $\left(  a^{-}(f)\right)  ^{\ast}\upharpoonright_{F_{0}}$ are
closable, the corresponding closures are denoted as $a^{-}(f)$\ and as
$a^{-}(f)^{\ast}$.

\begin{definition}
For $f\in\mathcal{H}$, the Segal quantum field operator $\boldsymbol{\Phi
}_{\text{S}}$ on $F_{0}$ is defined as
\begin{equation}
\boldsymbol{\Phi}_{\text{S}}(f)=\dfrac{1}{\sqrt{2}}[a^{-}(f)+a^{-}(f)^{\ast}].
\label{free field}%
\end{equation}

\end{definition}

The mapping \ from $\mathcal{H}$ into \ the self-adjoint operators on
$\mathfrak{F}_{s}(\mathcal{H})$ given by $f\rightarrow\boldsymbol{\Phi
}_{\text{S}}(f)$ is called \textit{the Segal quantization over} $\mathcal{H}$.
Notice that the Segal quantization is a real linear map.

\begin{remark}
By using \ the fundamental properties of the Segal quantization, see
\cite[Theorem X.41 ]{Reed-SimonII}, we obtain the following facts (among others):
\begin{enumerate}[(i)]
\item For each $f\in\mathcal{H}$, $\boldsymbol{\Phi}_{\text{S}}(f)$ is
essentially self-adjoint on $F_{0}$.

\item The commutation relations: for each $\psi\in F_{0}$, and $f$,
$g\in\mathcal{H}$,%
\begin{equation}
\boldsymbol{\Phi}_{\text{S}}(f)\boldsymbol{\Phi}_{\text{S}}(g)\psi
-\boldsymbol{\Phi}_{\text{S}}(g)\boldsymbol{\Phi}_{\text{S}}(f)\psi=\sqrt
{-1}\operatorname{Im}\left(  \left\langle f,g\right\rangle \right)  \psi,
\label{CMM}%
\end{equation}
that is, $\left[  \boldsymbol{\Phi}_{\text{S}}(f),\boldsymbol{\Phi}_{\text{S}%
}(g)\right]  =\sqrt{-1}\operatorname{Im}\left(  \left\langle f,g\right\rangle
\right)  I $, on $F_{0}$.
\end{enumerate}
\end{remark}

\subsubsection{The free Hermitian field of unit mass}

We define for each $f\in\mathcal{H}_{\infty}\left(  \mathbb{R}\right)  $,%
\[
\boldsymbol{\Phi}(f)=\boldsymbol{\Phi}_{\text{S}}(Rf)\,,
\]
with $R$ defined as in Lemma \ref{lemma16A}, and for each 
$g\in\mathcal{H}_{\infty}\left(  \mathbb{C}\right)$,
\begin{equation}
\boldsymbol{\Phi}(g)=\boldsymbol{\Phi}(\operatorname{Re}g)+\sqrt
{-1}\boldsymbol{\Phi}(\operatorname{Im}g)\,. \label{Quantum_field}%
\end{equation}
We call the mapping $g\rightarrow\boldsymbol{\Phi}(g)$ \textit{the
free Hermitian scalar field of unit mass}.

\begin{remark}
By extending the mapping $R$ as in Remark \ref{rem17}, the field $f\mapsto \boldsymbol{\Phi}(f)$ 
remains well-defined. We emphasize that the presence of $R$ (in any of its forms) means that we
are working on-shell.
\end{remark}

\subsubsection{The $p$-adic restricted Poincar\'{e} group}
As we do not have the structure of light cones available, we must choose a substitute for them.
Here we will base our treatment on the mass shells $V^\pm$.

We define \textit{the }$p$\textit{-adic restricted Lorentz} group as%
\[
\mathcal{L}_{+}^{\uparrow}=\left\{  \Lambda\in\boldsymbol{O}(\mathfrak{q}%
);\Lambda\left(  V^{\pm}\right)  =V^{\pm}\right\}  .
\]
This group is non trivial since transformations of the form%
\[
\left\{  \left[
\begin{array}
[c]{cc}%
1 & 0\\
0 & \digamma
\end{array}
\right]  \in\boldsymbol{O}(\mathfrak{q})\text{; }\digamma\in\boldsymbol{O}%
(\mathfrak{q}_{0})\right\}  ,
\]
belong to $\mathcal{L}_{+}^{\uparrow}$.
A further justification for choosing $V^\pm$ as a replacement for the light cones comes
from the fact that the distributions $\delta_{\pm}\left(
\mathfrak{q}-1\right)  $ are invariant under $\mathcal{L}_{+}^{\uparrow}$, see
\cite[Lemma 163]{Zuniga-LNM-2016}.

We define \textit{the }$p$\textit{-adic restricted Poincar\'{e} group}
$\mathcal{P}_{+}^{\uparrow}$ as the set of pairs $\left(  a,\Lambda\right)  $,
where $a\in\mathbb{Q}_{p}^{4}$ and $\Lambda\in\mathcal{L}_{+}^{\uparrow}$,
with the group operation
\[
\left(  a,\Lambda_{1}\right)  \left(  b,\Lambda_{2}\right)  =\left(
a+\Lambda_{1}b,\Lambda_{1}\Lambda_{2}\right)  .
\]

The group $\mathcal{P}_{+}^{\uparrow}$ acts naturally on $\mathbb{Q}_{p}^{4} $
by setting $\left(  a,\Lambda\right)  x=\Lambda x+a$. With the topology inherited from
$\left(  \mathbb{Q}_{p}^{4},\left\Vert \cdot\right\Vert _{p}\right)  $,
$\mathcal{L}_{+}^{\uparrow}$ and $\mathcal{P}_{+}^{\uparrow}$ become locally
compact topological groups.

On $L_{\mathbb{C}}^{2}\left(  V^{+},d\lambda\right)  $, we define the
following projective representation of the restricted Poincar\'{e} group:%
\begin{equation}
\left(  U\left(  a,\Lambda\right)  \psi\right)  \left(  k\right)  =\chi
_{p}\left(  \mathcal{B}\left(  a,k\right)  \right)  \psi\left(  \Lambda
^{-1}k\right)  . \label{representation_Poincare}%
\end{equation}

\subsection{The $p$-adic Wightman axioms}

We present here a $p$-adic counterpart of the classical Wightman axioms, see
e.g. \cite{SW64,Reed-SimonII}, and references therein. We use units where the
rationalized Planck's constant and the speed of light are equal to one. We
take $H=\mathfrak{F}_{s}(L_{\mathbb{C}}^{2}\left(  V^{+},d\lambda\right)  ) $,
$\mathfrak{U}=\Gamma\left(  U\left(  \cdot,\cdot\right)  \right)  $, with
$U\left(  \cdot,\cdot\right)  $ being defined as in (\ref{representation_Poincare}),
$\boldsymbol{\Phi}$ as in (\ref{Quantum_field}), and $D=F_{0}$. A
$p$\textit{-adic Hermitian scalar quantum field theory} is a quadruple
$\left\{  H,\mathfrak{U},\boldsymbol{\Phi},D\right\}  $ which satisfies the
following properties:

\noindent\textbf{Relativistic invariance of states}. $\ H$ is a separable
Hilbert space and $\mathfrak{U}\left(  \cdot,\cdot\right)  $ is a strongly
continuous unitary representation on $H$ of the $p$-adic restricted
Poincar\'{e} group.

\noindent\textbf{Spectral condition}. We define \textit{the closed forward
semigroup} $\overline{S(V^{+})}$ as the topological closure of the additive
semigroup generated by the vectors of $V^{+}$. Notice that since $V^{+}%
\subset\mathbb{Z}_{p}^{4}$, $\overline{S(V^{+})}$ is a compact subset of
$\mathbb{Z}_{p}^{4}$. Furthermore, since $\mathcal{L}_{+}^{\uparrow}\left(
V^{+}\right)  =V^{+}$, we have $\mathcal{L}_{+}^{\uparrow}\left(
\overline{S(V^{+})}\right)  =\overline{S(V^{+})}$. The $p$-adic counterpart of
the spectral condition is the following: there exists a projection-valued
measure $E_{V^{+}}$ on $\mathbb{Q}_{p}^{4}$ corresponding to $\mathfrak{U}%
(a,I)$ having su\-pport in $\overline{S(V^{+})}$.

\begin{remark}
In the classical case by using a Stone type theorem, see \cite[Theorem
VIII.12]{Reed-SimonI}, one shows the existence of four commuting operators
$P_{0}$, $P_{1}$, $P_{2}$, $P_{3}$, on a suitable Hilbert space so that
$\mathfrak{U}(a,I)=e^{i\sum a_{j}P_{j}}$. In the $p$-adic case, we do not have
a complete theory of semigroups, with $p$-adic time, for operators acting on
complex-valued functions. For this reason, at the moment, we do not have a
definition for the $p$-adic counterparts of the operators $P_{0}$, $P_{1}$,
$P_{2}$, $P_{3}$, and consequently, we do not know their spectra.
\end{remark}

\noindent\textbf{Existence and uniqueness of the vacuum.} There exists a
unique vector $\Upsilon_{0}\in H$ such that $U\left(  a,I\right)  \Upsilon
_{0}=\Upsilon_{0}$ for all $a\in\mathbb{Q}_{p}^{4}$, this vector is called
\textit{the vacuum}.

\noindent\textbf{Invariant domains for fields.} There exists a dense subspace
$D\subset H$ and a map from $\mathcal{H}_{\infty}\left(  \mathbb{C}\right)  $
to the unbounded operators on $H$ such that:
\begin{enumerate}[(i)]
\item For each $f\in\mathcal{H}_{\infty}\left(  \mathbb{C}\right)$, we have that
$D\subset Dom\left(  \boldsymbol{\Phi}\left(  f\right)  \right)  $, $D\subset
Dom\left(  \boldsymbol{\Phi}\left(  f\right)  ^{\ast}\right)  $, and
$\boldsymbol{\Phi}\left(  f\right)  ^{\ast}\upharpoonright D=\boldsymbol{\Phi
}\left(  \overline{f}\right)  \upharpoonright D$.

\item $\Upsilon_{0}\in D$, and \ $\boldsymbol{\Phi}\left(  f\right)
D\subset D$ for any $f\in\mathcal{H}_{\infty}\left(  \mathbb{C}\right)  $.

\item For a fixed $\psi\in D$, the map $f\rightarrow\boldsymbol{\Phi
}\left(  f\right)  \psi$ is linear.
\end{enumerate}
\noindent\textbf{Regularity of the field.} For any $\psi_{1}$ and $\psi_{2}$
in $D$, the map
\[
f\rightarrow\left\langle \psi_{1},\boldsymbol{\Phi}\left(  f\right)  \psi
_{2}\right\rangle _{H}%
\]
is an element of $\mathcal{H}_{\infty}^{\ast}\left(  \mathbb{C}\right)  $.
In the Archimedean case this is just a tempered distribution, here it
turns out to be an element of $\mathcal{H}_{\infty}^{\ast}\left(  \mathbb{C}\right)$,
providing yet another argument to consider this space as the correct replacement in
the $p-$adic framework of the Schwartz space $\mathcal{S}$.

\noindent\textbf{Poincar\'{e} invariance of the field.} For each $\left(
a,\Lambda\right)  \in\mathcal{P}_{+}^{\uparrow}$, $\mathfrak{U}(a,\Lambda
)D\subset D$, and for all $f\in\mathcal{H}_{\infty}\left(  \mathbb{C}\right)
$, $\psi\in D$,%
\[
\mathfrak{U}\left(  a,\Lambda\right)  \boldsymbol{\Phi}\left(  f\right)
\mathfrak{U}\left(  a,\Lambda\right)  ^{-1}\psi=\boldsymbol{\Phi}\left(
\left(  a,\Lambda\right)  f\right)  \psi,
\]
where
\[
\left(  a,\Lambda\right)  f\left(  x\right)  =f\left(  \Lambda^{-1}\left(
x-a\right)  \right)  .
\]

\noindent\textbf{Local commutativity.}
The $p$-adic local commutativity property states that\ if $f$, $g$ are in $\mathcal{D}%
_{\mathbb{C}}\left(  \mathbb{Z}_{p}^{4}\right)  $, then
\[
\left[  \boldsymbol{\Phi}(f),\boldsymbol{\Phi}\left(  g\right)  \right]
\Psi=\left(  \boldsymbol{\Phi}(f)\boldsymbol{\Phi}\left(  g\right)
-\boldsymbol{\Phi}\left(  g\right)  \boldsymbol{\Phi}(f)\right)  \Psi=0,
\]
for all $\Psi\in D$. In the Archimedean case, the commutator vanishes whenever
the test functions $f,g$ are supported on two respective spacelike-separated subsets,
that is, $f(x)g(y)=0$ whenever $x-y$ does not belong to the interior of the light cone.
This subset can be characterized as the `ball of radius $0$' of Minkowski spacetime
in the sense of the theory of indefinite quadratic forms (see, e.g., \cite{Hor10} and
references therein). Our result can be seen as the equivalent statement in the $p-$adic case,
with the unit ball playing this role.

\noindent\textbf{Cyclicity of the vacuum.} The set $D_{0}$ of finite linear
combinations of vectors \ of the form $\boldsymbol{\Phi}\left(  f_{1}\right)
\cdots\boldsymbol{\Phi}\left(  f_{n}\right)  \Upsilon_{0}$ is dense in $H$.

\begin{theorem}
\label{Theorem2} The following hold true:
\begin{enumerate}[(i)]
\item\label{itemi} The quadruple
\[
\left\{  \mathfrak{F}_{s}(L_{\mathbb{C}}^{2}\left(  V^{+},d\lambda\right)
),\Gamma\left(  U\left(  \cdot,\cdot\right)  \right)  ,\boldsymbol{\Phi}%
,F_{0}\right\}
\]
satisfies the $p$-adic Wightman axioms.

\item\label{itemii} For each $f\in\mathcal{H}_{\infty}\left(  \mathbb{C}\right)  $,%
\[
\boldsymbol{\Phi}\left(  \square_{\mathfrak{q},\alpha}f\right)  =0.
\]
\end{enumerate}
\end{theorem}

\begin{proof}
In the proof of the first part \eqref{itemi}, we use the notation
\[
\mathfrak{F}_{s}=\mathfrak{F}_{s}(L_{\mathbb{C}}^{2}\left(  V^{+}%
,d\lambda\right)  )=\allowbreak\oplus_{n=0}^{\infty}\mathcal{H}_{s}^{\left(
n\right)  }.
\]

\noindent\textbf{Relativistic invariance of states}. We first note that
$\mathfrak{F}_{s}$\ is \ separable because $L_{\mathbb{C}}^{2}\left(
V^{+},d\lambda\right)  $ is separable, see Remark \ref{Note_separable_space}
(i). On the other hand, since $V^{+}$ is invariant under $\mathcal{L}_{+}^{\uparrow
}$, $U\left(  \cdot,\cdot\right)  $ is a strongly continuous unitary
representation of $\mathcal{P}_{+}^{\uparrow}$ on $L_{\mathbb{C}}^{2}\left(
V^{+},d\lambda\right)  $, see (\ref{representation_Poincare}). By definition
$\Gamma\left(  U\right)  $ is the unitary operator on $\mathfrak{F}_{s}$ given on $\mathcal{H}_{s}^{\left(  n\right)  }$ by $\otimes
_{k=1}^{n}U\left(  \cdot,\cdot\right)  $, consequently $\Gamma\left(
U\right)  :\mathcal{H}_{s}^{\left(  n\right)  }\rightarrow\mathcal{H}%
_{s}^{\left(  n\right)  }$ determines a strongly continuous unitary representation
of $\mathcal{P}_{+}^{\uparrow}$ on $\mathcal{H}_{s}^{\left(  n\right)  }$.
Notice that $\Gamma\left(  U\right)  $ is strongly continuous in $F_{0} $, and since
$F_{0}$ is dense in $\mathfrak{F}_{s}$ we conclude that $\Gamma\left(
U\right)  $ is a strongly continuous unitary representation of $\mathcal{P}%
_{+}^{\uparrow}$ \ on $\mathfrak{F}_{s}$.

\noindent\textbf{Spectral condition}. We show that the four parameter group
$\Gamma\left(  U\left(  a,I\right)  \right)  $ has associated a
projection-valued measure supported on $\overline{S(V^{+})}$. The argument
needed is exactly the classical one, see \cite[p. 213]{Reed-SimonII}. The
notion of closed forward semigroup, which is the $p-$adic counterpart of
the closed forward light cone, allows us to carry out the calculations as
in the classical case. We first notice that $L_{\mathbb{C}}^{2}\left(
V^{+},d\lambda\right)  $ is already a spectral representation of $U(a,I)$
since%
\begin{equation}
\left\langle \varphi,U(a,I)\varphi\right\rangle _{L_{\mathbb{C}}^{2}\left(
V^{+},d\lambda\right)  }=%
{\displaystyle\int\limits_{V^{+}}}
\chi_{p}\left(  \mathfrak{B}\left(  a,k\right)  \right)  \left\vert
\varphi\left(  k\right)  \right\vert ^{2}d\lambda\left(  k\right)  .
\label{Eq_spectral_cond}%
\end{equation}
Notice that if we define for $\varphi$, $\theta\in L_{\mathbb{C}}^{2}\left(
V^{+},d\lambda\right)  $, the set function%
\[
B\rightarrow%
{\displaystyle\int\limits_{V^{+}}}
\overline{\varphi\left(  k\right)  }\chi_{p}\left(  \mathfrak{B}\left(
a,k\right)  \right)  \theta\left(  k\right)  d\lambda\left(  k\right)  ,
\]
$B$ being a Borel set in $V^{+}$, and denote the corresponding projection-valued
measure as $d(\varphi,E_{k}\varphi)$, in the case $\varphi=\theta$, then
\ (\ref{Eq_spectral_cond}) can be rewritten as
\[
\left\langle \varphi,U(a,I)\varphi\right\rangle _{L_{\mathbb{C}}^{2}\left(
V^{+},d\lambda\right)  }=%
{\displaystyle\int\limits_{V^{+}}}
\chi_{p}\left(  \mathfrak{B}\left(  a,k\right)  \right)  d(\varphi
,E_{k}\varphi).
\]
Now, since $\Gamma\left(  U(a,I)\right)  \upharpoonright\mathcal{H}%
_{s}^{\left(  n\right)  }=%
{\textstyle\bigotimes\nolimits_{k=1}^{n}}
U(a,I)$, if $\varphi^{\left(  n\right)  }\in\mathcal{H}_{s}^{\left(  n\right)
}$ with $n>0$, then
\begin{align*}
& \left\langle \varphi^{(n)},U(a,I)\varphi\right\rangle =\\%
& \int_{V^+}\cdots\int_{V^+}
\chi_{p}\left(  \mathfrak{B}\left(  a,\sum_{i=1}^{n}k_{i}\right)  \right)
\left\vert \varphi^{\left(  n\right)  }\left(  k_{1},\ldots,k_{n}\right)
\right\vert ^{2}\prod\limits_{k=1}^{n}d\lambda\left(  k_{i}\right)  =\\%
& \int_{V^+}
\chi_{p}\left(  \mathfrak{B}\left(  a,l\right)  \right)  d\mu_{\varphi
^{\left(  n\right)  }}(l)\,,
\end{align*}
where
\[
\mu_{\varphi^{(n)}}(A)=
\int\underset{\sum k_i\in A}{\cdots}\int
\left\vert \varphi^{\left(  n\right)  }\left(  k_{1},\ldots,k_{n}\right)
\right\vert ^{2}\prod\limits_{k=1}^{n}d\lambda\left(  k_{i}\right) \,,
\]
$A$ being a Borel set in $\overline{S(V^+)}$.
Since $\lambda$ is supported on $V^{+}\subset\overline{S(V^{+})}$ and
$S(V^{+})$\ is an additive semigroup, then $\mu_{\varphi^{\left(  n\right)  }}$ is
supported on $\overline{S(V^{+})}$, for any $\varphi^{\left(  n\right)  }%
\in\mathcal{H}_{s}^{\left(  n\right)  }$. We now take $\Psi=\left\{
\Psi^{\left(  n\right)  }\right\}  _{n\in\mathbb{N}}$ in $\mathfrak{F}_{s}$
and denote by $\mu_{\Psi}$ the spectral measure so that%
\[
\left\langle \Psi,\Gamma\left(  U\left(  a,I\right)  \right)  \Psi
\right\rangle =\int\chi_{p}\left(  \mathfrak{B}\left(  a,k\right)  \right)
d\mu_{\Psi}\left(  k\right)  ,
\]
then $\mu_{\Psi}=\sum_{n=0}^{\infty}\mu_{\Psi^{(n)}}$ since $\Gamma\left(
U(a,I)\right)  :\mathcal{H}_{s}^{\left(  n\right)  }\rightarrow\mathcal{H}%
_{s}^{\left(  n\right)  }$.

\noindent\textbf{Existence and uniqueness of the vacuum. }The 
argument in the $p$-adic case is the same as the Archimedean one, see \cite[p. 213]{Reed-SimonII}.

\noindent\textbf{Invariant domains for fields. }By Proposition \ \ref{Prop-1},
we have
\begin{equation}
\mathcal{H}_{\infty}(\mathbb{C})\overset{R}{\rightarrow}L_{\mathbb{C}}%
^{2}\left(  V^{+},d\lambda\right)  \rightarrow F_{0}\rightarrow\mathfrak{F}%
_{s}(L_{\mathbb{C}}^{2}\left(  V^{+},d\lambda\right)  ), \label{Eq_sequence_2}%
\end{equation}
where all the arrows denote continuous mappings. By using sequence
(\ref{Eq_sequence}), $\mathcal{D}_{\mathbb{C}}(V^{+})\subset\mathcal{D}%
_{\mathbb{C}}(\mathbb{Q}_{p}^{4})$, and since $\mathcal{D}_{\mathbb{C}%
}(\mathbb{Q}_{p}^{4})\subset\mathcal{H}_{\infty}(\mathbb{C})$, \ $\mathcal{F}%
(\mathcal{D}_{\mathbb{C}})=\mathcal{D}_{\mathbb{C}}$, and $\mathcal{D}%
_{\mathbb{C}}(V^{+})$ is dense in $L_{\mathbb{C}}^{2}\left(  V^{+}%
,d\lambda\right)  $, we conclude that $R(\mathcal{H}_{\infty}(\mathbb{C}))$ is
dense in $L_{\mathbb{C}}^{2}\left(  V^{+},d\lambda\right)$, and hence 
$\oplus^\infty_{n=0}S_n(\otimes_{n}R(\mathcal{H}_{\infty}(\mathbb{C})))$
in $\mathfrak{F}_s(L^2_\mathbb{C}(V^+,d\lambda ))$.

If $f$ is real-valued, we use that $\boldsymbol{\Phi}_S(f)$ 
is essentially self-adjoint on $F_{0}$,
the fact that $\boldsymbol{\Phi}_S(f):F_{0}\rightarrow F_{0}$, and
sequence \eqref{Eq_sequence_2}, jointly with the density of $R(\mathcal{H}%
_{\infty}(\mathbb{C}))$ to obtain that $\boldsymbol{\Phi}(f)
\upharpoonright_{F_{0}}$ is essentially self-adjoint, and $\boldsymbol{\Phi}$$\left(  f\right)  :F_{0}%
\rightarrow F_{0}$. \ If $f$ is complex-valued, the results follows from the
previous discussion by using the definition of $\boldsymbol{\Phi}$$\left(
f\right)  $.

\noindent\textbf{Regularity of the field. }\ Suppose that $\psi_{1}$,
$\psi_{2}\in F_{0}$ and that $f_{n}$ $\rightarrow$\ $f\in$ $\mathcal{H}%
_{\infty}(\mathbb{C})$ (i.e. $f_{n}$ $\overset{\left\Vert \cdot\right\Vert
_{l}}{\rightarrow}$\ $f$ for any $l\in\mathbb{N}$), with $f_{n}$ real-valued.
Then (\ref{Eq21}) implies that
\[
\widehat{f}_{n}\mid_{V^{+}}\overset{L_{\mathbb{C}}^{2}\left(  V^{+}%
,d\lambda\right)  }{\rightarrow}\ \widehat{f}\mid_{V^{+}},
\]
i.e. $R(f_{n})\rightarrow R(f)$ in $\mathfrak{F}_{s}$, see sequence
(\ref{Eq_sequence_2}). Now by using Segal's quantization, cf. Theorem X.41-(d)
in \cite{Reed-SimonII}, we have $\boldsymbol{\Phi}\left(  f_{n}\right)
\psi\rightarrow\boldsymbol{\Phi}\left(  f\right)  \psi$ \ for all $\psi$ in
$F_{0}$, therefore%
\[
\left\langle \psi_{1},\boldsymbol{\Phi}\left(  f_{n}\right)  \psi
_{2}\right\rangle \rightarrow\left\langle \psi_{1},\boldsymbol{\Phi}\left(
f\right)  \psi_{2}\right\rangle \text{.}%
\]
By treating the real and imaginary parts of $f$ separately, we obtain that
$\left\langle \psi_{1},\boldsymbol{\Phi}\left(  f\right)  \psi_{2}%
\right\rangle $ is a complex-valued bilinear form in $F_{0}\times F_{0}$, and
that
\begin{equation}
\left\vert \left\langle \psi_{1},\boldsymbol{\Phi}\left(  f\right)  \psi
_{2}\right\rangle \right\vert \leq\left\Vert \psi_{1}\right\Vert \left\Vert
\boldsymbol{\Phi}\left(  f\right)  \psi_{2}\right\Vert . \label{Eq27}%
\end{equation}
We now estimate $\left\Vert \boldsymbol{\Phi}\left(  f\right)  \psi
_{2}\right\Vert $. By the definition of $\boldsymbol{\Phi}\left(  f\right)  $,
it is sufficient to consider that $f$ is real-valued. By taking $\psi
_{2}=\left\{  \psi_{2}^{\left(  n\right)  }\right\}  _{n\in N}$, $x_{i}%
\in\mathbb{Q}_{p}^{4}$ for $i\in\{1,\ldots,n\}$, $y\in V^{+}$, and using that%
\begin{align*}
\left(  \boldsymbol{\Phi}\left(  f\right)  \psi_{2}\right)  ^{\left(
n\right)  }\left(  x_{1},\cdots,x_{n}\right)   &  =\frac{\sqrt{n+1}}{\sqrt{2}%
}\int_{V^{+}}\overline{\widehat{f}(y)}\psi_{2}^{\left(  n+1\right)  }\left(
y,x_{1},\cdots,x_{n}\right)  d\lambda\left(  y\right) \\
&  +\frac{1}{\sqrt{2n}}\sum_{i=1}^{n}\widehat{f}(x_{i})\psi_{2}^{\left(
n-1\right)  }\left(  x_{1},\cdots,\widetilde{x}_{i},\cdots,x_{n}\right)  ,
\end{align*}
where $\widetilde{x}_{i}$ means that $x_{i}$\ is omitted, we have%
\begin{gather*}
\left\Vert \left(  \boldsymbol{\Phi}\left(  f\right)  \psi_{2}\right)
^{\left(  n\right)  }\right\Vert _{\mathcal{H}_{s}^{\left(  n\right)  }}%
^{2}=\\
\frac{(n+1)}{2}\int\limits_{\mathbb{Q}_{p}^{4n}}\left\vert \int_{V^{+}%
}\overline{\widehat{f}(y)}\psi_{2}^{\left(  n+1\right)  }\left(
y,x_{1},\cdots,x_{n}\right)  d\lambda\left(  y\right)  \right\vert ^{2}%
\prod\limits_{j=1}^{n}d^{4}x_{j}+\\
\frac{1}{2n}\int\limits_{\mathbb{Q}_{p}^{4n}}\left\vert \sum_{i=1}^{n}%
\widehat{f}(x_{i})\psi_{2}^{\left(  n-1\right)  }\left(  x_{1},\cdots
,\widetilde{x}_{i},\cdots,x_{n}\right)  \right\vert ^{2}\prod\limits_{j=1}%
^{n}d^{4}x_{j}=:I_{0}+I_{1}.
\end{gather*}
To estimate $I_{0}$, we use the Cauchy-Schwartz inequality, estimation
\eqref{Eq21}, and Remark \ref{Note_separable_space} (iii) to get:
\begin{align*}
I_{0}  &  \leq\frac{(n+1)}{2}\left\{  \int_{V^{+}}\left\vert \widehat{f\left(
y\right)  }\right\vert ^{2}d\lambda\left(  y\right)  \right\}  \times\\
&  \left\{  \int\limits_{\mathbb{Q}_{p}^{4n}}\int_{V^{+}}\left\vert \psi
_{2}^{\left(  n+1\right)  }\left(  y,x_{1},\cdots,x_{n}\right)  \right\vert
^{2}d\lambda\left(  y\right)  \prod\limits_{j=1}^{n}d^{4}x_{j}\right\} \\
&  \leq C_{1}(n)\left\Vert f\right\Vert _{l}^{2}\int\limits_{\mathbb{Q}%
_{p}^{4n}}\int\limits_{\mathbb{Q}_{p}^{4}}\left\vert \psi_{2}^{\left(
n+1\right)  }\left(  y,x_{1},\cdots,x_{n}\right)  \right\vert ^{2}d^{4}%
y\prod\limits_{j=1}^{n}d^{4}x_{j}\\
&  \leq C_{1}(n)\left\Vert f\right\Vert _{l}^{2}\left\Vert \psi_{2}^{\left(
n+1\right)  }\right\Vert _{\mathcal{H}_{s}^{\left(  n+1\right)  }}^{2}\,,
\end{align*}
for any $l\in\mathbb{N}$. For $I_{1}$, we have%
\[
I_{1}\leq\frac{1}{2n}\left(  n\left\Vert f\right\Vert _{0}\left\Vert \psi
_{2}^{\left(  n-1\right)  }\right\Vert _{\mathcal{H}_{s}^{\left(  n-1\right)
}}\right)  ^{2}=n\left\Vert f\right\Vert _{0}^{2}\left\Vert \psi_{2}^{\left(
n-1\right)  }\right\Vert _{\mathcal{H}_{s}^{\left(  n-1\right)  }}^{2}.
\]
Consequently,
\[
\left\Vert \boldsymbol{\Phi}\left(  f\right)  \psi_{2}\right\Vert \leq
\sqrt{2}\left\Vert f\right\Vert _{l}\left\Vert \psi_{2}\right\Vert \text{ for any
}l\in\mathbb{N}\text{,}%
\]
which implies that \
\[
f\rightarrow\left\langle \psi_{1},\boldsymbol{\Phi}\left(  f\right)  \psi
_{2}\right\rangle \text{ is an element of }\mathcal{H}_{\infty}^{\ast}\left(
\mathbb{C}\right)  \text{,}%
\]
see (\ref{dual_space}).

\noindent\textbf{Poincar\'{e} invariance of the field.} \ The proof is identical to that of
Theorem X.42 in \cite{Reed-SimonII}.

\noindent\textbf{Cyclicity of the vacuum. }The cyclicity of the vacuum for
$\boldsymbol{\Phi}\left(  \cdot\right)  $ follows from Theorem X. 41 (parts
(b) and (d)) in \cite{Reed-SimonII}, by using the fact that the mapping%
\begin{equation}%
\begin{array}
[c]{cccc}%
R: & \mathcal{D}_{\mathbb{C}}(\mathbb{Q}_{p}^{4}) & \rightarrow &
L_{\mathbb{C}}^{2}\left(  V^{+},d\lambda\right) \\
&  &  & \\
& f & \rightarrow & \widehat{f}\mid_{V^{+}}%
\end{array}
\label{Eq30}%
\end{equation}
has a dense range. Indeed, by using that $\mathcal{D}_{\mathbb{C}}(V^{+})$ is
dense in $L_{\mathbb{C}}^{2}\left(  V^{+},d\lambda\right)  $, see Remark
\ref{note_Prop_1}, and the sequence (\ref{Eq_sequence}), we conclude that
$\mathcal{D}_{\mathbb{C}}(\mathbb{Q}_{p}^{4})$ is dense in $L_{\mathbb{C}}%
^{2}\left(  V^{+},d\lambda\right)  $. Finally, (\ref{Eq30}) follows from the
fact that $\mathcal{F}(\mathcal{D}_{\mathbb{C}}(\mathbb{Q}_{p}^{4}%
))=\mathcal{D}_{\mathbb{C}}(\mathbb{Q}_{p}^{4})$.

\noindent\textbf{Local commutativity.} Segal's quantization can be performed on the
field $\boldsymbol{\Phi}(f)$, $f\in\mathcal{H}_{\infty
}\left(  \mathbb{C}\right)  $, see \cite[Theorem X.41]{Reed-SimonII}.
Local commutativity in this context means that
\begin{equation}
\left[  \boldsymbol{\Phi}\left(  f\right)  ,\boldsymbol{\Phi}\left(  g\right)
\right]  \psi=\boldsymbol{\Phi}\left(  f\right)  \boldsymbol{\Phi}\left(
g\right)  \psi-\boldsymbol{\Phi}\left(  g\right)  \boldsymbol{\Phi}\left(
f\right)  \psi=0, \label{CCR}%
\end{equation}
for any $f$, $g\in\mathcal{H}_{\infty}\left(\mathbb{C}\right)$ with support
on an appropriate domain, and for all $\psi\in F_{0}$. Without loss of
generality we may suppose that $f$ and $g$ in \eqref{CCR} are real-valued
since $\boldsymbol{\Phi}$ is linear. Since the range of $R:\mathcal{D}_{\mathbb{C}}\rightarrow
L_{\mathbb{C}}^{2}\left(  V^{+},d\lambda\right)  $ is dense in $L_{\mathbb{C}%
}^{2}\left(  V^{+},d\lambda\right)  $, we may assume that $f$, $g$ belong
to\ $\mathcal{D}_{\mathbb{C}}$, cf. \cite[Theorem X.41-(d)]{Reed-SimonII}. By
using the Segal quantization, cf. \cite[Theorem X.41-(c) ]{Reed-SimonII}, we
have%
\begin{align*}
\left[  \boldsymbol{\Phi}\left(  f\right)  ,\boldsymbol{\Phi}\left(  g\right)
\right]  \psi &  =\sqrt{-1}\operatorname{Im}\left\langle Rf,Rg\right\rangle
_{L_{\mathbb{C}}^{2}\left(  V^{+},d\lambda\right)  }\psi\\
&  =\frac{1}{2}\left\{  \int\limits_{V^{+}}\left\{  \overline{\widehat
{f}\left(  k\right)  }\widehat{g}\left(  k\right)  -\widehat{f}\left(
k\right)  \overline{\widehat{g}\left(  k\right)  }\right\}  d\lambda\left(
k\right)  \right\}  \psi.
\end{align*}
Now, we define%
\begin{equation}
\Delta\left(  x\right)  =\int\limits_{V^{+}}\left\{  \chi_{p}\left(
-\mathcal{B}\left(  x,k\right)  \right)  -\chi_{p}\left(  \mathcal{B}\left(
x,k\right)  \right)  \right\}  d\lambda\left(  k\right)  ,
\label{Eq_formula-delta}%
\end{equation}
which is a well-defined function in $\mathbb{Q}_{p}^{4}$ because $V^{+}$ is
open and compact. Then%
\begin{equation}
\left[  \boldsymbol{\Phi}\left(  f\right)  ,\boldsymbol{\Phi}\left(  g\right)
\right]  \psi=\frac{1}{2}\left\{  \int\limits_{\mathbb{Q}_{p}^{4}}%
\int\limits_{\mathbb{Q}_{p}^{4}}\Delta\left(  x-y\right)  f\left(  x\right)
g\left(  y\right)  d^{4}xd^{4}y\right\}  \psi. \label{CCR2}%
\end{equation}
Therefore, the study of the local commutativity in the $p-$adic quantum field theory
of a scalar field becomes the study of the vanishing of $\Delta(x)$ as 
a distribution on 
$\mathcal{D}_{\mathbb{C}}(\mathbb{Q}_{p}^{4})\times\mathcal{D}_{\mathbb{C}}(\mathbb{Q}_{p}^{4})$. 
It is then enough to observe that $\Delta\left(  x\right)  \equiv0$ if
$x\in\mathbb{Z}_{p}^{4}$, because $\left. \chi_p\right|_{\mathbb{Z}_p}\equiv 1$.

Finally, to prove the second part \eqref{itemii} notice that, since \ $\square_{\mathfrak{q},\alpha}:\mathcal{H}_{\infty}\left(
\mathbb{C}\right)  \rightarrow\mathcal{H}_{\infty}\left(  \mathbb{C}\right)
$, see Lemma \ref{lemma15} , $\boldsymbol{\Phi}$ $\left(  \square
_{\mathfrak{q},\alpha}f\right)  $, $f\in\mathcal{H}_{\infty}\left(
\mathbb{C}\right)  $, is well-defined, and since $\mathcal{H}_{\infty}\left(
\mathbb{C}\right)  \subset L_{\mathbb{C}}^{2}\left(  \mathbb{Q}_{p}^{4}%
,d^{4}k\right)  $, we have $\mathcal{F}\left(  \square_{\mathfrak{q},\alpha
}f\right)  \allowbreak=\left\vert \mathfrak{q}-1\right\vert _{p}^{\alpha
}\mathcal{F}(f)$, so $R(\square_{\mathfrak{q},\alpha}f)=0$, and consequently
$\boldsymbol{\Phi}\left(  \square_{\mathfrak{q},\alpha}f\right)  =0$, for all
$f\in\mathcal{H}_{\infty}\left(  \mathbb{C}\right)  $.
\end{proof}

\subsection{Conjugated fields}

We take $\mathcal{H}=L_{\mathbb{C}}^{2}\left(  V^{+},d\lambda\right)  $ as
before. Recall that $\left(  k_{0},\boldsymbol{k}\right)  \in V^{+}$ if and
only if $\left(  k_{0},-\boldsymbol{k}\right)  \in V^{+}$. By using this fact,
we define%
\[%
\begin{array}
[c]{cccc}%
\boldsymbol{C}: & \mathcal{H} & \rightarrow & \mathcal{H}\\
&  &  & \\
& f\left(  k_{0},\boldsymbol{k}\right)  & \rightarrow & \overline{f\left(
k_{0},-\boldsymbol{k}\right)  }\text{.}%
\end{array}
\]
Then $\boldsymbol{C}$ induces a \textit{conjugation}
on $\mathcal{H}$, i.e. $\boldsymbol{C}$ gives an antilinear isometry
satisfying $\boldsymbol{C}^{2}=I$. We set $\mathcal{H}_{\boldsymbol{C}%
}:=\left\{  f\in\mathcal{H};\boldsymbol{C}f=f\right\}  $.

We recall that $\omega\left(  \boldsymbol{k}\right)  :U_{\mathfrak{q}%
}\rightarrow\mathbb{Q}_{p}$ is \ a non-vanishing analytic \ function. We
define%
\[
\mu (\boldsymbol{k})=
\begin{cases}
\sqrt{ |\omega (\boldsymbol{k})|_p}\ &\mbox{ if }\boldsymbol{k}\in U_{\mathfrak{q}}\,,\\[6pt]
0\ &\mbox{ if } \boldsymbol{k}\in\mathbb{Q}_{p}^{3}\backslash
U_{\mathfrak{q}}.
\end{cases}
\]
Then $\mu\left(  \boldsymbol{k}\right)  \in\mathcal{D}_{\mathbb{R}}%
(\mathbb{Q}_{p}^{3})$.

We now define the canonical fields corresponding to $\boldsymbol{C}$ as
follows:%
\begin{align*}
\boldsymbol{\varphi}\left(  f\right)   &  =\frac{1}{\sqrt{2}}\left\{  \left(
a^{-}\left(  Rf\right)  \right)  ^{\ast}+a^{-}\left(  \boldsymbol{C}Rf\right)
\right\}  \text{, for }f\in\mathcal{H}_{\infty}(\mathbb{C})\text{, and}\\
& \\
\boldsymbol{\pi}\left(  f\right)   &  =\frac{\sqrt{-1}}{\sqrt{2}}\left\{
\left(  a^{-}\left(  \mu Rf\right)  \right)  ^{\ast}-a^{-}\left(
\boldsymbol{C}\mu Rf\right)  \right\}  \text{, for }f\in\mathcal{H}_{\infty
}(\mathbb{C})\text{.}%
\end{align*}
We call $f\rightarrow\boldsymbol{\varphi}\left(  f\right)  $ the
\textit{canonical free field} over $\mathcal{H}_{\boldsymbol{C}}$ of mass $1$,
and $f\rightarrow\boldsymbol{\pi}\left(  f\right)  $ the \textit{canonical
conjugate momentum} over $\mathcal{H}_{\boldsymbol{C}}$ of mass $1$. These
maps are complex linear and $\boldsymbol{\varphi}\left(  f\right)  $,
$\boldsymbol{\pi}\left(  f\right)  $\ are self-adjoint if and only if
$Rf\in\mathcal{H}_{\boldsymbol{C}}$.

The distribution $\delta\left(  x_{0}-t_{0}\right)  g\left(  \boldsymbol{x}%
\right)  $ is defined as the direct product of the distributions
$\delta\left(  x_{0}-t_{0}\right)  $ and $g\left(  \boldsymbol{x}\right)  $:
\[%
\begin{array}
[c]{cccc}%
\delta\left(  x_{0}-t_{0}\right)  \times g\left(  \boldsymbol{x}\right)  : &
\mathcal{D}_{\mathbb{C}}(\mathbb{Q}_{p})\times\mathcal{D}_{\mathbb{C}%
}(\mathbb{Q}_{p}^{3}) & \rightarrow & \mathbb{C}\\
&  &  & \\
& \sum_{i}\phi_{i}\left(  x_{0}\right)  \theta_{i}\left(  \boldsymbol{x}%
\right)  & \rightarrow & \sum_{i}\phi_{i}\left(  t_{0}\right)  \int
_{\mathbb{Q}_{p}^{3}}g\left(  \boldsymbol{x}\right)  \theta_{i}\left(
\boldsymbol{x}\right)  d^{3}\boldsymbol{x},
\end{array}
\]
see e.g. \cite{V-V-Z}. If $g\in L_{\mathbb{C}}^{2}\left(  \mathbb{Q}_{p}%
^{3},d^{3}\boldsymbol{x}\right)  $, then the Fourier transform of the
distribution $\delta\left(  x_{0}-t_{0}\right)  g\left(  \boldsymbol{x}%
\right)  $ is $\chi_{p}\left( k_{0}t_{0}\right)  \widehat
{g}\left(  \boldsymbol{k}\right)  $, where $\widehat{g}\left(  \boldsymbol{k}%
\right)  \in L_{\mathbb{C}}^{2}\left(  \mathbb{Q}_{p}^{3},d^{3}\boldsymbol{k}%
\right)  $ is the $3-$dimensional Fourier transform with respect to the
bilinear form $-\mathfrak{B}_{0}\left(  \boldsymbol{x},\boldsymbol{k}\right)
$. By using Lemma \ref{lemma16A}, we can extend the projection $R$ to the
distributions of the form $\delta\left(  x_{0}-t_{0}\right)  g\left(
\boldsymbol{x}\right)  $, $g\in L_{\mathbb{C}}^{2}\left(  \mathbb{Q}_{p}%
^{3},d^{3}\boldsymbol{x}\right)  $, and thus we extend \ the class of functions
on which $\boldsymbol{\varphi}\left(  \cdot\right)  $ and $\boldsymbol{\pi
}\left(  \cdot\right)  $ are defined to include these distributions.

In the case $t_{0}=0$, with $g$ real-valued, we have%
\[
\left(  \boldsymbol{C}R\widehat{\delta g}\right)  \left(  k_{0},\boldsymbol{k}%
\right)  =\overline{R\widehat{\delta g}\left(  k_{0},-\boldsymbol{k}\right)
}=\overline{R\widehat{g}\left(  k_{0},-\boldsymbol{k}\right)  }=\overline
{\widehat{g} (-\boldsymbol{k}))}  =\widehat{g}\left(
\boldsymbol{k}\right)  =R\left(  \widehat{\delta g}\right)  .
\]
Consequently, $R\left(  \delta g\right)  $ and $\mu R\left(  \delta g\right)
$ are in $\mathcal{H}_{\boldsymbol{C}}$, and $\boldsymbol{\varphi}\left(
\delta g\right)  $, $\boldsymbol{\pi}\left(  \delta g\right)  $\ are
self-adjoint if $g\in L_{\mathbb{C}}^{2}\left(  \mathbb{Q}_{p}^{3}%
,d^{3}\boldsymbol{x}\right)  $ is real. We call the maps $g\rightarrow
\boldsymbol{\varphi}\left(  \delta g\right)  $ and $g\rightarrow
\boldsymbol{\pi}\left(  \delta g\right)  $ \textit{the time-zero fields}.

From now on, we will only use `test functions' of the form $\delta g$ with
$g\in L_{\mathbb{C}}^{2}\left(  \mathbb{Q}_{p}^{3},d^{3}\boldsymbol{x}\right)
$ in $\boldsymbol{\varphi}\left(  \cdot\right)  $ and $\boldsymbol{\pi}\left(
\cdot\right)  $, and write $\boldsymbol{\varphi}\left(  g\right)  $ and
$\boldsymbol{\pi}\left(  g\right)  $ instead of \ $\boldsymbol{\varphi}\left(
\delta g\right)  $ and $\boldsymbol{\pi}\left(  \delta g\right)  $. If $f$ and
$g$ are functions from $L_{\mathbb{R}}^{2}\left(  \mathbb{Q}_{p}^{3}%
,d^{3}\boldsymbol{x}\right)  $, by using Theorem X.43-(c), we have%
\begin{equation}
\left[  \boldsymbol{\varphi}\left(  f\right)  ,\boldsymbol{\pi}\left(
g\right)  \right]  \psi=\sqrt{-1}\left\{  \int\limits_{V^{+}}\overline
{\widehat{f}(k)}\widehat{g}(k)\mu\left(  k\right)  d\lambda(k)\right\}
\psi\text{, \ for all }\psi\in F_{0}. \label{CCR4}%
\end{equation}

\subsection{Transferring fields from $\mathfrak{F}_{s}\left(  L_{\mathbb{C}%
}^{2}\left(  V^{+},d\lambda\right)  \right)  $ to $\mathfrak{F}_{s}\left(
L_{\mathbb{C}}^{2}\left(  U_{\mathfrak{q}},d^{3}\boldsymbol{k}\right)
\right)  $}

We use the notation%
\[
a^{\dagger}\left(  f\right)  =\left(  a^{-}\left(  f\right)  \right)  ^{\ast
}\text{, \ \ \ }a\left(  f\right)  =\left(  a^{-}\left(  \boldsymbol{C}%
f\right)  \right)  .
\]
As we already mentioned, each function $f(\boldsymbol{k})=f\left(
\sqrt{\omega\left(  \boldsymbol{k}\right)  },\boldsymbol{k}\right)
\in\allowbreak L_{\mathbb{C}}^{2}\left(  V^{+},d\lambda\right)  $ is a
function on $U_{\mathfrak{q}}$. We take 
$$
\left(  Jf\right)  \left(
k_{0},\boldsymbol{k}\right)  =\frac{f\left(  \sqrt{\omega\left(
\boldsymbol{k}\right)  },\boldsymbol{k}\right)  }{\left\vert \sqrt
{\omega\left(  \boldsymbol{k}\right)  }\right\vert _{p}^{\frac{1}{2}}}
$$
as before. Then $J$ is a unitary isometry of $L_{\mathbb{C}}^{2}\left(
V^{+},d\lambda\right)  $ onto $L_{\mathbb{C}}^{2}\left(  U_{\mathfrak{q}%
},d^{3}\boldsymbol{k}\right)  $. The annihilation and creation operators on
$\mathfrak{F}_{s}\left(  L_{\mathbb{C}}^{2}\left(  U_{\mathfrak{q}}%
,d^{3}\boldsymbol{k}\right)  \right)  $, $\widetilde{a}\left(  \cdot\right)
$, $\widetilde{a}^{\dagger}\left(  \cdot\right)  $ are related to $a\left(
\cdot\right)  $\ and $a^{\dagger}\left(  \cdot\right)  $\ by the formulas:%
\begin{align*}
\widetilde{a}\left( Jf \right)   &  =\Gamma\left(  J\right)  a\left(  f\right)  \Gamma\left(
J\right)  ^{-1}\text{,}\\
\widetilde{a}^{\dagger}\left( Jf \right)   &  =\Gamma\left(  J\right)  a^{\dagger}\left(
f\right)  \Gamma\left(  J\right)  ^{-1}\text{.}%
\end{align*}
By using the unitary map $\Gamma\left(  J\right)  $, we carry the quantum
fields over $\mathfrak{F}_{s}\left(  L_{\mathbb{C}}^{2}\left(  U_{\mathfrak{q}%
},d^{3}\boldsymbol{k}\right)  \right)  $ as follows:%
\[
\widetilde{\boldsymbol{\Phi}}\left(  f\right)  =\Gamma\left(  J\right)
\boldsymbol{\Phi}\left(  f\right)  \Gamma\left(  J\right)  ^{-1}=\frac
{1}{\sqrt{2}}\left\{  \widetilde{a}\left(  \widetilde{\boldsymbol{C}}\frac
{Rf}{\left\vert \sqrt{\omega\left(  \boldsymbol{k}\right)  }\right\vert
_{p}^{\frac{1}{2}}}\right)  +\widetilde{a}^{\dagger}\left(  \frac
{Rf}{\left\vert \sqrt{\omega\left(  \boldsymbol{k}\right)  }\right\vert
_{p}^{\frac{1}{2}}}\right)  \right\}
\]
for $f\in\mathcal{H}_{\mathbb{\infty}}(\mathbb{R})$, and
\[
\widetilde{\boldsymbol{\varphi}}\left(  f\right)  =\Gamma\left(  J\right)
\boldsymbol{\varphi}\left(  f\right)  \Gamma\left(  J\right)  ^{-1}=\frac
{1}{\sqrt{2}}\left\{  \widetilde{a}\left(  \frac{R\left(  f\delta\right)
}{\left\vert \sqrt{\omega\left(  \boldsymbol{k}\right)  }\right\vert
_{p}^{\frac{1}{2}}}\right)  +\widetilde{a}^{\dagger}\left(  \frac{R\left(
f\delta\right)  }{\left\vert \sqrt{\omega\left(  \boldsymbol{k}\right)
}\right\vert _{p}^{\frac{1}{2}}}\right)  \right\}
\]
for $f\in L_{\mathbb{C}}^{2}(\mathbb{Q}_{p}^{3},d^{3}\boldsymbol{x})$, where
$\widetilde{\boldsymbol{C}}=\Gamma\left(  J\right)  \boldsymbol{C}%
\Gamma\left(  J\right)  ^{-1}$ acts by $\left(  \widetilde{\boldsymbol{C}%
}g\right)  \left(  \boldsymbol{k}\right)  =\overline{g(-\boldsymbol{k})}$.

We drop the tilde $\ \widetilde{\cdot}$, and from now on, we work with
fields on $\mathfrak{F}_{s}\left(  L_{\mathbb{C}}^{2}\left(  U_{\mathfrak{q}%
},d^{3}\boldsymbol{k}\right)  \right)  $, for $f$, $g$ real-valued. Then,
formula \ (\ref{CCR4}) becomes%
\[
\left[  \boldsymbol{\varphi}\left(  f\right)  ,\boldsymbol{\pi}\left(
f\right)  \right]  =\sqrt{-1}\int_{U_{\mathfrak{q}}}f(\boldsymbol{x}%
)g(\boldsymbol{x})d^{3}\boldsymbol{x},
\]
which is the canonical commutation relation in $L_{\mathbb{C}}^{2}%
(U_{\mathfrak{q}},d^{3}\boldsymbol{x})$.

\subsection{\label{Section 5.5}Some classical calculations}

In this section, we discuss in a $p-$adic frame the annihilation and creation
operators introduced above, to show that they conform to the common usage in
the Physics literature. We start by defining
\[
D_{0}=\left\{  \psi;\psi\in F_{0},\psi^{\left(  n\right)  }\in\mathcal{D}%
_{\mathbb{C}}(U_{\mathfrak{q}}^{3n})\text{ for all }n\right\}
\]
and for each $l\in\mathbb{Q}_{p}^{3}$ (we do not use bold letters for
$3$-dimensional vectors) an operator $a\left(  l\right)  $ on $\mathfrak{F}%
_{s}\left(  L_{\mathbb{C}}^{2}\left(  U_{\mathfrak{q}},d^{3}x\right)  \right)
=\oplus_{n=0}^{\infty}\mathcal{H}_{s}^{\left(  n\right)  }$ with domain
$D_{0}$ by%
\[
\left(  a\left(  l\right)  \psi\right)  ^{\left(  n\right)  }\left(
k_{1},\ldots,k_{n}\right)  =\sqrt{n+1}\psi^{\left(  n+1\right)  }\left(
l,k_{1},\ldots,k_{n}\right)  ,\text{ \ }n\geq 0\, .
\]
The formal adjoint of $a\left(  l\right)$ is given by%

\[
\left(  a\left(  l\right)  ^{\dagger}\psi\right)  ^{\left(  n\right)  }\left(
k_{1},\ldots,k_{n}\right)  =\frac{1}{\sqrt{n}}\sum\limits_{j=1}^{n}%
\delta\left(  l-k_{j}\right)  \psi^{\left(  n-1\right)  }\left(  k_{1}%
,\ldots,\widetilde{k}_{j},\ldots k_{n}\right)  \text{, }%
\]
for $n\geq1$, and by definition $\left(  a\left(  l\right)  ^{\dagger}%
\psi\right)  ^{\left(  n\right)  }\left(  k_{1},\ldots,k_{n}\right)  =0$ for
$n=0$. This operator is a well-defined quadratic form on $D_{0}\times D_{0}$:
if $\psi_{2}=\left\{  \psi_{2}^{\left(  n\right)  }\right\}  _{n\in\mathbb{N}%
}$, $\psi_{1}=\left\{  \psi_{1}^{\left(  n\right)  }\right\}  _{n\in
\mathbb{N}}\allowbreak\in F_{0}$, then the quadratic form
\begin{multline*}
\left\langle \psi_{2},a\left(  l\right)  ^{\dagger}\psi_{1}\right\rangle
=\sum\limits_{n=1}^{\infty}\left\langle \psi_{2}^{\left(  n\right)  },\left(
a\left(  l\right)  ^{\dagger}\psi_{1}\right)  ^{\left(  n\right)
}\right\rangle _{\mathcal{H}_{s}^{\left(  n\right)  }}=\\
\sum\limits_{n=1}^{\infty}\frac{1}{\sqrt{n}}\sum\limits_{j=1}^{n}%
\int\limits_{U_{\mathfrak{q}}^{n-1}}\overline{\psi_{2}^{\left(  n\right)
}\left(  k_{1},\ldots,k_{j-1},l,k_{j+1},\ldots,k_{n}\right)  }\times\\
\psi_{1}^{\left(  n-1\right)  }\left(  k_{1},\ldots,k_{j-1},k_{j+1}%
,\ldots,k_{n}\right)  \prod\limits_{\substack{i=1 \\i\neq j}}^{n}d^{3}k_{i}%
\end{multline*}
is well-defined.
The formulas%
\begin{equation}
a\left(  g\right)  =\int_{U_{\mathfrak{q}}}a\left(  k\right)  g\left(
-k\right)  d^{3}k\text{ \ and \ }a^{\dagger}\left(  g\right)  =\int
_{U_{\mathfrak{q}}}a^{\dagger}\left(  k\right)  g\left(  k\right)
d^{3}k\text{,} \label{operator_valued_distribution}%
\end{equation}
hold for all $g\left(  k\right)  \in\mathcal{D}_{\mathbb{C}}(U_{\mathfrak{q}%
})$, if the equalities are understood in the sense of quadratic forms, i.e.%
\[
\left\langle \psi_{2},a\left(  g\right)  \psi_{1}\right\rangle :=\int
_{U_{\mathfrak{q}}}\left\langle \psi_{2},a\left(  k\right)  \psi
_{1}\right\rangle g\left(  -k\right)  d^{3}k\text{ }%
\]
and
\[
\left\langle \psi_{2},a\left(  g\right)  \psi_{1}\right\rangle :=\int
_{U_{\mathfrak{q}}}\left\langle \psi_{2},a^{\dagger}\left(  k\right)  \psi
_{1}\right\rangle g\left(  k\right)  d^{3}k\text{.}%
\]
On the other hand, since $a\left(  l\right)  :D_{0}\rightarrow D_{0}$, the
powers of $a\left(  l\right)  $ are well-defined on $D_{0}$. Then%
\[
\left\langle  \psi_{1},\left(  a\left(  l\right)  ^{\dagger}\right)  ^{n}\psi
_{2}\right\rangle  =\left\langle  \left(  a\left(  l\right)  \right)  ^{n}\psi_{1}%
,\psi_{2}\right\rangle\,,
\]
for each $n$, where the equality is to be understood in the sense of quadratic forms, and
\[
\left\langle  \psi_{1},\left(  \prod\limits_{i=N_{1}+1}^{N_{2}}a^{\dagger}\left(
l_{i}\right)  \right)  \left(  \prod\limits_{i=1}^{N_{1}}a\left(
l_{i}\right)  \right)  \psi_{2}\right\rangle
\]
is a well-defined quadratic form on $D_{0}\times D_{0}$. In addition, if
$f_{i}\in\mathcal{D}_{\mathbb{C}}(U_{\mathfrak{q}})$, then the following
expressions are well-defined as quadratic forms: The product
\begin{gather*}
\left(  \prod\limits_{i=N_{1}+1}^{N_{2}}a^{\dagger}\left(  f_{i}\right)
\right)  \left(  \prod\limits_{i=1}^{N_{1}}a\left(  f_{i}\right)  \right)  =\\
\int\limits_{U_{\mathfrak{q}}^{3N_{2}}}\left(  \prod\limits_{i=N_{1}+1}^{N_{2}%
}a^{\dagger}\left(  k_{i}\right)  \right)  \left(  \prod\limits_{i=1}^{N_{1}%
}a\left(  -k_{i}\right)  \right)  \left(  \prod\limits_{i=1}^{N_{2}}%
f_{i}\left(  k_{i}\right)  \right)  d^{3}k_{1}\cdots d^{3}k_{N_{2}}\,,
\end{gather*}
the number operator
\[
N=\int\limits_{U_{\mathfrak{q}}}a^{\dagger}\left(  k\right)  a\left(
k\right)  d^{3}k\,,
\]
and the free Hamiltonian of unit mass,
\[
H_{0}=\int\limits_{U_{\mathfrak{q}}}\mu\left(  k\right)  a^{\dagger}\left(
k\right)  a\left(  k\right)  d^3k\,.
\]

Finally, by using quadratic forms on $D_{0}$ we can express the free scalar
field and the time zero fields in terms of $a^{\dagger}\left(  k\right)  $ and
$a\left(  k\right)  $ (i.e. by using (\ref{operator_valued_distribution}) with
$g$ real-valued):%
\begin{multline*}
\boldsymbol{\Phi}\left(  t,x\right)  =\\
\frac{1}{\sqrt{2}}\int\limits_{U_{\mathfrak{q}}}\left\{  \chi_{p}\left(
\sqrt{\omega\left(  k\right)  }t-\mathfrak{B}_{0}\left(  k,x\right)  \right)
a^{\dagger}\left(  k\right)  +\chi_{p}\left(  -\sqrt{\omega\left(  k\right)
}t+\mathfrak{B}_{0}\left(  k,x\right)  \right)  a\left(  k\right)  \right\} \\
\times\frac{d^{3}k}{\left\vert \sqrt{\omega\left(  k\right)  }\right\vert
_{p}^{\frac{1}{2}}},
\end{multline*}%
\[
\boldsymbol{\varphi}\left(  x\right)  =\frac{1}{\sqrt{2}}\int
\limits_{U_{\mathfrak{q}}}\left\{  \chi_{p}\left(  -\mathfrak{B}_{0}\left(
k,x\right)  \right)  a^{\dagger}\left(  k\right)  +\chi_{p}\left(
\mathfrak{B}_{0}\left(  k,x\right)  \right)  a\left(  k\right)  \right\}
\frac{d^{3}k}{\left\vert \sqrt{\omega\left(  k\right)  }\right\vert
_{p}^{\frac{1}{2}}},
\]%
\[
\boldsymbol{\pi}\left(  x\right)  =\frac{\sqrt{-1}}{\sqrt{2}}\int
\limits_{U_{\mathfrak{q}}}\left\{  \chi_{p}\left(  -\mathfrak{B}_{0}\left(
k,x\right)  \right)  a^{\dagger}\left(  k\right)  -\chi_{p}\left(
\mathfrak{B}_{0}\left(  k,x\right)  \right)  a\left(  k\right)  \right\}
\left\vert \sqrt{\omega\left(  k\right)  }\right\vert _{p}^{\frac{1}{2}}%
d^{3}k.
\]

\subsection{A $p$-adic Klein-Gordon equation}

In this section, we consider the inhomogeneous $p-$adic Klein-Gordon equation:
\begin{equation}
\square_{\mathfrak{q},\alpha}u\left( t,\mathbf{x}\right)  =h\left( t,\mathbf{x}\right)  ,
\label{Eq_Klein_Gordon}%
\end{equation}
where $\left( t,\mathbf{x}\right)  \in\mathbb{Q}_{p}\times\mathbb{Q}^{3}_{p}$ and
$h\left( t,\mathbf{x}\right)  \in\mathcal{D}_{\mathbb{C}}(\mathbb{Q}_{p}%
\times\mathbb{Q}^{3}_{p})$. We use the techniques and results
of \cite[Chapter 6]{Zuniga-LNM-2016}. By a solution (or weak solution) we
understand a distribution from $\mathcal{D}_{\mathbb{C}}^{\prime}%
(\mathbb{Q}_{p}\times\mathbb{Q}^{3}_{p})$ satisfying (\ref{Eq_Klein_Gordon}).
We denote by $E_{\mathfrak{q}}^{0}\left( t,\mathbf{x}\right)  $, the fundamental
solution of (\ref{Eq_Klein_Gordon}) obtained in Theorem \ref{Theorem1}.

\begin{theorem}
\label{Theorem3} The following hold true:
\begin{enumerate}[(i)]
\item The equation
\begin{equation}
\square_{\mathfrak{q},\alpha}u\left( t,\mathbf{x}\right)  =0 \label{Eq_Klein-Gordon_1}%
\end{equation}
admits plane waves, this means that if $\left(  E^{\pm},\boldsymbol{\kappa}\right)
\in V^{\pm}$, that is, they form a fixed pair of solutions to $E^{\pm}=\pm\sqrt{\omega\left(  \boldsymbol{\kappa}\right)  }$,
then $\chi_{p}\left\{  -\mathcal{B}\left(  \left(  t,\mathbf{x}\right)  ,\left(
E^{\pm},\boldsymbol{\kappa}\right)  \right)  \right\}  $ is a weak solution of
(\ref{Eq_Klein-Gordon_1}).

\item The distributions%
\begin{gather*}
\int\limits_{U_{\mathfrak{q}}}\chi_{p}\left\{  -\mathcal{B}\left(  \left(
t,\mathbf{x}\right)  ,\left(  \sqrt{\omega\left(  \mathbf{k}\right)  },\mathbf{k}\right)  \right)
\right\}  \frac{d^{3}\mathbf{k}}{\left\vert \sqrt{\omega\left(  \mathbf{k}\right)  }\right\vert
_{p}}+\\
\int\limits_{U_{\mathfrak{q}}}\chi_{p}\left\{  \mathcal{B}\left(  \left(
t,\mathbf{x}\right)  ,\left(  -\sqrt{\omega\left(  \mathbf{k}\right)  },\mathbf{k}\right)  \right)
\right\}  \frac{d^{3}\mathbf{k}}{\left\vert \sqrt{\omega\left(  \mathbf{k}\right)  }\right\vert
_{p}}%
\end{gather*}
are the unique weak solutions of \eqref{Eq_Klein-Gordon_1} (up to the multiplication
by a non-zero complex constant) which are invariant
under $\mathcal{L}_{+}^{\uparrow}$.

\item The distributions%
\begin{multline*}
u(t,\mathbf{x};A,B,C)=E_{\mathfrak{q}}^{0}\left( t,\mathbf{x}\right)  \ast h\left( t,\mathbf{x}\right)
+\\
C\int\limits_{U_{\mathfrak{q}}}\left\{  \chi_{p}\left(  -\sqrt{\omega\left(
\mathbf{k}\right)  }t+\mathfrak{B}_{0}\left(  \mathbf{k},\mathbf{x}\right)  \right)  A\left(  \mathbf{k}\right)
+\chi_{p}\left(  \sqrt{\omega\left( \mathbf{k}\right)  }t+\mathfrak{B}_{0}\left(
\mathbf{k},\mathbf{x}\right)  \right)  B\left( \mathbf{k}\right)  \right\} \\
\times\frac{d^{3}\mathbf{k}}{\left\vert \sqrt{\omega\left( \mathbf{k}\right)  }\right\vert
_{p}},
\end{multline*}
where $C$ is a non-zero complex number, and $A\left( \mathbf{k}\right)  $, $B\left(
\mathbf{k}\right)  \in\mathcal{D}_{\mathbb{C}}(\mathbb{Q}_{p}^{3})$, are \ weak
solutions of (\ref{Eq_Klein_Gordon}).
\end{enumerate}
\end{theorem}

\begin{proof}\mbox{}
\begin{enumerate}[(i)]
\item Since $\mathcal{F}_{\substack{k_{0}\rightarrow t\\k\rightarrow x}%
}^{-1}(\delta\left(  k_{0}-E^{\pm},\mathbf{k}-\boldsymbol{\kappa}\right)  )=\chi_{p}\left\{
-\mathcal{B}\left(  \left(  E^{\pm},\boldsymbol{\kappa}\right)  ,\left(  t,\mathbf{x}\right)
\right)  \right\}  $, the condition $E^{\pm}=\pm\sqrt{\omega\left(
\boldsymbol{\kappa}\right)  }$ implies that $k_{0}^{\pm}=\pm\sqrt{\omega\left(
\mathbf{k}\right)  }$, so $\delta\left(  k_{0}-E^{\pm},\mathbf{k}-\boldsymbol{\kappa}\right)  $ is supported on $V^{\pm}\subset V$. The result follows from the
fact that the weak solutions of (\ref{Eq_Klein-Gordon_1}) are exactly the
distributions from $\mathcal{D}_{\mathbb{C}}^{\prime}(\mathbb{Q}_{p}%
\times\mathbb{Q}^{3}_{p})$ whose Fourier transform is supported on $V$, see
\cite[Lemma 169]{Zuniga-LNM-2016}.

\item The distributions of the form $C\delta_{V}$, for $C\in\mathbb{C}^{\times}$, 
are the unique solutions of (\ref{Eq_Klein-Gordon_1}) which
are invariant under $\boldsymbol{O}(\mathfrak{q})$, see \cite[Lemma
169]{Zuniga-LNM-2016} and \cite[Proposition 2-2.]{Rallis-Schiffman}. By
writing $C\delta_{V}=C\delta_{V^{+}}+C\delta_{V^{-}}$ in $\mathcal{D}%
_{\mathbb{C}}^{\prime}(\mathbb{Q}_{p}\times\mathbb{Q}^{3}_{p})$ and using the
fact that $\delta_{V^{\pm}}$ are invariant under $\mathcal{L}_{+}^{\uparrow
}=\left\{  \Lambda\in\boldsymbol{O}(\mathfrak{q});\Lambda\left(  V^{\pm
}\right)  =V^{\pm}\right\}  $, see \cite[Lemma 163]{Zuniga-LNM-2016}, we
conclude that $C\delta_{V^{+}}+C\delta_{V^{-}}$are the unique weak solutions
of (\ref{Eq_Klein-Gordon_1}) which are invariant under $\mathcal{L}%
_{+}^{\uparrow}$. The announced formula follows by computing the inverse
Fourier transform of $\delta_{V^{\pm}}$.

\item The result follows from the second part by using Theorem \ref{Theorem1}.
\end{enumerate}
\end{proof}

\begin{remark}
Notice that $\left\vert \sqrt{\omega\left( \mathbf{k}\right)  }\right\vert
_{p}A\left( \mathbf{k}\right)  $, $\left\vert \sqrt{\omega\left( \mathbf{k}\right)
}\right\vert _{p}B\left( \mathbf{k}\right)$, are test functions, and also
\begin{multline*}
\int\limits_{U_{\mathfrak{q}}}\chi_{p}\left(  \sqrt{\omega\left( \mathbf{k}\right)
}t+\mathfrak{B}_{0}\left( \mathbf{k},\mathbf{x}\right)\right) B\left( \mathbf{k}\right)  \frac
{d^{3}\mathbf{k}}{\left\vert \sqrt{\omega\left( \mathbf{k}\right)  }\right\vert _{p}}\\
=\int\limits_{U_{\mathfrak{q}}}\chi_{p}\left(  \sqrt{\omega\left( \mathbf{k}\right)
}t-\mathfrak{B}_{0}\left( \mathbf{k},\mathbf{x}\right)  \right)  B\left( -\mathbf{k}\right)
\frac{d^{3}\mathbf{k}}{\left\vert \sqrt{\omega\left( \mathbf{k}\right)  }\right\vert _{p}},
\end{multline*}
so the unique weak solution of $\square_{\mathfrak{q},\alpha
}u\left( t,\mathbf{x}\right)  =0$ (with $C=1/\sqrt{2}$) invariant under
$\mathcal{L}_{+}^{\uparrow}$ corresponds to the free scalar field
$\boldsymbol{\Phi}\left( t,\mathbf{x}\right)  $, with $a\left( \mathbf{k}\right)  =\left\vert
\sqrt{\omega\left( \mathbf{k}\right)  }\right\vert _{p}A\left( \mathbf{k}\right)  $,
$a^{\dagger}\left( \mathbf{k}\right)  =\left\vert \sqrt{\omega\left( \mathbf{k}\right)
}\right\vert _{p}B\left( \mathbf{k}\right)$. As we have seen, these solutions can be quantized
using the machinery of the second quantization in such a way that Wightman
axioms are satisfied.
\end{remark}

\begin{acknowledgement}
The first author was supported by Conacyt, the Fundaci\'{o}n Sof\'{\i}%
a Kovalevskaia and the Sociedad Matem\'{a}tica Mexicana. The third author
was partially supported by Conacyt Grant No. 250845. The authors thank
Sergii Torba for many useful discussions.
\end{acknowledgement}

\end{document}